\documentclass[ecta,nameyear,final]{econsocart1}

\usepackage[colorlinks,citecolor=blue,linkcolor=blue,urlcolor=blue]{hyperref}

\usepackage{amsmath}
\usepackage{amssymb}
\usepackage{comment}
\usepackage{enumitem}
\usepackage{graphicx}
\usepackage{natbib}

\graphicspath{{Figure/}}

\startlocaldefs

\numberwithin{equation}{section}

\theoremstyle{plain}

\newtheorem{theorem}{Theorem}

\newtheorem{lemma}[theorem]{Lemma}

\theoremstyle{definition}

\newtheorem{remark}{Remark}
\newtheorem{proposition}{Proposition}[section]
\newtheorem{assumption}{Assumption}[section]

\newcommand{\argmax}{\mathrm{argmax}}
\newcommand{\Var}{\operatorname{Var}}

\newcommand{\plim}{\operatorname{plim}}

\newcommand{\Op}{O_p}
\newcommand{\R}{\mathbb R}
\newcommand{\cT}{\mathcal T}
\newcommand{\mydate}{June, 2026}

\endlocaldefs

\allowdisplaybreaks[2]

\begin{document}
	
\begin{frontmatter}

\title{Causal Inference Using Factor Models}
\runtitle{Causal Factor Models}

\begin{aug}
\author[add1]{\fnms{Jushan}~\snm{Bai}\ead[label=e1]{jb3064@columbia.edu}}
\author[add2]{\fnms{Peng}~\snm{Wang}\ead[label=e2]{pwang@ust.hk}}
\address[add1]{%
	\orgdiv{Department of Economics},
	\orgname{Columbia University}}

\address[add2]{%
	\orgdiv{Department of Economics},
	\orgname{HKUST}}
\end{aug}	

\begin{funding}
	Wang gratefully acknowledges financial support from the Hong Kong Research Grants Council through GRF 16506024.
\end{funding}

\mydate

\begin{abstract}

We develop a factor-model framework for causal inference in panels with
policy interventions. Treatment effects are represented as structural changes in treated units' exposure
to latent common shocks and, in extensions, changes in the factor process itself.
The approach does not impose the standard
parallel-trends restriction, accommodates one or many treated units, and
targets systematic effects when unit-time idiosyncratic effects are not
point identified. We provide estimation and inference under both fixed
and treatment-dependent factor processes. Simulations show coverage close to nominal levels. In applications to
California tobacco control and German reunification, the method produces
estimates broadly consistent with synthetic control while delivering
formal confidence intervals.

\end{abstract}


\end{frontmatter}


\section{Introduction}\label{sec:intro}

Causal inference methods for policy evaluation, particularly synthetic control
and difference-in-differences, have become widely used in empirical work. The
synthetic control method, introduced by \citet{Abadie_Gardeazabal_2003} and
developed further by \citet{Abadie_Diamond_Hainmeuller_2010,Abadie_Diamond_Hainmeuller_2015}, constructs a weighted
combination of control units that matches the treated unit's pre-treatment
characteristics and outcomes, and uses this synthetic counterpart to form a
post-intervention counterfactual. Similarly, the difference-in-differences
approach, exemplified by \citet{Card_Krueger_1994}, compares the change in
outcomes for a treated group before and after an intervention to the
corresponding change for a control group over the same period.

Both methodologies rely on restrictions that make the untreated
potential outcomes for treated units recoverable from the behavior of
untreated or control units. In difference-in-differences, this
restriction is commonly expressed as parallel trends; in synthetic
control, it is based on pre-treatment fit and the stability of the
relationship between treated and control outcomes. However, real-world
data often exhibit heterogeneous trends, challenging the plausibility
of these restrictions. \citet{Gobillon_Magnac_2016} and
\citet{Xu_2017} address heterogeneous trends through panel data models
with interactive fixed effects, allowing for unit-specific exposure to
common shocks.

In this paper, we employ a factor model to decompose potential outcomes into
systematic and idiosyncratic components. A policy intervention can affect both
components. The individual causal effect can therefore be written as the sum
of a systematic and an idiosyncratic causal component. We emphasize that,
while the idiosyncratic component is not point identified for a fixed unit and
date, the systematic component is identified.
Moreover, we interpret the impact of the policy intervention as a
structural-break problem in the factor representation, which can be rigorously studied using established
structural-break theory. The structural-break perspective also clarifies
which features of the data are required for identification under different
empirical environments.

The paper makes three contributions. First, it separates the realized unit-time treatment effect into a systematic
component, generated by changes in factor loadings, slope coefficients, and
possibly factors, and an idiosyncratic component that is not point identified
for a fixed unit and date.
Second,
it provides estimation and inference for the systematic component in
settings with one treated unit or many treated units, making the framework
applicable to empirical settings in which synthetic-control methods are
commonly used.
Third, it extends the framework to policies that may alter the factor
process itself, showing when such effects are identifiable and when additional
restrictions, such as constant or affine factor shifts, are needed.

These contributions
are developed across a number of empirically relevant environments. First, the policy
intervention may leave the common-shock process unchanged but alter treated
units' exposure to these shocks, so that treatment effects arise through
structural breaks in factor loadings and possibly slope coefficients. This
benchmark is appropriate for targeted interventions whose general-equilibrium
feedback to aggregate trends is plausibly negligible. It remains feasible even
when the number of treated units is very small, including a single treated
unit, because the common factors can be learned from the control group and
treated parameters are identified from time-series variation. Second, the
intervention may also alter the common-shock process itself; in this case, a
flexible post-treatment factor process can be identified when the treated
cross section is large, which permits direct estimation of post-treatment
treated factors. Third, when the intervention may affect both factors and
loadings but the treated group is small, a fully flexible post-treatment factor
process is not separately identifiable; we therefore impose a restricted
post-treatment factor perturbation, such as an affine shift, which preserves
tractable identification and inference while still allowing the policy to
change the common-shock process.

A main feature of our approach is a \emph{dual} modeling strategy: 
we model both the counterfactual potential outcome $Y_{it}(0)$ and the observed outcome
$Y_{it}(1)$ within the same factor structure. In many panel settings,
researchers observe only a subset of the relevant determinants of outcomes, and
idiosyncratic variation can be large relative to the predictable component.
Consequently, if one models only $Y_{it}(0)$, the gap between the realized
outcome $Y_{it}(1)$ and the model-implied counterfactual $\hat Y_{it}(0)$ may
be dominated by residual noise and unmeasured determinants, making the
difference $Y_{it}(1)-\hat Y_{it}(0)$ an unstable object for inference.  
By modeling both potential outcomes, the analysis targets the change
in the systematic component directly, instead of treating unit-time idiosyncratic noise as part of the causal signal.

Our framework relates to, but differs from, existing approaches to panel-data
program evaluation. Difference-in-differences and synthetic-control methods
typically focus on constructing the untreated counterfactual outcome for treated
units. Interactive fixed-effect methods enrich this counterfactual by treating
latent factors as controls for unobserved confounding, analogous to observed
covariates. In these approaches, the factor component is primarily a control
for the untreated outcome process.

In contrast, our framework models both potential outcomes within a factor
structure and allows the treated units' factor loadings to change after the
intervention. Thus the factor component is not merely a control for unobserved
heterogeneity: changes in exposure to common shocks are themselves part of the
systematic causal effect. This shift leads to a different target, the
systematic component of the unit-time treatment effect, rather than the
difference between a realized treated outcome and an imputed untreated outcome.
Section \ref{sec:model-setup} discusses these connections in more detail.

The remainder of the paper is organized as follows. Section \ref{sec:model-setup} introduces the
model and conceptual framework and compares our approach to existing methods.
Section \ref{sec:Identification} studies identification of systematic causal effects. Section \ref{sec:potential-factor} extends the framework
to allow for potential factors, in which the intervention may also alter the
common-shock process. Section \ref{sec:Estimation-and-Inference} develops estimation and inference procedures.
Section \ref{sec:mcmc} reports Monte Carlo evidence. Section \ref{sec:applications} presents empirical
applications to the two datasets analyzed by \citet{Abadie_Diamond_Hainmeuller_2010,Abadie_Diamond_Hainmeuller_2015}.
The last section concludes. Additional technical material is collected in the Appendix, and proofs are
provided in the Online Appendix.

\section{Modeling the Potential Outcomes}\label{sec:model-setup}

We consider a panel data set that includes an outcome variable and
some covariates. The observed outcome variable $Y$ is indexed by
unit and time, i.e., $Y_{it}$, $i=1,2,\ldots,n$, $t=1,2,\ldots,T$.
The potential outcome for unit $i$ in period $t$ is denoted by $Y_{it}\left(d\right)$,
$d=0,1$, with $d=1$ referring to the case of treatment and $d=0$
for the case of no treatment. Assume that a policy intervention occurs
in period $T_{0}+1$ with $1<T_{0}<T$. Let $D_{it}$ denote the observed
treatment dummy. The policy intervention only applies to units $i\leq n_{0}$
without directly affecting units $i>n_{0}$. To focus on the main
idea, we assume that the policy intervention occurs in the same period
for all treated units. The treatment status can be summarized by
\begin{equation}
D_{it}=\begin{cases}
0, & i>n_{0}\;\&\;1\leq t\leq T,\;\text{ or }\;i\leq n_{0}\;\&\;t\leq T_{0},\\
1, & i\leq n_{0}\;\&\;t>T_{0}.
\end{cases}\label{eq:Dit}
\end{equation}
More compactly, $D_{it}=1\left\{ i\leq n_{0},t>T_{0}\right\} $,
where $1\left\{ \cdot\right\} $ is the indicator function.

Our theory will be developed under the conditions where $n-n_{0}\rightarrow\infty$,
$T_{0}\rightarrow\infty$, and $T-T_{0}\rightarrow\infty$. Within
this framework, the number of treated units, $n_{0}$, may either
be a fixed constant or approach infinity. Importantly, we allow for
the case of a single treated unit ($n_{0}=1$), which is directly
related to the synthetic control literature.

We use a dual modeling strategy that characterizes both potential
outcomes through a factor model:
\begin{equation}
Y_{it}\left(d\right)=\lambda_{i}\left(d\right)^{\prime}f_{t}+X_{it}^{\prime}\beta\left(d\right)+\varepsilon_{it}\left(d\right),\;d=0,1,\label{eq:mainmodel}
\end{equation}
where $X_{it}$ denotes the vector of observed covariates and $f_{t}$
is an $r\times1$ vector of unobserved factors. For a given treatment
status $d$, $\lambda_{i}\left(d\right)$ represents the $r\times1$
vector of unobservable factor loadings, $\beta\left(d\right)$ is the vector of
coefficients, and $\varepsilon_{it}\left(d\right)$ denotes the idiosyncratic
error, where the factor loadings, slope coefficients, and idiosyncratic errors are
indexed by treatment status  $d$.

We interpret the factors $f_t$ as latent common shocks or trends generated
outside individual units, such as aggregate macroeconomic conditions,
industry-wide demand shifts, or other pervasive forces that move many units
simultaneously. The policy intervention is assumed not to create these common
forces; rather, it changes how treated units respond to them. In the factor
representation, this corresponds to structural breaks in the treated units'
factor loadings (and possibly slope coefficients), while the factor process
itself remains common across units. This benchmark is especially appropriate
when the treated group is small relative to the population or when the
intervention is plausibly partial-equilibrium, so that general-equilibrium
feedback to aggregate trends is negligible. We relax this restriction in the
potential-factor analysis, allowing the intervention to alter the factor
process when such feedback is plausible.

For treated units, the observed outcome is
\begin{align*}
Y_{it} & =\begin{cases}
\lambda_{i}\left(0\right)^{\prime}f_{t}+X_{it}^{\prime}\beta\left(0\right)+\varepsilon_{it}\left(0\right)=Y_{it}\left(0\right), & t\leq T_{0},\\
\lambda_{i}\left(1\right)^{\prime}f_{t}+X_{it}^{\prime}\beta\left(1\right)+\varepsilon_{it}\left(1\right)=Y_{it}\left(1\right), & t>T_{0},
\end{cases},\;i=1,\ldots,n_{0}.
\end{align*}
For untreated units, the observed outcome is
\[
Y_{it}=\lambda_{i}\left(0\right)^{\prime}f_{t}+X_{it}^{\prime}\beta\left(0\right)+\varepsilon_{it}\left(0\right)=Y_{it}\left(0\right),\;i=n_{0}+1,\ldots,n,\;t=1,\ldots,T.
\]

The causal effect for the treated units is defined as:
\[
\tau_{it}=Y_{it}\left(1\right)-Y_{it}\left(0\right)=Y_{it}-Y_{it}\left(0\right),\;t>T_{0},\;i=1,\ldots,n_{0}.
\]
Using the factor model, we can rewrite $\tau_{it}$ for $i\le n_{0}$
and $t>T_{0}$ as the sum of a systematic causal effect and an
idiosyncratic causal effect:
\begin{align}
\tau_{it} & =\left\{ \lambda_{i}\left(1\right)^{\prime}f_{t}+X_{it}^{\prime}\beta\left(1\right)+\varepsilon_{it}\left(1\right)\right\} -\left\{ \lambda_{i}\left(0\right)^{\prime}f_{t}+X_{it}^{\prime}\beta\left(0\right)+\varepsilon_{it}\left(0\right)\right\} \nonumber \\
 & =\underbrace{\left[\lambda_{i}\left(1\right)-\lambda_{i}\left(0\right)\right]^{\prime}f_{t}+X_{it}^{\prime}\left[\beta\left(1\right)-\beta\left(0\right)\right]}_{\text{systematic causal effect}}+\underbrace{\left[\varepsilon_{it}\left(1\right)-\varepsilon_{it}\left(0\right)\right]}_{\text{idiosyncratic causal effect}}.\label{eq:unit-causaleffect}
\end{align}
This representation enables us to assess the sources of heterogeneous
causal effects arising from structural breaks in factor loadings,
covariate coefficients, or both. We will use $\tau_{it}^{*}$ to denote
the systematic causal effect
\begin{equation}
\tau_{it}^{*}=\left[\lambda_{i}\left(1\right)-\lambda_{i}\left(0\right)\right]^{\prime}f_{t}
+X_{it}^{\prime}\left[\beta\left(1\right)-\beta\left(0\right)\right].\label{eq:sys-causaleffect}
\end{equation}

Our object of interest is the systematic causal effect $\tau_{it}^*$. Note that the
idiosyncratic causal effect $\varepsilon_{it}\left(1\right)-\varepsilon_{it}\left(0\right)$
is unobserved, making the individual causal effect $\tau_{it}$ non-identifiable
without additional assumptions. Nevertheless, the cross-sectional average
causal effect $\bar{\tau}_{t}\equiv\frac{1}{n_{0}}\sum_{i=1}^{n_{0}}\tau_{it}$
remains identifiable, as the average difference in errors $\frac{1}{n_{0}}\sum_{i=1}^{n_{0}}\left[\varepsilon_{it}\left(1\right)-\varepsilon_{it}\left(0\right)\right]$
will converge to zero in probability if $n_{0}\rightarrow\infty$, under some standard assumptions on the error terms $\varepsilon_{it}(d)$.
Likewise, the time average causal effect $\bar{\tau}_{i}\equiv\frac{1}{T-T_{0}}\sum_{t=T_{0}+1}^{T}\tau_{it}$
is also identifiable if $T-T_0$ goes to infinity. Thus the average causal effect and the
average systematic causal effect can be asymptotically equivalent under additional assumptions.
In Section \ref{sec:Identification}, we detail the identification
of the systematic causal effect $\tau_{it}^{*}$.

The existing literature on causal inference typically focuses on modeling
the potential outcome $Y_{it}\left(0\right)$ exclusively, estimating
a model-implied $\hat{Y}_{it}\left(0\right)$, while imposing no model
on $Y_{it}\left(1\right)$. Consequently, the implied estimator for
the treatment effect for the treated unit is defined as $Y_{it}\left(1\right)-\hat{Y}_{it}\left(0\right)$.
We argue in subsections \ref{subsec:CompareIFE} and \ref{subsec:CompareSynth}
that this approach may lead to estimators with large imputation errors due to the presence
of $\varepsilon_{it}\left(d\right)$, which is $O_{p}\left(1\right)$.
To see a simple example, assume $Y_{it}\left(1\right)$ is model-free
while the model for $Y_{it}\left(0\right)$ consists of a single regressor
and no factors, i.e., $Y_{it}\left(0\right)=X_{it}\beta\left(0\right)+\varepsilon_{it}\left(0\right)$.
Given a consistent estimator $\hat{\beta}\left(0\right)$ for $\beta\left(0\right)$,
this model implies $\hat{Y}_{it}\left(0\right)=X_{it}\hat{\beta}\left(0\right)$.
The conventional estimator for the treatment effect for the treated
would be $Y_{it}-\hat{Y}_{it}\left(0\right)$ for $i\leq n_{0}$ and
$t>T_{0}$, which can be further represented as
\begin{align}
Y_{it}-\hat{Y}_{it}\left(0\right) & =Y_{it}\left(1\right)-Y_{it}\left(0\right)+Y_{it}\left(0\right)-\hat{Y}_{it}\left(0\right)\nonumber \\
 & =\tau_{it}+\underbrace{X_{it}\left[\beta\left(0\right)-\hat{\beta}\left(0\right)\right]}_{=o_{p}\left(1\right)}
 +\underbrace{\varepsilon_{it}\left(0\right)}_{=O_{p}\left(1\right)},\label{eq:conventionalcausal}
\end{align}
where $\tau_{it} =Y_{it}\left(1\right)-Y_{it}\left(0\right)$.
The estimator contains an idiosyncratic imputation error that does not vanish for a fixed treated unit and post-treatment period.
As a result, the conventional estimator $Y_{it}-\hat{Y}_{it}\left(0\right)$
is generally inconsistent for $\tau_{it}$.

In contrast, our approach models both potential outcomes,
\(Y_{it}(0)\) and \(Y_{it}(1)\). The purpose is not to claim that each
potential outcome can be predicted without error. In panel applications,
the fitted value of a potential outcome may differ substantially from
the realized outcome because of omitted variables, idiosyncratic shocks,
or model misspecification. This concern is especially relevant when
panel regressions have low explanatory power.

The distinction matters because conventional imputation methods usually
model only the missing untreated outcome \(Y_{it}(0)\), while using the
realized treated outcome \(Y_{it}(1)\) directly. Even if the imputation
recovers the systematic component of \(Y_{it}(0)\), the resulting
estimator may still contain the one-sided idiosyncratic error associated
with the untreated potential outcome. Thus, for a fixed unit and period,
the imputation error need not vanish, and the resulting estimator need
not be consistent for the realized individual effect. This problem can
be severe when \(\varepsilon_{it}(0)\) is large.

Our dual-modeling approach instead places the two potential outcomes on
equal footing. Both \(Y_{it}(0)\) and \(Y_{it}(1)\) are represented by
their systematic components under the factor structure. Using
(\ref{eq:unit-causaleffect}) and (\ref{eq:sys-causaleffect}), we can
write
\begin{align}
\tau_{it}^{*}
&=
\tau_{it}
-
\underbrace{
\left[
\varepsilon_{it}(1)-\varepsilon_{it}(0)
\right]
}_{\text{idiosyncratic causal component}}.
\label{eq:systematic}
\end{align}
Hence, even when \(\varepsilon_{it}(1)\) and \(\varepsilon_{it}(0)\)
are individually large, their difference may be small. In that case,
the systematic causal effect \(\tau_{it}^{*}\) can be close to the
realized individual effect \(\tau_{it}\).

The condition $\varepsilon_{it}(1)-\varepsilon_{it}(0)=0$
is not imposed as an identifying assumption for \(\tau_{it}^{*}\).
Rather, it is a special case under which the systematic causal effect
coincides with the realized individual effect. In general, the object of
interest in this paper is \(\tau_{it}^{*}\), the component of the
unit-time treatment effect generated by the systematic parts of the two
potential outcomes.

Table \ref{tab:Outcomes} provides a timeline of the potential and
observed outcomes. For the control group $(j>n_{0})$, the observed
outcome $Y_{jt}$ equals the potential outcome $Y_{jt}(0)$ for all
$t$. For the treatment group $(i\leq n_{0})$, the observed outcome
$Y_{it}$ equals the potential outcome $Y_{it}(0)$ before the intervention
$(t\leq T_{0})$ and equals the potential outcome $Y_{it}(1)$ after the intervention
$(t>T_{0})$. The last three rows of Table \ref{tab:Outcomes}
give, respectively, the counterfactual, individual and systematic
causal effects for the treated group. Our discussion of (\ref{eq:unit-causaleffect})
and (\ref{eq:sys-causaleffect}) demonstrates that we can consider
the systematic causal effect ($\tau_{it}^{*}$) as the primary objective
of interest.

Figure \ref{fig:unit-trends} provides an example of the relationship
among factors, realized individual trend before the intervention ($\lambda_{i}\left(0\right)^{\prime}\cdot f_{s}$,
$s\leq T_{0}$) and after the intervention ($\lambda_{i}\left(1\right)^{\prime}\cdot f_{t}$,
$t>T_{0}$), as well as the potential individual trend ($\lambda_{i}\left(0\right)^{\prime}\cdot f_{t}$,
$t>T_{0}$). In the figure, $f_{t}$ is represented as a smooth
function of time $t$, though smoothness is not required. The illustration
depicts only a single factor for clarity. We have omitted the unit
subscript $i$ from the factor loadings for notational simplicity;
thus, $\lambda\left(d\right)$ in the figure corresponds to $\lambda_{i}\left(d\right)$
for the treated unit $i$, where $d=0,1$. To emphasize the conceptual
framework, we have omitted the covariates in Table \ref{tab:Outcomes}
and Figure \ref{fig:unit-trends}.

\begin{table}[!tp]
\begin{centering}
\begin{tabular}{c|c|c}
\hline
 & $s\leq T_{0}$  & $t>T_{0}$\tabularnewline
\hline
\hline
$\underset{\left(i\leq n_{0}\right)}{\text{Treated}}$  & $Y_{is}=\underbrace{\lambda_{i}\left(0\right)'f_{s}+\varepsilon_{is}\left(0\right)}_{Y_{is}\left(0\right)}$  & $Y_{it}=\underbrace{\lambda_{i}\left(1\right)^{\prime}f_{t}+\varepsilon_{it}\left(1\right)}_{Y_{it}\left(1\right)}$\tabularnewline
\hline
$\underset{\left(j>n_{0}\right)}{\text{Control}}$  & $Y_{js}=\underbrace{\lambda_{j}\left(0\right)^{\prime}f_{s}+\varepsilon_{js}\left(0\right)}_{Y_{js}\left(0\right)}$  & $Y_{jt}=\underbrace{\lambda_{j}\left(0\right)^{\prime}f_{t}+\varepsilon_{jt}\left(0\right)}_{Y_{jt}\left(0\right)}$\tabularnewline
\hline
$\underset{\left(i\leq n_{0}\right)}{\text{Counterfactual}}$  &  & $Y_{it}\left(0\right)=\lambda_{i}\left(0\right)^{\prime}f_{t}+\varepsilon_{it}\left(0\right)$\tabularnewline
\hline
$\underset{\left(i\leq n_{0}\right)}{\text{Individual causal effect}}$  &  & $\tau_{it}=Y_{it}\left(1\right)-Y_{it}\left(0\right)$\tabularnewline
\hline
$\underset{\left(i\leq n_{0}\right)}{\text{Systematic causal effect}}$  &  & $\tau_{it}^{*}=\lambda_{i}\left(1\right)^{\prime}f_{t}-\lambda_{i}\left(0\right)^{\prime}f_{t}$\tabularnewline
\hline
\end{tabular}
\end{centering}
\caption{Outcomes before and after the intervention.}\label{tab:Outcomes}
\end{table}

\begin{figure}[!tp]
\begin{centering}
\includegraphics[width=0.85\columnwidth]{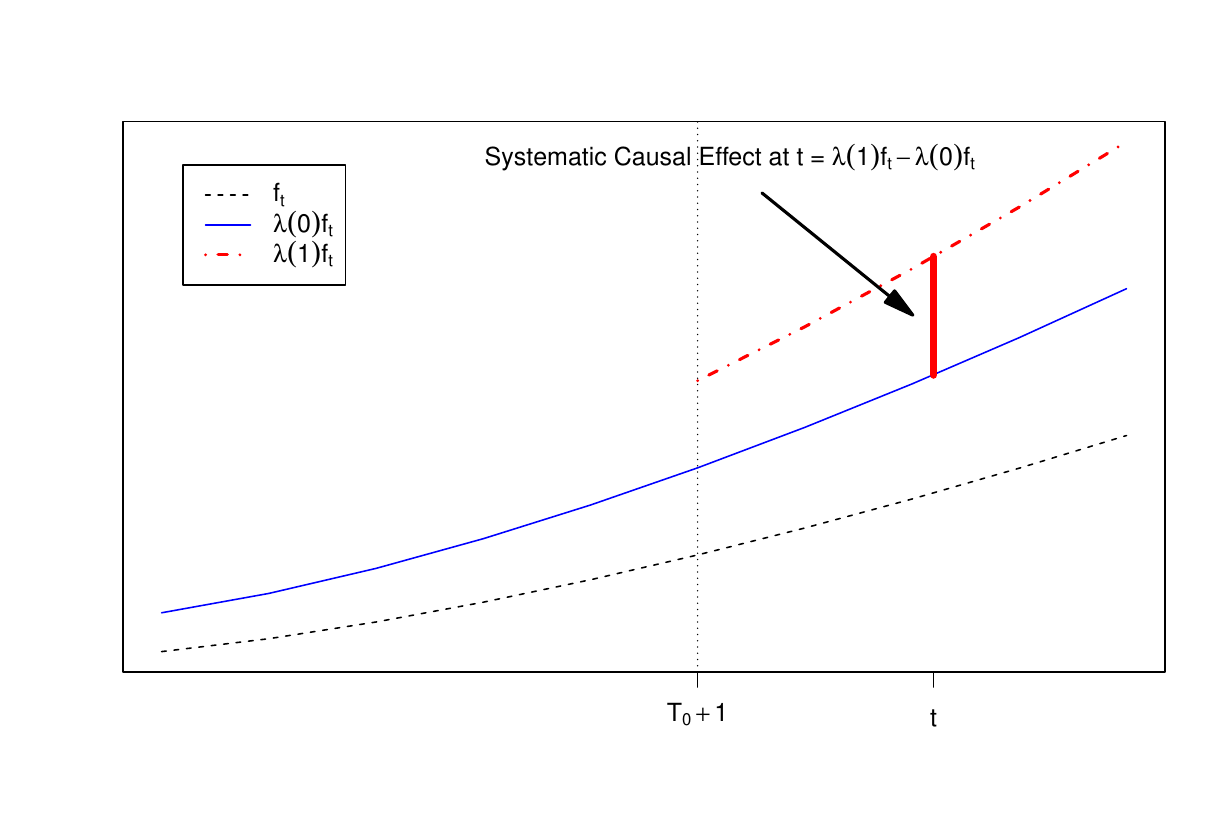}
\end{centering}
\caption{Individual trend before and after the intervention.}\label{fig:unit-trends}
\end{figure}

In Section \ref{sec:Identification}, we show that the systematic
causal effect $\tau_{it}^{*}$ in (\ref{eq:sys-causaleffect}) is
identifiable because $\left\{ \lambda_{i}\left(d\right),f_{t},\beta\left(d\right)\right\} $,
$d=0,1$, are all identifiable. Then a natural estimator for $\tau_{it}^{*}$
is given by
\begin{align}
\hat{\tau}_{it}^{*} & =\left[\hat{\lambda}_{i}\left(1\right)-\hat{\lambda}_{i}\left(0\right)\right]^{\prime}\hat{f}_{t}
+X_{it}^{\prime}\left[\hat{\beta}\left(1\right)-\hat{\beta}\left(0\right)\right],\label{eq:estimator-main}
\end{align}
where the hatted variables are the corresponding estimators.

We will show that $\hat{\tau}_{it}^{*}$ is a consistent estimator for the systematic causal effect $\tau_{it}^{*}$, i.e.,
\[
\hat{\tau}_{it}^{*}-\tau_{it}^{*}=o_{p}\left(1\right),
\]
when the number of untreated units ($n-n_{0}$) and the number of
periods before and after the treatment ($T_{0}$ and $T-T_{0}$)
are large. Such a model-based approach allows us to
identify the systematic component irrespective of the idiosyncratic errors. Existing causal inference methods, such
as difference-in-differences, synthetic control, and matrix completion
(e.g., \citet{Athey_2021,Bai_Ng_2021}) focus on constructing the counterfactual $Y_{it}\left(0\right)$
from the control group. In contrast, our model-based method emphasizes
the systematic component of the difference $Y_{it}\left(1\right)-Y_{it}\left(0\right)$,
rather than attributing a one-sided idiosyncratic residual to the causal effect.

\subsection{Relationship with Interactive Fixed Effects Models}\label{subsec:CompareIFE}

The causal model that we propose is closely related to interactive
fixed-effect approaches such as \citet{Gobillon_Magnac_2016} and \citet{Xu_2017},
but differs in the role assigned to the factors. In those approaches, the
factor component is primarily used as a latent control for unobserved
confounding: common shocks enter the untreated potential outcome with
unit-specific exposures, analogous to observed covariates. The treatment
effect is then estimated as the difference between the realized treated
outcome and an imputed untreated outcome.

In contrast, our framework allows the treated unit's exposure to the same
common shocks to change after the intervention. Thus the factor loadings are
not merely nuisance parameters used to control for unobserved heterogeneity;
changes in these loadings are themselves part of the systematic causal effect.
This distinction is especially important when the intervention changes how a
treated unit responds to aggregate conditions, industry-wide shocks, or other
latent forces.

The potential outcomes in the causal model of \citet{Gobillon_Magnac_2016} and \citet{Xu_2017} have the following representation
\begin{align}
Y_{it}\left(0\right) & =\lambda_{i}^{\prime}f_{t}+X_{it}^{\prime}\beta+\varepsilon_{it}, \nonumber\\
Y_{it}\left(1\right) & =\delta_{it}+\lambda_{i}^{\prime}f_{t}+X_{it}^{\prime}\beta+\varepsilon_{it}, \label{eq:interactive-effects}
\end{align}
where the error term $\varepsilon_{it}$ is not indexed by the treatment
status ($d=0,1$) and $\delta_{it}=Y_{it}\left(1\right)-Y_{it}\left(0\right)$
is defined as the individual causal effect. The estimator for this
causal effect is
\begin{align}
\hat{\delta}_{it} & =Y_{it}\left(1\right)-\hat{Y}_{it}\left(0\right)\nonumber \\
 & =Y_{it}-\left(\hat{\lambda}_{i}^{\prime}\hat{f}_{t}+X_{it}^{\prime}\hat{\beta}\right),\;\;t>T_{0},\;i=1,\ldots,n_{0}.\label{eq:estimand-old}
\end{align}
Plugging in the model for $Y_{it}\left(1\right)$ to obtain
\begin{align*}
\hat{\delta}_{it} & =\left(\delta_{it}+\lambda_{i}^{\prime}f_{t}+X_{it}^{\prime}\beta+\varepsilon_{it}\right)-\left(\hat{\lambda}_{i}^{\prime}\hat{f}_{t}+X_{it}^{\prime}\hat{\beta}\right)\\
 & =\delta_{it}+\left(\lambda_{i}^{\prime}f_{t}-\hat{\lambda}_{i}^{\prime}\hat{f}_{t}\right)+X_{it}^{\prime}\left(\beta-\hat{\beta}\right)+\varepsilon_{it}.
\end{align*}
Under appropriate assumptions such as those in \citet{Bai_2009},
$\lambda_i'f_t-\hat{\lambda}_i'\hat f_t=o_p(1)$ and
$\beta-\hat\beta=o_p(1)$. Therefore,
\begin{equation}
\hat{\delta}_{it}-\delta_{it}
=
o_p(1)+\varepsilon_{it}
=
O_p(1).
\label{eq:imputation-error}
\end{equation}
Consequently, under the modeling strategies of \citet{Gobillon_Magnac_2016} and \citet{Xu_2017}, the idiosyncratic error $\varepsilon_{it}$
persists as an imputation error in the individual causal estimate. This
component does not vanish for a fixed treated unit and time period, but
it may be averaged out when the number of treated units, $n_0$, is large.
This  framework may be less suitable when $n_0$ is small, as in
many applications in the synthetic control literature.

Similarly, \citet{Callaway_Karami_2023} propose to model the untreated
potential outcomes as
$
Y_{it}\left(0\right)=\theta_{t}+\xi_{i}+\lambda_{i}^{\prime}f_{t}+\varepsilon_{it}.
$
This framework focuses on the average treatment effect. Their framework is
helpful for the applications in the difference-in-differences literature
but is not designed for the small $n_{0}$ case as in the synthetic
control literature.

More importantly, \citet{Gobillon_Magnac_2016,Xu_2017,Callaway_Karami_2023} focus on modeling the untreated potential outcome
$Y_{it}(0)$ and impose no restrictions on the causal effect $\delta_{it}$.
They do not specify an explicit model for the treated potential outcome
$Y_{it}(1)$. This flexibility comes at the cost of a nonvanishing
imputation error in estimators of individual causal effects, as shown
in \eqref{eq:imputation-error}. In contrast, we explicitly model both
$Y_{it}(0)$ and $Y_{it}(1)$, which restricts the causal effect
$\tau_{it}$ through \eqref{eq:unit-causaleffect}. The resulting estimator
\eqref{eq:estimator-main} targets the systematic unit-time causal effect
and does not contain the one-sided idiosyncratic imputation component that
appears in \eqref{eq:imputation-error}, even when the number of treated
units is small.

Recently, \citet{Brown_2023} proposed a framework
for modeling potential outcomes across all treatment statuses using
factor models, employing \citet{Pesaran_2006}'s Common Correlated Effects
(CCE) method to identify average causal effects. Their approach differs
from ours in two main ways. First, they concentrate on identifying
average causal effects, which necessitates that the group size for
treated units approaches infinity at the same rate as the group size
for untreated units.
In contrast, our framework can identify unit-time systematic effects even when
the treated group is small, and average realized effects under additional
averaging assumptions.
 While their method is designed
for a large $n$ and small $T$ setup, we focus on a large $n$ and
large $T$ configuration. Second, in their framework, the common component
$\lambda_{i}^{\prime}f_{t}$ remains unaffected by treatment status,
aligning with the setup used by \cite{Ferman_2021} in the context
of the synthetic control method and \cite{Freyaldenhoven_2021} in the
context of event-study design. In contrast, our approach permits this
component to vary across treatment statuses.

Panel models with factor structures have been studied extensively; see,
for example,\citet{Ahn_2001,Ahn_2013,Pesaran_2006,Bai_2009}. Much of this literature focuses on estimating slope
coefficients associated with observed covariates. Recent work has brought
factor models into causal inference, mainly to estimate treatment effects
averaged across treated units or post-treatment periods.
This paper instead uses a dual factor structure for $Y_{it}(0)$ and $Y_{it}(1)$ to
study systematic unit-time causal effects for a given unit $i$ and time period
$t$. The resulting estimand has a direct structural interpretation and
can be consistently estimated under standard large-panel conditions.

\subsection{Relationship with Difference-in-Differences Using Interactive Fixed Effects}

We next show that difference-in-differences, with or without interactive fixed
effects, can be viewed as a special case of our setup.
 A standard difference-in-differences model with constant treatment effects and
a common treatment timing is represented by the following two-way-fixed-effects
(TWFE) panel regression model:
\begin{equation*}
Y_{it}=\alpha_{i}+\theta_{t}+\rho D_{it}+X_{it}^{\prime}\beta+\varepsilon_{it},
\end{equation*}
where $D_{it}$ is the treatment indicator defined in (\ref{eq:Dit}).
Adding the interactive fixed effects gives
\begin{equation*}
Y_{it}=\alpha_{i}+\theta_{t}+\rho D_{it}+\lambda_{i}^{\prime}f_{t}+X_{it}^{\prime}\beta+\varepsilon_{it}.
\end{equation*}
In both specifications, $\rho$
represents the treatment effect.

The corresponding potential outcomes are given by
\begin{align*}
Y_{it}\left(0\right) & =\alpha_{i}+\theta_{t}+\lambda_{i}^{\prime}f_{t}+X_{it}^{\prime}\beta+\varepsilon_{it},\nonumber \\
Y_{it}\left(1\right) & =\alpha_{i}+\theta_{t}+\rho+\lambda_{i}^{\prime}f_{t}+X_{it}^{\prime}\beta+\varepsilon_{it},  
\end{align*}
where the error term $\varepsilon_{it}$ is not indexed by the treatment
status $d$. As a result, $\tau_{it}=\tau_{it}^{*}$. Define $g_{t}=\left[1,\theta_{t},f_{t}^{\prime}\right]^{\prime}$,
$\lambda_{i}\left(0\right)=\left[\alpha_{i},1,\lambda_{i}^{\prime}\right]^{\prime}$,
$\lambda_{i}\left(1\right)=\left[\alpha_{i}+\rho,1,\lambda_{i}^{\prime}\right]^{\prime}$,
and then the potential outcomes can be represented as
\begin{align}
Y_{it}\left(0\right) & =\lambda_{i}\left(0\right)^{\prime}g_{t}+X_{it}^{\prime}\beta+\varepsilon_{it},\nonumber \\
Y_{it}\left(1\right) & =\lambda_{i}\left(1\right)^{\prime}g_{t}+X_{it}^{\prime}\beta+\varepsilon_{it}, \label{eq:did-po-1}
\end{align}
which is a special case of causal model (\ref{eq:mainmodel})
in which $\beta\left(d\right)=\beta,\;d=0,1$. Our representation
of the causal effect (\ref{eq:unit-causaleffect}) or (\ref{eq:sys-causaleffect})
takes into account variations in $\lambda_{i}$ and $\beta$ across
different treatment statuses to model the heterogeneous causal effects.

If one wants to use a TWFE model to learn about the heterogeneous causal
effects, the regression model can be specified as
\begin{equation*}
Y_{it}=\alpha_{i}+\theta_{t}+\rho_{i}D_{it}+\lambda_{i}^{\prime}f_{t}+X_{it}^{\prime}\beta+\varepsilon_{it},
\end{equation*}
where the coefficient of $D_{it}$ is individual-dependent. The corresponding
potential outcomes are given by
\begin{align*}
Y_{it}\left(0\right) & =\alpha_{i}+\theta_{t}+\lambda_{i}^{\prime}f_{t}+X_{it}^{\prime}\beta+\varepsilon_{it},\nonumber \\
Y_{it}\left(1\right) & =\alpha_{i}+\theta_{t}+\rho_{i}+\lambda_{i}^{\prime}f_{t}+X_{it}^{\prime}\beta+\varepsilon_{it}. 
\end{align*}
Define $g_{t}=\left[1,\theta_{t},f_{t}^{\prime}\right]^{\prime}$,
$\lambda_{i}\left(0\right)=\left[\alpha_{i},1,\lambda_{i}^{\prime}\right]^{\prime}$,
$\lambda_{i}\left(1\right)=\left[\alpha_{i}+\rho_{i},1,\lambda_{i}^{\prime}\right]^{\prime}$.
Then the preceding potential-outcome equations can be represented as
(\ref{eq:did-po-1}), again a special case of our causal model
(\ref{eq:mainmodel}).

\subsection{Relationship with Synthetic Control}\label{subsec:CompareSynth}
To deliver the main idea, assume that the potential outcomes follow
a factor model without covariates and the error term is not indexed
by $d$ (so $\tau_{it}=\tau_{it}^{*}$),
\begin{equation*}
Y_{it}\left(d\right)=\lambda_{i}\left(d\right)^{\prime}f_{t}+\varepsilon_{it},\;\;d=0,1.  
\end{equation*}
Assume $i=1$ is treated with a policy intervention in period $T_{0}+1$,
and the unaffected control units are $i=2,\ldots,n$. The synthetic control
method constructs the counterfactual $Y_{1t}\left(0\right)$ as a
weighted average of the observed outcomes for the control units:
\[
\widehat{Y}_{1t}\left(0\right)=\sum_{i=2}^{n}\omega_{i}Y_{it},\quad t>T_{0},\;\omega_{i}\ge0,\;\sum_{i=2}^{n}\omega_{i}=1.
\]
Then $\widehat{Y}_{1t}\left(0\right)$ is equal to
\begin{align*}
\widehat{Y}_{1t}\left(0\right)
 & =\left(\sum_{i=2}^{n}\omega_{i}\lambda_{i}\left(0\right)^{\prime}\right)f_{t}+\sum_{i=2}^{n}\omega_{i}\varepsilon_{it}.
\end{align*}
The synthetic causal effect for $t>T_{0}$ is
\begin{align*}
\tau_{1t}^{synth} & \equiv Y_{1t}-\widehat{Y}_{1t}\left(0\right)\\
 & =\lambda_{1}\left(1\right)^{\prime}f_{t}+\varepsilon_{1t}-\left\{ \left(\sum_{i=2}^{n}\omega_{i}\lambda_{i}\left(0\right)^{\prime}\right)f_{t}+\sum_{i=2}^{n}\omega_{i}\varepsilon_{it}\right\} \\
 & =\left[\lambda_{1}\left(1\right)-\sum_{i=2}^{n}\omega_{i}\lambda_{i}\left(0\right)\right]^{\prime}f_{t}+\left[\varepsilon_{1t}-\sum_{i=2}^{n}\omega_{i}\varepsilon_{it}\right].
\end{align*}
Our model-based causal effect is given by
\[
\tau_{1t}^{factor}\equiv\left[\lambda_{1}\left(1\right)-\lambda_{1}\left(0\right)\right]^{\prime}f_{t}=Y_{1t}\left(1\right)-Y_{1t}\left(0\right),\quad t>T_{0}.
\]
The difference between these two is
\[
\tau_{1t}^{synth}-\tau_{1t}^{factor}=\left[\lambda_{1}\left(0\right)-\sum_{i=2}^{n}\omega_{i}\lambda_{i}\left(0\right)\right]^{\prime}f_{t}+\left[\varepsilon_{1t}-\sum_{i=2}^{n}\omega_{i}\varepsilon_{it}\right],\;t>T_{0}.
\]
The synthetic control chooses the weights such that the distance between
$Y_{1t}$ and $\sum_{i=2}^{n}\omega_{i}Y_{it}$ is small for $t\leq T_{0}$.
Note that for $t\leq T_{0}$,
\begin{align*}
Y_{1t}-\sum_{i=2}^{n}\omega_{i}Y_{it} & =Y_{1t}\left(0\right)-\sum_{i=2}^{n}\omega_{i}Y_{it}\left(0\right)\\
 & =\lambda_{1}\left(0\right)^{\prime}f_{t}+\varepsilon_{1t}-\sum_{i=2}^{n}\omega_{i}\left[\lambda_{i}\left(0\right)^{\prime}f_{t}+\varepsilon_{it}\right]\\
 & =\left[\lambda_{1}\left(0\right)-\sum_{i=2}^{n}\omega_{i}\lambda_{i}\left(0\right)\right]^{\prime}f_{t}+\left[\varepsilon_{1t}-\sum_{i=2}^{n}\omega_{i}\varepsilon_{it}\right].
\end{align*}
In the case that the weights $\omega_{i}$ are chosen such that
$
\lambda_{1}\left(0\right)\approx\sum_{i=2}^{n}\omega_{i}\lambda_{i}\left(0\right),
$
we have
\[
\tau_{1t}^{synth}-\tau_{1t}^{factor}\approx\varepsilon_{1t}-\sum_{i=2}^{n}\omega_{i}\varepsilon_{it},\quad t>T_{0}.
\]

If we expect that the factors already capture most of the cross-sectional
dependence, then the residual correlation between $\varepsilon_{1t}$ and
$\sum_{i=2}^{n}\omega_i\varepsilon_{it}$ should be weak. In the special case
where $\{\varepsilon_{it}\}$ is i.i.d.\ across $i$ and the weighted control
residual is asymptotically negligible, i.e.,
$\sum_{i=2}^{n}\omega_i\varepsilon_{it}=o_p(1)$, we have
$\varepsilon_{1t}-\sum_{i=2}^{n}\omega_i\varepsilon_{it}
=\varepsilon_{1t}+o_p(1)=O_p(1)$. Therefore,
$\tau_{1t}^{\mathrm{synth}}-\tau_{1t}^{\mathrm{factor}}=O_p(1)$.
Hence, the estimator proposed in this paper and the synthetic-control estimator
are generally not asymptotically equivalent for a fixed treated unit and time
period. Their difference may become small under additional averaging over many
post-treatment periods or many treated units. For a given unit-time treatment
effect, however, the difference need not be negligible.

\citet{Hsiao_Ching_Wan_2012} adopt a similar method as synthetic control.
They start with a factor model for the potential outcome $Y_{it}\left(0\right)$
and propose using a linear function of outcomes for the untreated
units to estimate the counterfactual $Y_{1t}\left(0\right)$, $t>T_{0}$.
Accordingly, the estimator for the individual causal effect is
\[
Y_{1t}-\hat{Y}_{1t}\left(0\right),\;t>T_{0}.
\]
They use a simple regression method to estimate the optimal linear
function. Under their assumptions, \citet{Hsiao_Ching_Wan_2012} show that the
individual causal effect estimator is unbiased, and they also consider
the corresponding post-treatment average estimator.
In comparison, our theory provides large-sample inference for the
unit-time systematic causal effect.
Both our method and the approach of \citet{Hsiao_Ching_Wan_2012} do not require numerical optimization and are easy
to implement. \citet{Li_Bell_2017} extend \citet{Hsiao_Ching_Wan_2012}'s framework
and derive large sample theory for the average treatment effect under
weaker conditions. They further propose using LASSO selection when
the number of control units is large.

Recent research by \cite{Imbens_Viviano_2023} explicitly models $Y_{it}\left(0\right)$
as possessing a factor structure that includes both strong and weak
factors, with the factors serving to control for unobserved confounders.
They extend the Synthetic Difference-in-Differences framework developed
by \citet{Arkhangelsky_2021} and achieve identification under slightly
weaker assumptions, such as limited confoundedness over either units
or time. However, while their model assumes constant causal effect,
our focus is on heterogeneous and time-varying causal effects.

\citet{Li2020} and \citet{Chernozhukov_2025} have developed inferential
theory for estimating the time average causal effect for the treated
unit, $\bar{\tau}_{1}\equiv\frac{1}{T-T_{0}}\sum_{t=T_{0}+1}^{T}\tau_{1t}$,
under the assumption that both $T_{0}$ and $T-T_{0}$ are large.
While \citet{Li2020} assumes a fixed number of untreated units, \citet{Chernozhukov_2025} consider a large number of untreated units. In this paper,
we focus on the unit-specific systematic causal effect in each post-treatment
period ($\tau_{1t}^{*}$, $t>T_{0}$). Our framework assumes that $T_{0}$, $T-T_{0}$,
and the number of untreated units ($n-n_{0}$) are all large.

\section{Identification of the Causal Factor Model}\label{sec:Identification}

To illustrate the main idea, we first consider the case without covariates.
The factor model implies the systematic causal effect
\begin{align*}
\tau_{it}^{*} & =\left[\lambda_{i}\left(1\right)-\lambda_{i}\left(0\right)\right]^{\prime}f_{t},\quad i\leq n_{0},\;t>T_{0}.
\end{align*}
This systematic causal effect $\tau_{it}^{*}$ is identifiable under the
usual factor-model normalization. Although the factors and loadings are
identified only up to a common rotation, the causal contrast
$\left[\lambda_{i}\left(1\right)-\lambda_{i}\left(0\right)\right]^{\prime}f_{t}$
is invariant to that rotation. Thus, once the common factor space is
recovered from the untreated units and the pre- and post-treatment
loadings are expressed in the same normalized factor space,
\(\tau_{it}^{*}\) is identified under large $n-n_{0}$, $T_{0}$, and
$T-T_{0}$.

\begin{enumerate}[label=\alph*)]
\item The common factor space for $f_{t}$ ($1\leq t\leq T$) can be recovered,
up to the usual rotation, using existing methods such as the principal
component analysis of the untreated units
$\left\{ Y_{it}\right\}$, $i>n_{0},\;t=1,\ldots,T$.
\item $\lambda_{i}\left(0\right)$ is
identified in the same factor space by regressing $Y_{it}$ on $f_{t}$
for $t\leq T_{0}$, $i\leq n_{0}$.
\item  $\lambda_{i}\left(1\right)$
is identified in the same factor space by regressing $Y_{it}$ on
$f_{t}$ for $t>T_{0}$, $i\leq n_{0}$.
\item To test
$H_{0}:\tau_{it}^{*}=0$ for $t>T_{0},\;i\leq n_{0}$ it is sufficient to test
 $H_0: \lambda_i(1)=\lambda_i(0)$  for $i\leq n_0.$
Thus, a structural-break test in the treated-unit factor
regression $Y_{it}=\lambda_i'f_t+\varepsilon_{it},
\; i\leq n_0,\;1\leq t\leq T,$ provides a natural test of parameter stability.
\end{enumerate}

The identification strategy works for the cases when $n_{0}$
is either small or large. In particular, it works for the special
case where $n_{0}=1$, similar to the synthetic control setup.

A similar strategy can be applied when covariates are included. Consider
the data generating process
\[
Y_{it}\left(d\right)
=
\lambda_{i}\left(d\right)^{\prime}f_{t}
+
X_{it}^{\prime}\beta\left(d\right)
+
\varepsilon_{it}\left(d\right),
\quad d=0,1,
\]
where changes in \(\beta(d)\) capture systematic treatment-effect
heterogeneity associated with covariates. The systematic causal effect is
\begin{align}
\tau_{it}^{*}
&=
\left\{
\lambda_{i}\left(1\right)^{\prime}f_{t}
+
X_{it}^{\prime}\beta\left(1\right)
\right\}
-
\left\{
\lambda_{i}\left(0\right)^{\prime}f_{t}
+
X_{it}^{\prime}\beta\left(0\right)
\right\}
\nonumber \\
&=
\left[
\lambda_{i}\left(1\right)-\lambda_{i}\left(0\right)
\right]^{\prime}f_{t}
+
X_{it}^{\prime}
\left[
\beta\left(1\right)-\beta\left(0\right)
\right],
\quad i\leq n_{0},\;t>T_{0}.
\label{eq:tauit*}
\end{align}
Under the same factor normalization and appropriate rank conditions, this
object is identifiable. The factors and loadings are identified only up
to a common rotation, but the factor component
\[
\left[
\lambda_{i}\left(1\right)-\lambda_{i}\left(0\right)
\right]^{\prime}f_{t}
\]
is invariant to that rotation. The covariate component
\[
X_{it}^{\prime}
\left[
\beta\left(1\right)-\beta\left(0\right)
\right]
\]
is identified from the corresponding pre- and post-treatment regressions,
provided the usual rank conditions hold.

\begin{enumerate}[label=\alph*)]
\item The common factor space for \(f_{t}\), \(1\leq t\leq T\), can be
recovered, up to the usual rotation, using existing methods such as
panel regression with interactive fixed effects applied to the untreated
units \(\left\{Y_{it},X_{it}\right\}\), \(i>n_{0}\), \(t=1,\ldots,T\).

\item Given the normalized factor estimates, \(\lambda_i(0)\) and
\(\beta(0)\) are identified in the same representation by regressing
\(Y_{it}\) on \(f_t\) and \(X_{it}\) for \(i\leq n_0\), \(t\leq T_0\).

\item Given the same normalized factor estimates, \(\lambda_i(1)\) and
\(\beta(1)\) are identified in the same representation by regressing
\(Y_{it}\) on \(f_t\) and \(X_{it}\) for \(i\leq n_0\), \(t>T_0\).

\item  To test $H_{0}:\tau_{it}^{*}=0$ for $t>T_{0},\;i\leq n_{0}$ it is sufficient to test
 $H_0: \lambda_i(1)=\lambda_i(0)$  for $i\leq n_0$ and $\beta(1)=\beta(0)$.
Thus, a structural-break test in the treated-unit factor
regression $Y_{it}=\lambda_i'f_t+X_{it}\beta +\varepsilon_{it},
\; i\leq n_0,\;1\leq t\leq T$ provides a natural test of parameter stability.
\end{enumerate}

\subsection*{Discussion: Explicit Individual Fixed Effects}
So far, for notational compactness, the individual fixed effects have
not been written separately from the factor component. This is without
loss of generality. If one component of \(f_t\) is identically equal to
one, then the factor term \(\lambda_i'f_t\) implicitly includes an
individual-specific intercept. Equivalently, writing
\[
\widetilde f_t
=
\left(1,f_t'\right)'
\quad\text{and}\quad
\widetilde\lambda_i(d)
=
\left(a_i(d),\lambda_i(d)'\right)',
\]
the term $a_i(d)+\lambda_i(d)'f_t$
can be represented as $
\widetilde\lambda_i(d)'\widetilde f_t$.
Thus, an individual fixed effect can always be incorporated into the
factor structure by including a constant factor.

Although this representation is convenient, it is often more transparent
 to write the individual fixed effects explicitly. We
therefore consider the equivalent specification
\begin{equation} \label{eq:explicit-intercept}
Y_{it}(d)
=
a_i(d)+\lambda_i(d)'f_t
+
X_{it}'\beta(d)+\varepsilon_{it}(d),
\qquad d=0,1.
\end{equation}
Here \(a_i(d)\) is the individual fixed effect under treatment state
\(d\), while $f_t$ denotes the remaining common factors. The
corresponding systematic causal effect is
\begin{equation} \label{eq:tauit-intercept-model}
\tau_{it}^{*}=
a_i(1)-a_i(0)
+[\lambda_i(1)-\lambda_i(0)]'f_t +
X_{it}'\big[\beta(1)-\beta(0)\big].
\end{equation}

This explicit-intercept notation will be useful below, especially when
the treatment is allowed to shift the common factor process by a
constant. In that case, part of the factor shift can be absorbed into a
unit-specific intercept. Writing \(a_i(0)\) and \(a_i(1)\) separately
therefore clarifies the distinction between a treatment-induced change
in the individual fixed effect and a constant shift in the factor
component. Except in such cases, we continue to use the more compact
notation in which the individual fixed effects are not written
separately.

\section{Extension to the Case with Potential Factors}\label{sec:potential-factor}

In some settings the policy intervention may alter not only treated
units' exposure to common shocks, through factor loadings, but also the
common-shock process itself. This arises naturally when the intervention
is large enough to change equilibrium conditions faced by the treated
group, for example through market clearing prices, entry and exit,
institutional constraints, expectations, or other spillovers that are
shared across treated units. In such cases, modeling a common factor
process that is invariant to treatment may be too restrictive, and it is
appropriate to allow treated outcomes to be driven by a different
post-treatment factor process, denoted \(f_t(1)\).

Identification of a flexible post-treatment factor process requires
sufficient treated cross-sectional variation, which motivates our
large-\(n_0\) analysis. When the treated group is small, we consider
restricted perturbations of \(f_t(1)\) that restore tractable
identification while still permitting the intervention to affect both
the factor process and factor loadings.

\subsection{Potential Factors Framework with a Large Number of Treated Units}

Assume \(n_0/n\rightarrow c\in(0,1)\), with \(n_0\), \(n\), the number
of pre-treatment periods, and the number of post-treatment periods all
large. In this case, the post-treatment treated sample
contains enough cross-sectional and time-series variation to identify a
flexible factor process for the treated group.

Before the intervention, the same common factor process drives both
treated and untreated units. After the intervention, the untreated units
continue to be driven by \(f_t(0)\), while the treated units may be
driven by a different factor process \(f_t(1)\). Thus, for treated units
in the post-treatment period, \(f_t(0)\) can be interpreted as the
counterfactual factor process that would have prevailed without the
intervention, whereas \(f_t(1)\) is the realized factor process under
treatment.

The potential outcome model is
\[
Y_{it}(d)=\lambda_i(d)'f_t(d)+\varepsilon_{it}(d),\qquad d=0,1.
\]
Observed outcomes satisfy $Y_{it}=Y_{it}(0)$ for controls $(i>n_0)$ and,
for treated units $(i\le n_0)$, $Y_{it}=Y_{it}(0)$ before treatment and
$Y_{it}=Y_{it}(1)$ after treatment.
The systematic causal effect is
\begin{align}
\tau_{it}^{*} & =\lambda_{i}\left(1\right)^{\prime}f_{t}\left(1\right)-\lambda_{i}\left(0\right)^{\prime}f_{t}\left(0\right),\;\;t>T_{0},\;i\leq n_{0}.\label{eq:causaleffect-groupfactor}
\end{align}
This representation allows us to evaluate the source of the causal
effects due to structural breaks in both factor loadings and factors.

The systematic causal effect \(\tau_{it}^{*}\) is identifiable because the
post-treatment treated common component \(\lambda_i(1)'f_t(1)\) and the
counterfactual untreated common component \(\lambda_i(0)'f_t(0)\) are
identified.  A simple identification
strategy is given as follows.

\begin{enumerate}[label=\alph*)]
\item The factors $f_{t}\left(0\right)$ ($1\leq t\leq T$) can be identified
using principal component analysis of the untreated units $\left\{ Y_{it}\right\}$,
$i>n_{0},\;t=1,\ldots,T$.
\item $\lambda_{i}\left(0\right)$, $i\leq n_{0}$, is identified by regressing
$Y_{it}$ on $f_{t}\left(0\right)$ for $t\leq T_{0}$.
\item The counterfactual $\lambda_{i}\left(0\right)^{\prime}f_{t}\left(0\right)$,$\;i\leq n_{0},\;t>T_{0}$,
can then be constructed as the product of the above two.
\item The product $\lambda_{i}\left(1\right)^{\prime}f_{t}\left(1\right)$
is identified by the principal component analysis of $Y_{it}$, $i\leq n_{0},\;t>T_{0}$.
\end{enumerate}

The estimator of the systematic causal effect is given by
\begin{align}
\hat{\tau}_{it}^{*} & =\hat{\lambda}_{i}\left(1\right)^{\prime}\hat{f}_{t}\left(1\right)-\hat{\lambda}_{i}\left(0\right)^{\prime}\hat{f}_{t}\left(0\right),\;\;t>T_{0},\;i\leq n_{0}.\label{eq:estimator-groupfactor}
\end{align}

The above procedure can be extended to include regressors. With
regressors, the systematic causal effect becomes
\[
\tau_{it}^{*}
=\lambda_i(1)'f_t(1)-\lambda_i(0)'f_t(0)
+X_{it}'[\beta(1)-\beta(0)],
\qquad t>T_0,\ i\le n_0 .
\]
The corresponding estimator replaces the unknown quantities by their
sample analogs.
We show that $\hat{\tau}_{it}^{*}$ is a consistent estimator of the
systematic causal effect $\tau_{it}^{*}$ and derive its asymptotic
standard error in Proposition~\ref{prop:prop-2} below.
\subsection{Potential Factors Framework with a Small Number of Treated Units}
The identification strategy above relies on estimating $f_t(1)$ from the
post-treatment treated sample
$\{Y_{it}: i\le n_0,\ t>T_0\}$. When $n_0$ is small, and especially
when $n_0=1$, principal component analysis using treated post-treatment
outcomes is not feasible. A parsimonious alternative is to restrict the
effect of the intervention on the factor process.

Suppose that, for treated units in the post-treatment period,
\begin{equation}
\label{eq:ft_shift_constant_pf}
f_t(1)=f_t(0)+\Delta,\qquad t>T_0,
\end{equation}
where $\Delta$ is an $r\times 1$ constant vector. If the outcome model
contains a unit-specific intercept, this constant factor shift is
absorbed into the post-treatment intercept.\footnote{More generally, one may allow
\[
f_t(1)=A f_t(0)+\Delta,
\]
where $A$ is a constant $r\times r$ matrix. This generalization does not
affect the results below, because
\[
\lambda_i(1)'f_t(1)
=
\lambda_i(1)'A f_t(0)+\lambda_i(1)'\Delta
=
\tilde\lambda_i(1)'f_t(0)+\lambda_i(1)'\Delta,
\qquad
\tilde\lambda_i(1):=A'\lambda_i(1).
\]
The linear transformation $A$ is therefore absorbed into the
post-treatment factor loadings, while the constant shift
$\lambda_i(1)'\Delta$ is absorbed into the post-treatment intercept. There is no need to separately identify $A$ and $\Delta$.}
Indeed,
\[
a_i(1)+\lambda_i(1)'f_t(1)
=
a_i(1)+\lambda_i(1)'f_t(0)+\lambda_i(1)'\Delta
=
\tilde a_i(1)+\lambda_i(1)'f_t(0),
\]
where $a_i(1)$ is the post-treatment intercept, see (\ref{eq:explicit-intercept}), and
\[
\tilde a_i(1):=a_i(1)+\lambda_i(1)'\Delta.
\]
Thus, $\Delta$ itself need not be separately identified. Its contribution
to the systematic causal effect is summarized by the total intercept
shift
\[
\kappa_i
:=
\tilde a_i(1)-a_i(0)
=
a_i(1)-a_i(0)+\lambda_i(1)'\Delta .
\]
The corresponding systematic causal effect is
\begin{equation}
\label{eq:tau_star_constshift_pf}
\tau_{it}^{*,\Delta}
=
\kappa_i+
\big[\lambda_i(1)-\lambda_i(0)\big]'f_t(0)+ X_{it}'\big[\beta(1)-\beta(0)\big],
\quad t>T_0,\ i\le n_0.
\end{equation}
This has the same form as the baseline model once unit-specific
intercepts are included.  Hence,
the constant-shift specification is covered by the baseline estimation
and inference results. Proposition~\ref{prop:prop-1} below therefore applies after
augmenting the unit-level regression with a constant term.

To illustrate the estimation of $\tau_{it}^{*,\Delta}$, consider the
case with no additional covariates but with unit-specific intercepts.
Estimate $f_t(0)$ from the control units and denote the resulting
estimate by $\hat f_t(0)$. For each treated unit $i\le n_0$, estimate
$(a_i(0),\lambda_i(0))$ by regressing $Y_{it}$ on
$(1,\hat f_t(0))$ over the pre-treatment period $t\leq T_0$. Similarly,
estimate $(\tilde a_i(1),\lambda_i(1))$ by regressing $Y_{it}$ on
$(1,\hat f_t(0))$ over the post-treatment period $t>T_0$.

Let
\[
\hat\kappa_i:=\hat {\tilde a}_i(1)-\hat a_i(0),
\]
where $\hat \kappa_i$ captures the total intercept change, including both the
structural intercept change and the constant factor shift.\footnote{For notational simplicity, we may simply denote
$\tilde a_i(1)$ by $a_i(1)$ since $\Delta$ is not separately identifiable unless the number of treated units is large.}
 Define
\begin{equation}
\label{eq:tau_hat_constshift_pf}
\hat\tau_{it}^{*,\Delta}
=\hat\kappa_i +
\big[\hat\lambda_i(1)-\hat\lambda_i(0)\big]'\hat f_t(0),
\qquad t>T_0,\ i\le n_0 .
\end{equation}
This estimator is the intercept-augmented version of the baseline
estimator in Proposition \ref{prop:prop-1}. With additional covariates, the same
expression includes $X_{it}'[\hat\beta(1)-\hat\beta(0)]$, provided
that the common covariate coefficients are estimated together with
unit-specific intercepts.

\subsection{Relationship with Synthetic Interventions}

Extending the baseline synthetic control model, \citet{Agarwal_2024}
introduce a synthetic interventions framework designed to address
multiple treatments. Their model for potential outcomes under treatment
status \(d\in\{0,1,2,\ldots,D\}\), with \(D\geq 1\), employs a low-rank
tensor factor model represented as
\[
Y_{it}(d)
=
\sum_{l=1}^{r}u_{tl}v_{il}\lambda_{dl}
+
\varepsilon_{it}(d).
\]
While their focus is on estimating average treatment effects, they impose
a restriction by keeping the factor loadings \(v_{il}\) unaffected by
treatment status. Their model incorporates two types of factors:
time-varying factors \(u_{tl}\), which are not affected by treatment
status, and treatment-specific factors \(\lambda_{dl}\), which are
time-invariant.

In contrast, our potential factor model allows for greater flexibility
by permitting both the common factors and the factor loadings to vary
across treatment status. Specifically, if we consider the \(l\)-th
potential factor \(f_{lt}(d)\) in our setup, we can express the synthetic
interventions restriction as
$
f_{lt}(d)=u_{tl}\lambda_{dl}.
$
Under this restriction, their common-factor strategy can be interpreted
as a special case of our model when \(D=1\).

\section{Estimation and Inference}\label{sec:Estimation-and-Inference}

\subsection{The Intervention Does Not Affect the Factors}

Consider (\ref{eq:estimator-main}) as the estimator for the unit-specific
systematic causal effect in (\ref{eq:sys-causaleffect}):
\[
\hat{\tau}_{it}^{*}=\left[\hat{\lambda}_{i}\left(1\right)-\hat{\lambda}_{i}\left(0\right)\right]^{\prime}\hat{f}_{t}+X_{it}^{\prime}\left[\hat{\beta}\left(1\right)-\hat{\beta}\left(0\right)\right],\;t>T_{0},\;i\leq n_{0}.
\]
The factor estimate $\hat{f}_{t}$ is obtained using panel regression
with interactive fixed effects (see, for example, \citet{Bai_2009}) using
the control units. Then $\hat{\lambda}_{i}\left(0\right)$ and $\hat{\beta}\left(0\right)$
are obtained from a regression of $Y_{it}$ on $\hat{f}_{t}$ and
$X_{it}$ for $t\leq T_{0}$ and $i\leq n_{0}$. $\hat{\lambda}_{i}\left(1\right)$
and $\hat{\beta}\left(1\right)$ are obtained from another regression
of $Y_{it}$ on $\hat{f}_{t}$ and $X_{it}$ for $t>T_{0}$ and
$i\leq n_{0}$. Here we assume $\beta(0)$ and $\beta(1)$ are common across $i$, and only the factor loadings vary with $i$.
So $\hat \lambda_i(0)$ and $\hat \beta(0)$ are estimated by interacting individual dummies with $\hat f_t$ in a pooled regression to impose
common slope coefficients for the covariates.  $\hat \lambda_i(1)$ and $\hat \beta(1)$ are obtained similarly. If the slope coefficients are heterogeneous, simple time series regressions can be applied for  each
$i$ to obtain $\hat \lambda_i(0)$ and $\hat \beta_i(0)$, likewise for $\hat \lambda_i(1)$ and  $\hat \beta_i(1)$.
The large sample theory is summarized in Proposition \ref{prop:prop-1}.

\begin{proposition}
\label{prop:prop-1}
Fix a treated unit $i\le n_0$ and a post-treatment date $t\in\mathcal T_1$.
Suppose that Assumptions~\ref{ass:factor-moments}--\ref{ass:control-score-clt} in the appendix hold
 and that one of the following
three conditions is satisfied: (a) $n_0$ is fixed; (b) $n_0\to\infty$ and
$f_t\neq0$; or (c) $n_0\to\infty$, $f_t=0$, and
\eqref{eq:pooled-slope-rate} holds. Then, as $T_0,T_1$, and $n_1\to \infty$ with
$\sqrt{n_1}/T\to 0$ and $\sqrt{T}/n_1\to 0$, $\widehat{\tau}_{it}^{*}$ is a consistent
estimator of $\tau_{it}^{*}$ for $i\leq n_{0}$ and $t>T_{0}$
and
\[
\widehat V_{it}^{-1/2}\left(\widehat{\tau}_{it}^{*}-\tau_{it}^{*}\right)\Rightarrow N(0,1),
\]
where $\widehat V_{it}$ is an estimator for the variance of $\widehat{\tau}_{it}^{*}$.
\end{proposition}

Thus asymptotic normality holds whether the number of treated units $n_0$ is
fixed or diverges. When $n_0\to\infty$ and $f_t\neq0$, the result imposes no
additional relative growth restriction between $n_0$ and $n_1$ beyond the
maintained assumptions. When $n_0\to\infty$, the knife-edge case $f_t=0$
requires the additional rate condition \eqref{eq:pooled-slope-rate}; see
Remark~\ref{rem:pooled-slope-remainder} below. The result does not require
any proportionality restriction between $T_d$ and $T$; in particular, it does
not require $T_d/T\to c_d\in(0,1)$.

The variance estimator $\widehat V_{it}$ is given in the Appendix. The proof
of Proposition~\ref{prop:prop-1} is provided in the Online Appendix.

As discussed previously, the same result continues to hold when the outcome model includes
unit-specific intercepts. In this case, the common factors can be estimated
from the control units after removing unit-specific time averages from the
outcomes and covariates, as in Section~8 of \citet{Bai_2009}. The normalization
that the factors have zero time mean separates the unit-specific intercepts
from the common factors. After estimating the factors from the demeaned
control-unit panel, we recover the treated-unit coefficients by separately
regressing the treated-unit outcome on a constant, the estimated factors,
and the covariates over the pre-intervention and post-intervention periods.
The feasible treated effect is then defined as the sample analogue of
\eqref{eq:tauit-intercept-model}, with the unknown factors and coefficients
replaced by their estimates. With this modification, Proposition~\ref{prop:prop-1}
remains valid.

The intercept-augmented formulation also covers the constant-shift
specification in \eqref{eq:ft_shift_constant_pf}. When unit-specific
intercepts are included, a time-invariant shift in the factor process is
absorbed into the post-treatment intercept, and Proposition~\ref{prop:prop-1}
continues to apply.

\subsection{The Intervention Affects the Factors for the Treated Group, Large $n_0$}

Consider (\ref{eq:estimator-groupfactor}) as the estimator for the
unit-specific systematic causal effect (\ref{eq:causaleffect-groupfactor}):
\[
\hat{\tau}_{it}^{*}=\hat{\lambda}_{i}\left(1\right)^{\prime}\hat{f}_{t}\left(1\right)-\hat{\lambda}_{i}\left(0\right)^{\prime}\hat{f}_{t}\left(0\right), \qquad \;t>T_{0},\;i\leq n_{0}.
\]
The factor estimate $\hat{f}_{t}\left(0\right)$ is obtained using
the first $r$ leading principal components based on the control units.
Then $\hat{\lambda}_{i}\left(0\right)$ is obtained from a regression
of $Y_{it}$ on $\hat{f}_{t}(0)$ for $t\leq T_{0}$ and $i\leq n_{0}$.
The product $\hat{\lambda}_{i}\left(1\right)^{\prime}\hat{f}_{t}\left(1\right)$
is obtained as the common component estimator from principal component
analysis of $Y_{it}$ for $t>T_{0}$ and $i\leq n_{0}$. We have
the following proposition.

\begin{proposition} \label{prop:prop-2} (large $n_0$)
Under Assumptions~\ref{ass:factor-moments}--\ref{ass:control-score-clt} and  \ref{ass:potential-factors}
 in the appendix, and  $n_{0}/n\rightarrow c\in\left(0,1\right)$,
$T_{0}/T\rightarrow b\in\left(0,1\right)$, $\sqrt{n}/T\rightarrow0$,
and $\sqrt{T}/n\rightarrow0$, then $\widehat{\tau}_{it}^{*}$ is a consistent
estimator of $\tau_{it}^{*}$ for $i\leq n_{0}$ and $t>T_{0}$
and
\[
\widehat V_{it}^{-1/2}\left(\widehat{\tau}_{it}^{*}-\tau_{it}^{*}\right)\Rightarrow N(0,1),
\]
where $\widehat V_{it}$ is an estimator for the variance of $\widehat{\tau}_{it}^{*}$.
\end{proposition}

The variance estimator $\hat V_{it}$ and the proof are given in the online appendix.

\subsection{The Intervention Affects the Factors for the Treated Group, Small $n_0$}

When the number of treated units, $n_0$, is small, a fully flexible
post-treatment factor process $f_t(1)$ cannot be reliably estimated from
the treated post-treatment sample. Under the constant-shift specification
\eqref{eq:ft_shift_constant_pf}, however, the term
$\lambda_i(1)'\Delta$ is absorbed into the post-treatment unit-specific
intercept. Thus the systematic causal effect can be written as
\begin{equation}
\label{eq:tau_overall_infer}
\tau_{it}^{*,\Delta}
=\kappa_i+
\big[\lambda_i(1)-\lambda_i(0)\big]'f_t(0)
+X_{it}'\big[\beta(1)-\beta(0)\big],
\qquad t>T_0,\ i\le n_0 .
\end{equation}

Estimate $\hat f_t(0)$ using the control units as in Proposition~\ref{prop:prop-1}. For each
treated unit $i\le n_0$, estimate the pre- and post-treatment regressions
of $Y_{it}$ on $(1,\hat f_t(0),X_{it})$, allowing the intercept and factor
loadings to vary by unit. Let
\[
\hat\kappa_i=\hat a_i(1)-\hat a_i(0),\qquad
\hat\alpha_i=\hat\lambda_i(1)-\hat\lambda_i(0).
\]
The estimator is
\begin{equation}
\label{eq:tau_hat_overall_infer}
\hat\tau_{it}^{*,\Delta}
=
\hat\kappa_i+\hat\alpha_i'\hat f_t(0)
+X_{it}'\big[\hat\beta(1)-\hat\beta(0)\big].
\end{equation}
This is the intercept-augmented version of the estimator in
Proposition~\ref{prop:prop-1}. Therefore no separate inferential theory is needed for
the constant-shift case. Proposition \ref{prop:prop-1} is applicable to this case.

\section{Monte Carlo Simulations}\label{sec:mcmc}

In this section, we conduct Monte Carlo simulations to evaluate the
finite sample performance of Propositions 1 and 2. We begin with Proposition~\ref{prop:prop-1},
 focusing on the case with a single treated unit ($n_{0}=1$), specifically
the first unit in our sample. Let $n_1=n-n_0$ denote the number of untreated units. For given sample sizes of $n_{1}$,
$T_{0}$, and $T$, we simulate the potential outcomes according to
a factor model with two factors. To simplify the illustration, we
do not include covariates in the simulation. We assume that the idiosyncratic
errors are not indexed by treatment status, so the systematic causal
effect is equivalent to the individual causal effect ($\tau_{it}^{*}=\tau_{it}$).
We then estimate the causal effect for the treated unit using the
estimator in equation (\ref{eq:estimator-main}). We construct the 95\% confidence intervals according to
$\left[\hat{\tau}_{1t}\pm1.96\cdot SE\left(\hat{\tau}_{1t}\right)\right]$,
where the formula for $SE\left(\hat{\tau}_{1t}\right)=\sqrt{\widehat{var}\left(\hat{\tau}_{1t}\right)}$
is provided in the Appendix. Finally, we report the empirical coverage
rates for the true causal effect $\tau_{1t}$ for $t=T_{0}+1,T_{0}+1+m,T$,
under various combinations of $\left(n,T_{0},T\right)$. We define
$m=\left[\frac{T-T_{0}}{2}\right]$ as the nearest integer to $\frac{T-T_{0}}{2}$
so that $T_{0}+1+m$ is positioned in the middle of the treated periods.
The data are simulated according to the following factor model,
\[
Y_{it}\left(d\right)=\lambda_{i}^{\prime}\left(d\right)f_{t}+\varepsilon_{it}, \; d=0,1.
\]
We consider three data generating processes (DGPs):
\begin{itemize}[label=]
\item DGP1. $\left\{ \lambda_{ij}\left(0\right),\lambda_{ij}\left(1\right),\;j=1,2,f_{1t},\varepsilon_{it}\right\} $
are i.i.d.\ $N\left(0,1\right)$.
\item DGP2. $\left\{ \lambda_{ij}\left(0\right),\lambda_{ij}\left(1\right),\;j=1,2,f_{1t}\right\} $
are i.i.d.\ $N\left(1,1\right)$, and $\varepsilon_{it}$ is i.i.d.\ $N\left(0,4\right)$.
\item DGP3. $\left\{ \lambda_{ij}\left(0\right),\lambda_{ij}\left(1\right),\;j=1,2,f_{1t}\right\} $
are i.i.d.\ $N\left(1,1\right)$, and $\varepsilon_{it}$ is i.i.d.\ $\text{Uniform}\left(-2,2\right)$.
\end{itemize}
For DGPs 1-3, the second factor is simulated as $f_{2t}=0.8\cdot f_{2,t-1}+e_{t}$,
where $e_{t}$ is i.i.d.\ $N\left(0,1\right)$.

We set the number of Monte Carlo repetitions to 5000. In Table \ref{tab:cover-prop1},
we report the coverage rates of the 95\% confidence interval for the
three DGPs. The numbers in parentheses represent coverage rates for
$\tau_{1,T_{0}+1},\tau_{1,T_{0}+1+m},\tau_{1,T}$ respectively. We set
the treatment date $T_{0}+1=30$ when $T=50$, and $T_{0}+1=50$ when
$T=100$. For all DGPs and sample sizes, the coverage rates are reasonably close to the
nominal rate, with some undercoverage in the smaller-sample designs.

\begin{table}[!tp]
\begin{centering}
\begin{tabular}{llccc}
\hline
&  & $n_{1}=20$ & $n_{1}=40$ & $n_{1}=100$\tabularnewline
\hline
DGP1 & T=50 & $\left(0.931,0.917,0.919\right)$ & $\left(0.934,0.936,0.939\right)$ & $\left(0.940,0.935,0.941\right)$\tabularnewline
& T=100 & $\left(0.912,0.915,0.916\right)$ & $\left(0.938,0.935,0.932\right)$ & $\left(0.940,0.942,0.948\right)$\tabularnewline
\hline
DGP2 & T=50 & $\left(0.925,0.926,0.924\right)$ & $\left(0.940,0.939,0.937\right)$ & $\left(0.938,0.940,0.942\right)$\tabularnewline
& T=100 & $\left(0.922,0.911,0.911\right)$ & $\left(0.939,0.942,0.937\right)$ & $\left(0.949,0.941,0.946\right)$\tabularnewline
\hline
DGP3 & T=50 & $\left(0.927,0.921,0.926\right)$ & $\left(0.944,0.943,0.940\right)$ & $\left(0.934,0.935,0.935\right)$\tabularnewline
& T=100 & $\left(0.915,0.917,0.918\right)$ & $\left(0.936,0.942,0.933\right)$ & $\left(0.941,0.949,0.943\right)$\tabularnewline
\hline
\end{tabular}
\end{centering}
\caption{Coverage rates of the 95\% confidence interval under Proposition~\ref{prop:prop-1}  with $n_0=1$}\label{tab:cover-prop1}

\end{table}

Next, we apply the same procedure to examine Proposition \ref{prop:prop-2}, varying
the group sizes for treated units ($n_{0}$) and untreated units ($n_{1}$).
For given sample sizes of $n_{0}$, $n_{1}$, $T_{0}$, $T$, we simulate
potential outcomes according to the factor model with two factors.
Again, we assume that the idiosyncratic errors are not indexed by
treatment status, so the systematic causal effect is the same as the
individual causal effect ($\tau_{it}^{*}=\tau_{it}$). We report estimates
only for the first treated unit ($i=1$). The 95\% confidence intervals
are constructed as $\left[\hat{\tau}_{1t}\pm1.96\cdot SE\left(\hat{\tau}_{1t}\right)\right]$,
where the formula for $SE\left(\hat{\tau}_{1t}\right)=\sqrt{\widehat{var}\left(\hat{\tau}_{1t}\right)}$
is provided in the Appendix. We then report the empirical coverage rates
for the true causal effect $\tau_{1t}$ for $t=T_{0}+1,T_{0}+1+m,T$,
under different combinations of $\left(n_{0},n_{1},T_{0},T\right)$.
We specify $m=\left[\frac{T-T_{0}}{2}\right]$ as the nearest integer
to $\frac{T-T_{0}}{2}$. The data are simulated according to the following
factor model with potential factors,
\[
Y_{it}\left(d\right)=\lambda_{i}^{\prime}\left(d\right)f_{t}\left(d\right)+\varepsilon_{it},\;d=0,1.
\]
We consider three additional DGPs:
\begin{itemize}[label=]
\item DGP4. $\left\{ \lambda_{ij}\left(0\right),\lambda_{ij}\left(1\right),\;j=1,2,f_{1t}\left(0\right),f_{1t}\left(1\right),\varepsilon_{it}\right\} $
are i.i.d.\ $N\left(0,1\right)$.
\item DGP5. $\left\{ \lambda_{ij}\left(0\right),\lambda_{ij}\left(1\right),\;j=1,2,f_{1t}\left(0\right),f_{1t}\left(1\right)\right\} $
are i.i.d.\ $N\left(1,1\right)$, and $\varepsilon_{it}$ is i.i.d.
$N\left(0,4\right)$.
\item DGP6. $\left\{ \lambda_{ij}\left(0\right),\lambda_{ij}\left(1\right),\;j=1,2,f_{1t}\left(0\right),f_{1t}\left(1\right)\right\} $
are i.i.d.\ $N\left(1,1\right)$, and $\varepsilon_{it}$ is i.i.d.\
$\text{Uniform}\left(-2,2\right)$.
\end{itemize}
For DGPs 4-6, the second potential factors are simulated as $f_{2t}\left(0\right)=0.8\cdot f_{2,t-1}\left(0\right)+e_{1t}$
with $e_{1t}$ being i.i.d.\ $N\left(0,1\right)$, and $f_{2t}\left(1\right)=0.9\cdot f_{2,t-1}\left(1\right)+e_{2t}$
with $e_{2t}$ being i.i.d.\ $N\left(0,1\right)$.

The number of Monte Carlo repetitions is again set to 5000. In Table
\ref{tab:coverProp2}, we report the coverage rates of the 95\% confidence
interval for three choices of the number of treated units ($n_{0}$):
20, 40, 100. The numbers in parentheses in the tables represent coverage
rates for $\tau_{1,T_{0}+1},\tau_{1,T_{0}+1+m},\tau_{1,T}$ respectively.
We set the treatment date $T_{0}+1=30$ when $T=50$, and $T_{0}+1=50$
when $T=100$. For all DGPs and sample sizes, the coverage rates of
the confidence intervals remain close to the nominal rate of 95\%.

\begin{table}[!tp]
\begin{centering}
\begin{tabular}{llccc}
\hline
\multicolumn{5}{c}{ $n_0=20$}\tabularnewline
\hline
&  & $n_{1}=20$ & $n_{1}=40$ & $n_{1}=100$\tabularnewline
\hline
DGP4 & T=50 & $\left(0.929,0.931,0.928\right)$ & $\left(0.932,0.928,0.930\right)$ & $\left(0.925,0.929,0.920\right)$\tabularnewline
& T=100 & $\left(0.936,0.935,0.934\right)$ & $\left(0.934,0.935,0.939\right)$ & $\left(0.932,0.929,0.935\right)$\tabularnewline
\hline
DGP5 & T=50 & $\left(0.925,0.924,0.928\right)$ & $\left(0.928,0.928,0.927\right)$ & $\left(0.930,0.928,0.926\right)$\tabularnewline
& T=100 & $\left(0.931,0.938,0.941\right)$ & $\left(0.930,0.934,0.936\right)$ & $\left(0.937,0.933,0.935\right)$\tabularnewline
\hline
DGP6 & T=50 & $\left(0.925,0.936,0.925\right)$ & $\left(0.931,0.933,0.926\right)$ & $\left(0.922,0.926,0.926\right)$\tabularnewline
& T=100 & $\left(0.935,0.934,0.927\right)$ & $\left(0.936,0.935,0.937\right)$ & $\left(0.933,0.941,0.930\right)$\tabularnewline
\hline
\multicolumn{5}{c}{ $n_0=40$}\tabularnewline
\hline
&  & $n_{1}=20$ & $n_{1}=40$ & $n_{1}=100$\tabularnewline
\hline
DGP4 & T=50 & $\left(0.935,0.934,0.930\right)$ & $\left(0.932,0.939,0.942\right)$ & $\left(0.932,0.933,0.932\right)$\tabularnewline
& T=100 & $\left(0.936,0.935,0.934\right)$ & $\left(0.939,0.943,0.936\right)$ & $\left(0.941,0.944,0.942\right)$\tabularnewline
\hline
DGP5 & T=50 & $\left(0.927,0.926,0.927\right)$ & $\left(0.934,0.936,0.930\right)$ & $\left(0.932,0.937,0.936\right)$\tabularnewline
& T=100 & $\left(0.933,0.934,0.934\right)$ & $\left(0.938,0.943,0.941\right)$ & $\left(0.939,0.944,0.943\right)$\tabularnewline
\hline
DGP6 & T=50 & $\left(0.934,0.932,0.926\right)$ & $\left(0.933,0.938,0.935\right)$ & $\left(0.933,0.933,0.939\right)$\tabularnewline
& T=100 & $\left(0.936,0.934,0.932\right)$ & $\left(0.940,0.947,0.938\right)$ & $\left(0.942,0.941,0.938\right)$\tabularnewline
\hline
\multicolumn{5}{c}{ $n_0=100$}\tabularnewline
\hline
&  & $n_{1}=20$ & $n_{1}=40$ & $n_{1}=100$\tabularnewline
\hline
DGP4 & T=50 & $\left(0.928,0.930,0.926\right)$ & $\left(0.937,0.933,0.933\right)$ & $\left(0.932,0.936,0.936\right)$\tabularnewline
& T=100 & $\left(0.937,0.929,0.934\right)$ & $\left(0.942,0.944,0.942\right)$ & $\left(0.942,0.945,0.948\right)$\tabularnewline
\hline
DGP5 & T=50 & $\left(0.927,0.930,0.932\right)$ & $\left(0.935,0.937,0.939\right)$ & $\left(0.935,0.943,0.938\right)$\tabularnewline
& T=100 & $\left(0.936,0.935,0.936\right)$ & $\left(0.941,0.940,0.947\right)$ & $\left(0.944,0.940,0.940\right)$\tabularnewline
\hline
DGP6 & T=50 & $\left(0.929,0.928,0.932\right)$ & $\left(0.940,0.934,0.940\right)$ & $\left(0.932,0.936,0.930\right)$\tabularnewline
& T=100 & $\left(0.932,0.937,0.939\right)$ & $\left(0.946,0.942,0.940\right)$ & $\left(0.943,0.947,0.941\right)$\tabularnewline
\hline
\end{tabular}
\end{centering}
\caption{Coverage rates of the 95\% confidence interval under Proposition~\ref{prop:prop-2}}\label{tab:coverProp2}
\end{table}

We next examine the constant-shift factor specification for the case of a small number of treated units.
We focus on the case with a single treated unit ($n_{0}=1$), which
is specified as the first unit in the sample. The DGPs 7-9 extend
DGPs 1-3 under Proposition~\ref{prop:prop-1} to the constant-shift potential-factor case. The
data are simulated according to the following factor model,
\[
Y_{it}\left(d\right)=\lambda_{i}^{\prime}\left(d\right)f_{t}\left(d\right)+\varepsilon_{it},\;d=0,1,
\]
where
\begin{equation}
f_{t}\left(1\right)=f_{t}\left(0\right)+\Delta,\label{eq:potential-factor-shift}
\end{equation}
with $\Delta$ being a constant $r\times1$ vector. We consider three
data generating processes (DGPs):
\begin{itemize}[label=]
\item DGP7. $\left\{ \lambda_{ij}\left(0\right),\lambda_{ij}\left(1\right),\;j=1,2,f_{1t}\left(0\right),\varepsilon_{it}\right\} $
are i.i.d.\ $N\left(0,1\right)$.
\item DGP8. $\left\{ \lambda_{ij}\left(0\right),\lambda_{ij}\left(1\right),\;j=1,2,f_{1t}\left(0\right)\right\} $
are i.i.d.\ $N\left(1,1\right)$, and $\varepsilon_{it}$ is i.i.d.\ $N\left(0,4\right)$.
\item DGP9. $\left\{ \lambda_{ij}\left(0\right),\lambda_{ij}\left(1\right),\;j=1,2,f_{1t}\left(0\right)\right\} $
are i.i.d.\ $N\left(1,1\right)$, and $\varepsilon_{it}$ is i.i.d.\ $\text{Uniform}\left(-2,2\right)$.
\end{itemize}
The second factor is simulated as $f_{2t}\left(0\right)=0.8\cdot f_{2,t-1}\left(0\right)+e_{t}$,
where $e_{t}$ is i.i.d.\ $N\left(0,1\right)$. The potential factor
$f_{t}\left(1\right)$ is given by (\ref{eq:potential-factor-shift}). Across Monte Carlo simulations, elements of
$\Delta$ are random draws from $N\left(0,1\right)$.

We set the number of Monte Carlo repetitions to 5000. In Table \ref{tab:cover-constshift},
we report the coverage rates of the 95\% confidence interval for the
three DGPs 7-9. The numbers in parentheses represent coverage rates
for $\tau_{1,T_{0}+1},\tau_{1,T_{0}+1+m},\tau_{1,T}$ respectively, where
$m=\left[\frac{T-T_{0}}{2}\right]$. We set the treatment date $T_{0}+1=30$
when $T=50$, and $T_{0}+1=50$ when $T=100$. For all DGPs and sample
sizes, the coverage rates of the confidence intervals are all close
to the nominal rate of 95\%.

\begin{table}[!tp]
\begin{centering}
\begin{tabular}{llccc}
\hline
&  & $n_{1}=20$ & $n_{1}=40$ & $n_{1}=100$\tabularnewline
\hline
DGP7 & T=50 & $\left(0.929,0.931,0.929\right)$ & $\left(0.942,0.936,0.943\right)$ & $\left(0.939,0.941,0.937\right)$\tabularnewline
& T=100 & $\left(0.921,0.924,0.924\right)$ & $\left(0.942,0.937,0.943\right)$ & $\left(0.946,0.946,0.946\right)$\tabularnewline
\hline
DGP8 & T=50 & $\left(0.932,0.928,0.929\right)$ & $\left(0.941,0.944,0.942\right)$ & $\left(0.941,0.938,0.935\right)$\tabularnewline
& T=100 & $\left(0.920,0.921,0.931\right)$ & $\left(0.946,0.941,0.944\right)$ & $\left(0.943,0.947,0.942\right)$\tabularnewline
\hline
DGP9 & T=50 & $\left(0.932,0.930,0.934\right)$ & $\left(0.933,0.935,0.935\right)$ & $\left(0.941,0.935,0.939\right)$\tabularnewline
& T=100 & $\left(0.927,0.925,0.927\right)$ & $\left(0.946,0.937,0.932\right)$ & $\left(0.944,0.942,0.942\right)$\tabularnewline
\hline
\end{tabular}
\end{centering}
\caption{Coverage rates of the 95\% confidence interval under the constant-shift factor specification ($n_0=1$).}\label{tab:cover-constshift}
\end{table}

\section{Two Empirical Applications}\label{sec:applications}

\subsection{An Application to California's Tobacco-Control Program}

Using data from \citet{Abadie_Diamond_Hainmeuller_2010} on per capita cigarette sales
across 39 U.S.\ states, we construct a counterfactual California using
factor models and compare it with the synthetic California. We find
that the two approaches yield close results. Let $Y_C$ be the $n\times T$
data matrix for the control states ($n=38$). Let $S_{C}$ denote
the $n\times n$ sample covariance matrix for $Y_C$. Let $\mu_{j}$
be the $j$-th largest eigenvalue of $S_{C}$. Let $m=\min\left\{ n,T\right\} $.
Figure \ref{fig:CA-screeplot} provides the scree plot for the control
states, which plots $\frac{\mu_{j}}{\sum_{i=1}^{m}\mu_{i}}$ against
$j=1,2,\ldots,10$ (i.e., the ratio of the first ten largest eigenvalues
of the covariance matrix and the sum of all eigenvalues). Figure \ref{fig:CA-er-gr}
plots \citet{Ahn_Horenstein_2013}'s Eigenvalue Ratio (ER) and Growth
Ratio (GR) criterion functions, where
\[
ER\left(j\right)=\frac{\mu_{j}}{\mu_{j+1}},\;GR\left(j\right)=\frac{\log\left(\sum_{i=j}^{m}\mu_{i}/\sum_{s=j+1}^{m}\mu_{s}\right)}{\log\left(\sum_{s=j+1}^{m}\mu_{s}/\sum_{k=j+2}^{m}\mu_{k}\right)}.
\]
The number of factors can be consistently estimated by
\[
\hat{r}^{ER}=\underset{\left\{ 1\leq j\leq r_{max}\right\} }{\argmax}ER\left(j\right),\;\hat{r}^{GR}=\underset{\left\{ 1\leq j\leq r_{max}\right\} }{\argmax}GR\left(j\right).
\]

Both the $ER$ and $GR$ criteria select a single factor. The number of
factors can also be chosen using the information criteria of \citet{Bai_Ng_2002}, which suggest more than one factor in this application. As a
robustness check, we therefore estimate the model using both one factor
and two factors.

\begin{figure}[!tp]
\begin{centering}
\includegraphics[width=0.7\columnwidth]{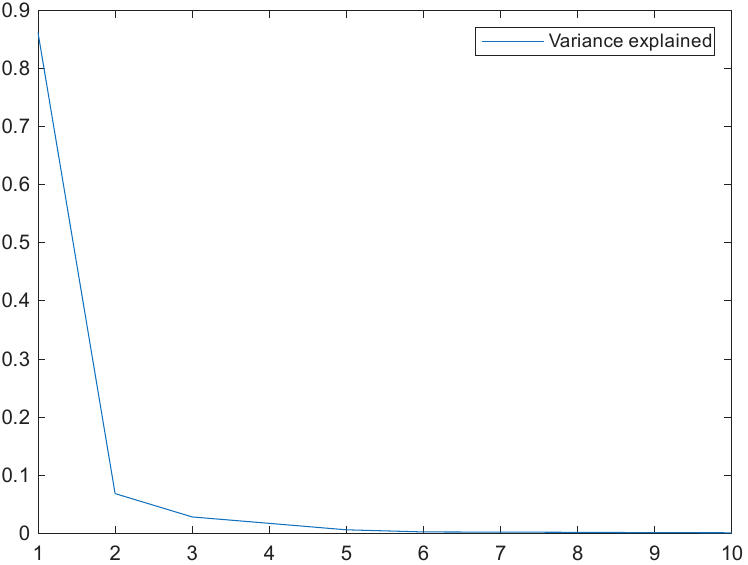}
\par\end{centering}
\caption{Total variance explained by the first 10 principal components.}\label{fig:CA-screeplot}
\end{figure}

\begin{figure}[!tp]
\begin{centering}
\includegraphics[width=0.7\columnwidth]{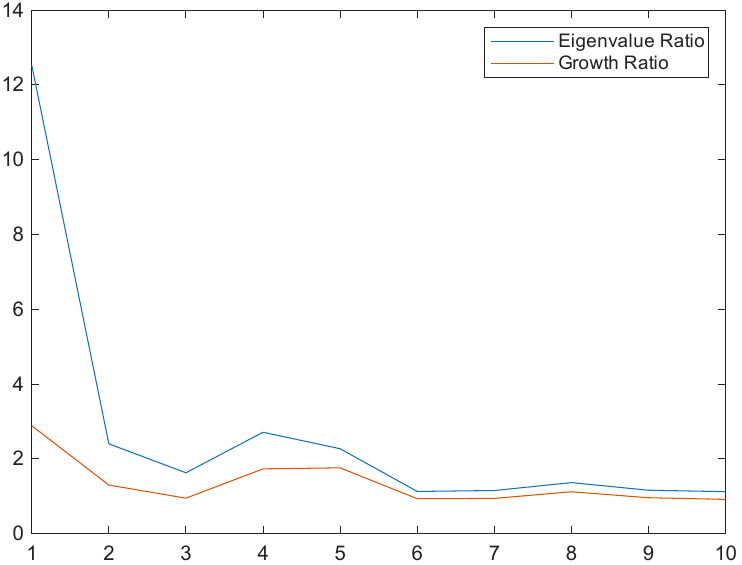}
\par\end{centering}
\caption{The eigenvalue ratio and growth ratio of \citet{Ahn_Horenstein_2013}.}\label{fig:CA-er-gr}
\end{figure}

Our analysis proceeds with the following steps.
\begin{itemize}[label=]
\item Step 1: use principal component analysis for the covariance of 38
control states to obtain the factor estimates $\hat{f}_{t}$,
$1\leq t\leq T$.
\item Step 2: regress $Y_{it}$ on $\left(1,\hat{f}_{t}\right)$ for $t\leq T_{0}$
to obtain the intercept $\hat{a}_{i}\left(0\right)$ and the factor
loading $\hat{\lambda}_{i}\left(0\right)$, $i=CA$. Regress $Y_{it}$
on $\left(1,\hat{f}_{t}\right)$ for $t>T_{0}$ to obtain $\hat{a}_{i}\left(1\right)$
and $\hat{\lambda}_{i}\left(1\right)$, $i=CA$.
\item Step 3: the estimator for the (systematic) causal effect is
\[
\hat{\tau}_{CA,t}=\left[\hat{\lambda}_{CA}\left(1\right)-\hat{\lambda}_{CA}\left(0\right)\right]'\hat{f}_{t}+\left[\hat{a}_{CA}\left(1\right)-\hat{a}_{CA}\left(0\right)\right],\;t>T_{0}.
\]
\end{itemize}
For a robustness check, we estimate the model using either one or two
factors. The t-statistic for testing the null hypothesis
$\kappa_{CA}=a_{CA}\left(1\right)-a_{CA}\left(0\right)=0$ is $t=1.38$
for the model with a single factor and $t=0.10$ for the model with two
factors. This provides little evidence of a post-treatment intercept
shift, including any constant shift in the factor process. Our estimator
$\hat{\tau}_{CA,t}$ provides an estimator for the causal effect that is
robust to such post-treatment intercept shifts.

In Figure~\ref{fig:CA-1or2factors}, we provide an illustrative comparison
of the observed California series with estimated untreated paths obtained
from factor models with one or two factors and from the synthetic-control
method. The vertical line marks 1988, the year in which Proposition 99
was passed in California. The treatment period runs from $T_0+1=1989$ to
$T=2000$. For the factor-model estimates, the predicted systematic
untreated component is
$\widehat m_{CA,t}(0)=\hat a_{CA}(0)+\hat\lambda_{CA}(0)'\hat f_t$.
For demonstration, the path plotted in the figure adds the
post-treatment treated residual, so that
$\widehat Y_{CA,t}(0)=\widehat m_{CA,t}(0)+\hat\varepsilon_{CA,t}$,
where $\hat\varepsilon_{CA,t}$ is obtained in Step~2. Under the
assumption that the idiosyncratic residual is not affected by treatment,
this object estimates the untreated potential outcome $Y_{CA,t}(0)$.
Unlike the synthetic-control comparison, however, our Step~3
causal-effect estimator is constructed directly and does not require
first constructing this plotted counterfactual path. In the figure legend,
``1factorCA'' denotes the one-factor estimate, whereas ``2factorCA''
denotes the two-factor estimate.

Figure \ref{fig:CA-causal-1or2factors} presents a comparison between the synthetic control approach and the one- or two-factor causal estimates computed in Step 3. Overall, the two methods produce comparable results.

\begin{figure}[!tp]
\begin{centering}
\includegraphics[width=0.7\columnwidth]{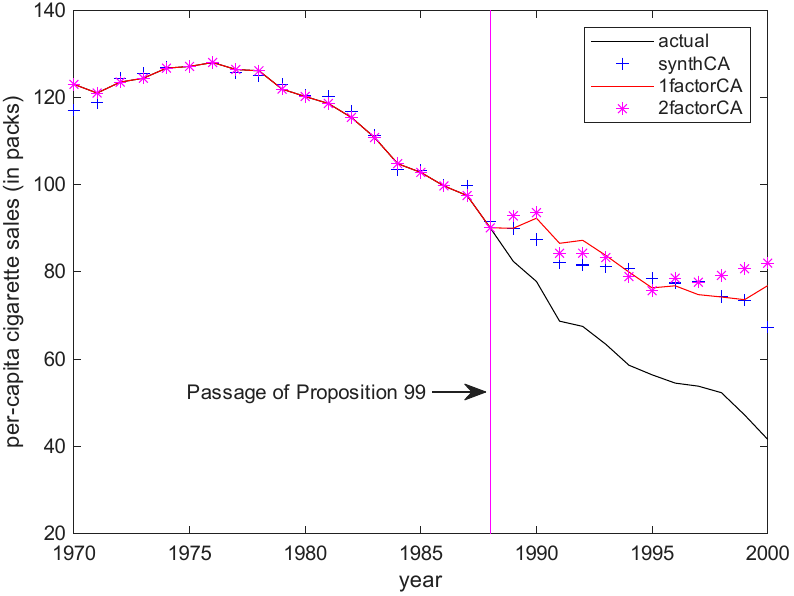}
\par\end{centering}
\caption{Counterfactual California using 1 or 2 factors vs.\ synthetic control}\label{fig:CA-1or2factors}
\end{figure}

\begin{figure}[!tp]
\begin{centering}
\includegraphics[width=0.7\columnwidth]{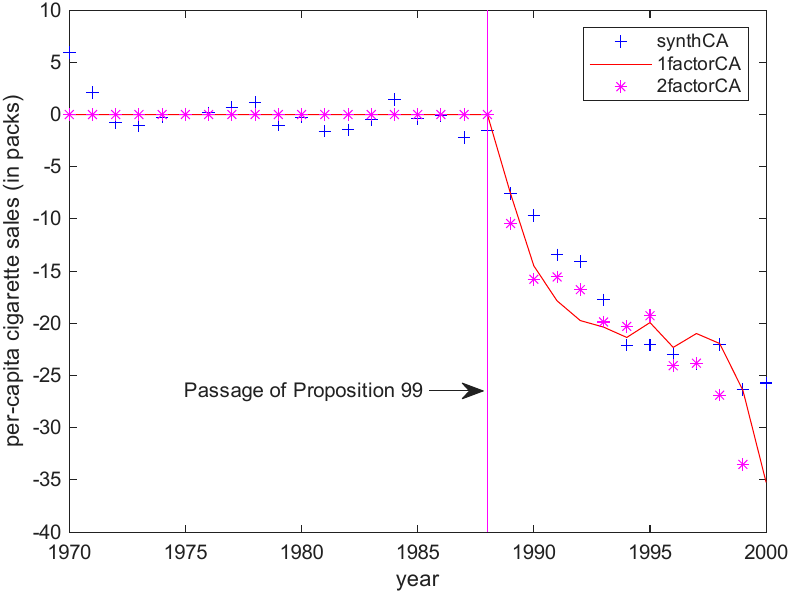}
\par\end{centering}
\caption{Causal estimates using 1 or 2 factors vs.\ synthetic control}\label{fig:CA-causal-1or2factors}
\end{figure}

We further investigate whether the policy intervention indeed induced
a structural break by regressing California's observed outcome on
the single factor using the whole sample. The Quandt likelihood ratio test for a structural break at an unknown date, using 15 percent trimming, rejects parameter stability with a p-value below 0.001. The maximum F-statistic is attained in 1984.
 The Chow test for a structural break at the 1989
intervention date also rejects stability,  yielding an F-statistic of 16.84 with a p-value below 0.001.
These tests provide evidence of instability in the treated-unit factor regression, including instability at the intervention date.
We interpret these tests as diagnostics supporting the relevance of allowing
treated factor loadings to change over time.

Applying the results from Section \ref{sec:Estimation-and-Inference},
we also construct the 95\% confidence interval of our causal estimates
based on the factor models. Using a single factor, Figure \ref{fig:Factor-causal-model-CA1-causal-ci}
reproduces Figure \ref{fig:CA-causal-1or2factors} with the shaded
region being the 95\% confidence interval around the causal estimates.
The confidence intervals generally cover the synthetic-control estimates and
indicate that the factor-model causal estimates are statistically different
from zero for the post-treatment periods shown.
Figure \ref{fig:Factor-causal-model-CA1-causal-ci-2} reports similar
results using 2 factors.

\begin{figure}[!tp]
\begin{centering}
\includegraphics[width=0.7\columnwidth]{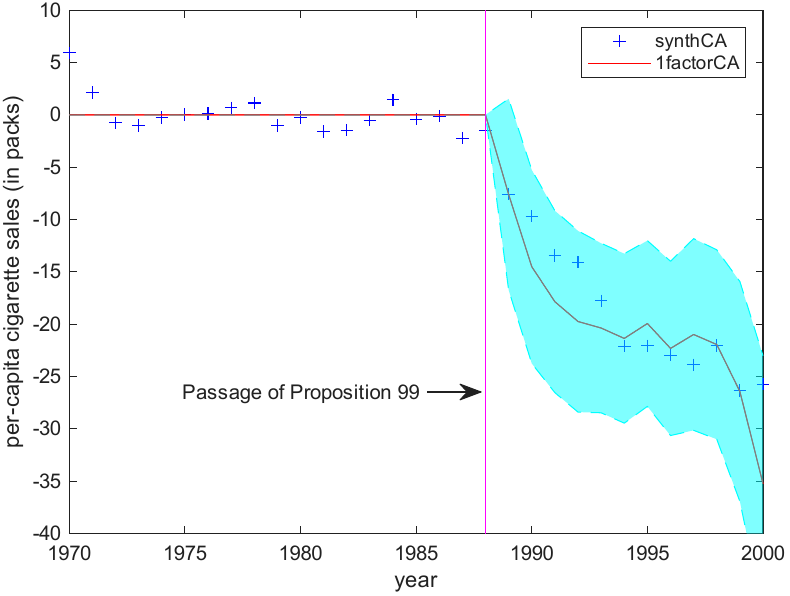}
\par\end{centering}
\caption{The 95\% confidence intervals for factor causal estimates}\label{fig:Factor-causal-model-CA1-causal-ci}
\end{figure}

\begin{figure}[!tp]
\begin{centering}
\includegraphics[width=0.7\columnwidth]{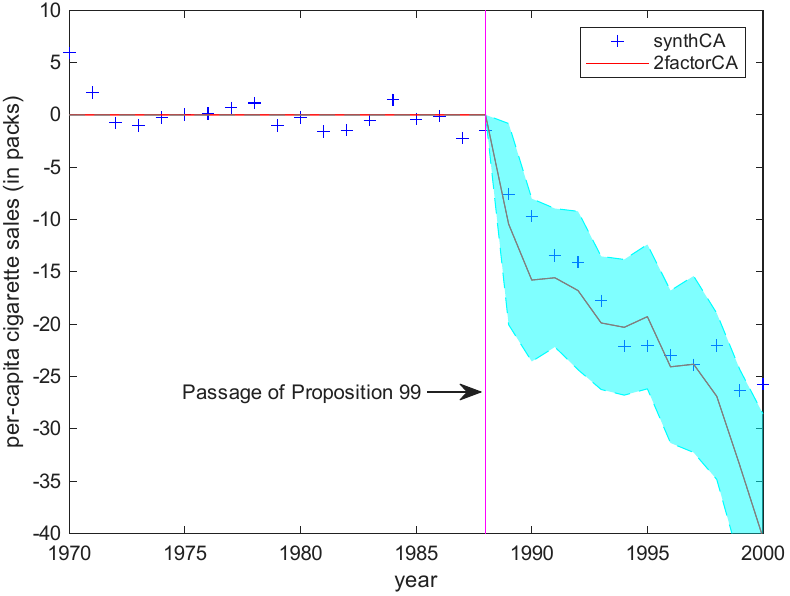}
\par\end{centering}
\caption{The 95\% confidence intervals for factor causal estimates}\label{fig:Factor-causal-model-CA1-causal-ci-2}
\end{figure}

\subsection{An Application to German Reunification}

In this section, we compare the causal estimates from the synthetic
control method and causal factor models using \citet{Abadie_Diamond_Hainmeuller_2015}'s
data on per capita GDP for 17 countries. The objective is to evaluate
the causal impact of German reunification on Germany's per capita
GDP. The synthetic control method uses 16 countries to construct the
synthetic West Germany. The causal factor model uses the same 16 countries
to construct the counterfactual West Germany. Based on the covariance
matrix of all countries excluding West Germany, \citet{Ahn_Horenstein_2013}'s Eigenvalue Ratio and Growth Ratio criteria point to a single
factor.

In Figure \ref{fig:gdp-all}, time series of per capita GDP for all
17 countries (1969--2003) demonstrate strong comovement but not necessarily
parallel trends. The vertical line represents the year 1990 and the treatment
periods are from  $T_{0}+1=1991$ to $T=2003$.

\begin{figure}[!tp]
\begin{centering}
\includegraphics[width=0.7\columnwidth]{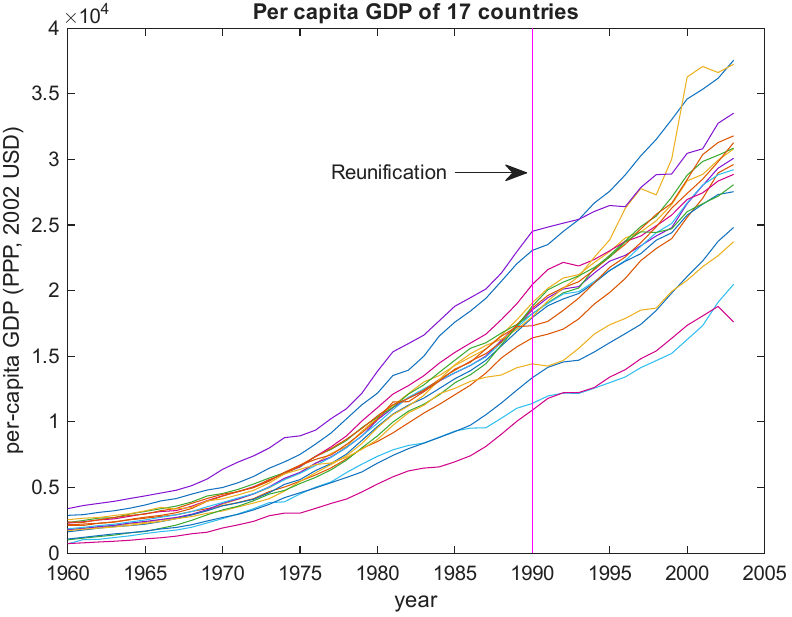}
\par\end{centering}
\caption{Time series plot of per capita GDP}\label{fig:gdp-all}
\end{figure}

For a robustness check, we estimate the causal factor model with one
or two factors. In Figure \ref{fig:WG-1or2factors}, we compare the actual data with counterfactual paths for West Germany
constructed using factor models with one or two factors and using synthetic
control.  The factor causal
effects remain similar to the synthetic control estimates as
shown by Figures \ref{fig:WG-1or2factors} and \ref{fig:WG-causal-1or2factors}.

\begin{figure}[!tp]
\begin{centering}
\includegraphics[width=0.7\columnwidth]{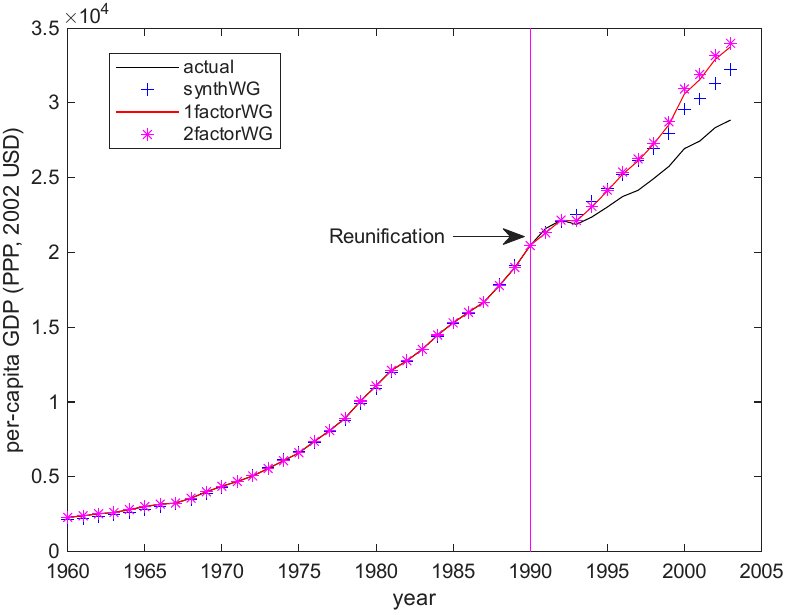}
\par\end{centering}
\caption{Causal factor model versus synthetic control: the counterfactual West
Germany}\label{fig:WG-1or2factors}
\end{figure}

\begin{figure}[!tp]
\begin{centering}
\includegraphics[width=0.7\columnwidth]{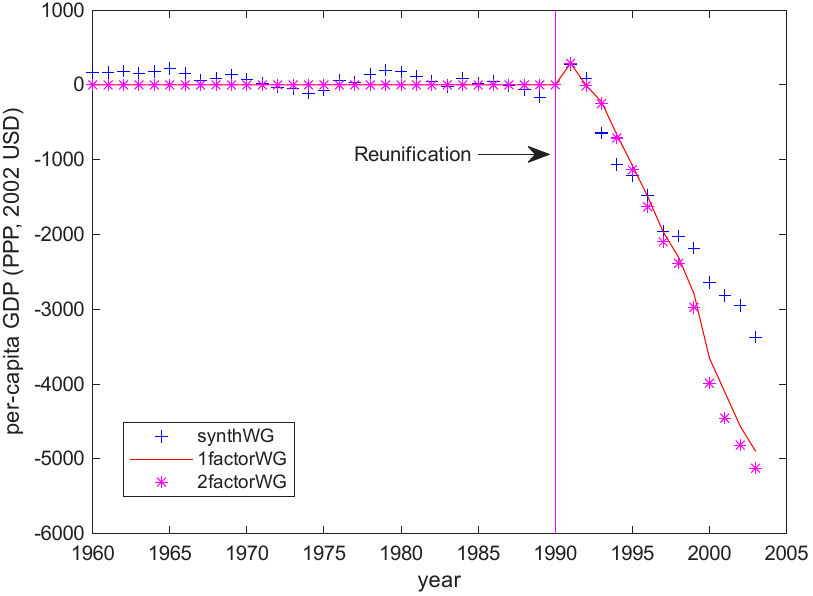}
\par\end{centering}
\caption{Causal factor model versus synthetic control: the causal effects}\label{fig:WG-causal-1or2factors}
\end{figure}

We then regress West Germany's observed outcome on the single factor
using the full sample. The Quandt likelihood ratio test for a structural break at an unknown date,
using 15 percent trimming, rejects parameter stability with a p-value below
0.001. The maximum F-statistic is attained in 1993. The Chow test for a
structural break at the 1991 treatment date also rejects stability, producing
an F-statistic of 634.5 with a p-value below 0.001.

Figures \ref{fig:germany-causal-ci-1factor} and \ref{fig:germany-causal-ci-2factor}
present the 95\% confidence intervals of the causal effects using
one or two factors.
These intervals indicate that our causal estimates
are statistically different from zero for most of the post-treatment periods shown.
The path of our estimates
aligns closely with the synthetic control estimates prior to 2000,
beyond which slight deviations emerge. In the 2-factor specification,
the 95\% confidence intervals generally cover the synthetic control
estimates.

Finally, we investigate the possibility of a post-treatment intercept
shift for West Germany. The t-statistic for testing the null hypothesis
$\kappa_{WG}=a_{WG}\left(1\right)-a_{WG}\left(0\right)=0$ is $t=11.81$
for the specification with a single factor and $t=6.23$ for the
specification with two factors. This suggests a statistically significant post-treatment intercept shift,
which may reflect either a direct change in the unit-specific intercept or a
constant shift in the factor process.

\begin{figure}[!tp]
\begin{centering}
\includegraphics[width=0.7\columnwidth]{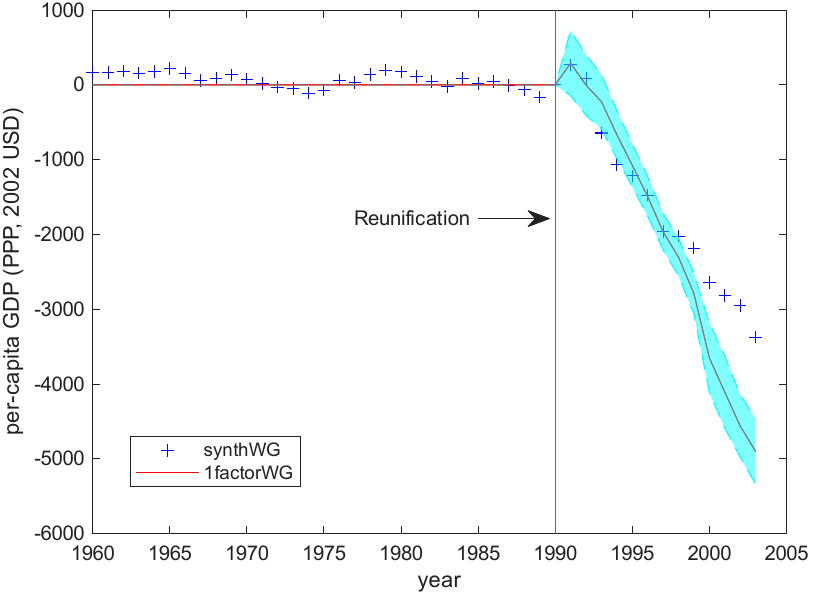}
\par\end{centering}
\caption{The 95\% confidence intervals for causal estimates with one factor}\label{fig:germany-causal-ci-1factor}
\end{figure}

\begin{figure}[!tp]
\begin{centering}
\includegraphics[width=0.7\columnwidth]{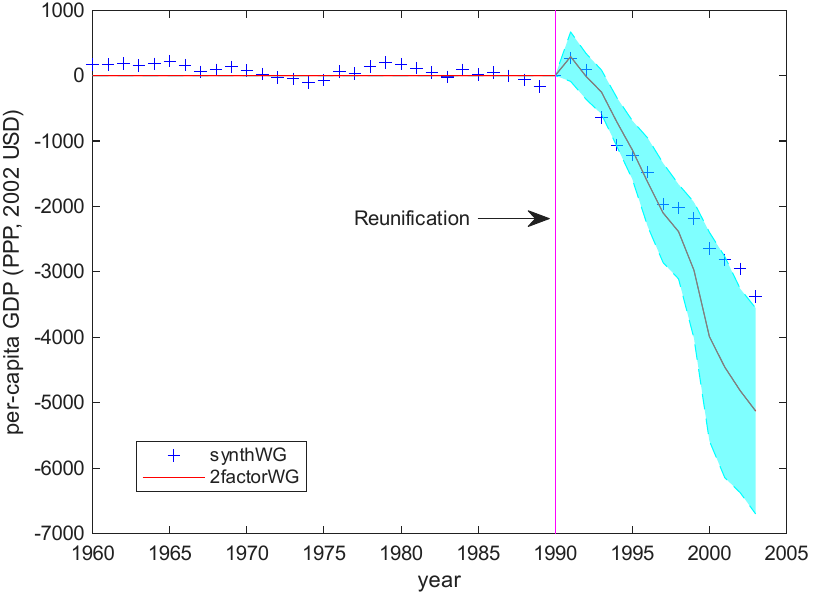}
\par\end{centering}
\caption{The 95\% confidence intervals for causal estimates with two factors}\label{fig:germany-causal-ci-2factor}
\end{figure}

\section{Conclusion}\label{sec:conclusion}

This paper develops a factor-model framework for causal inference in
panel data that links systematic treatment effects to changes in the
factor representation of potential outcomes. By modeling both potential
outcomes within a factor structure, the approach separates systematic
effects from unit-time idiosyncratic noise, provides interpretable
decompositions, and delivers feasible inference without imposing the
standard parallel-trends restriction.

We consider two broad environments. In the first, the policy intervention
does not alter the common-shock process, so the same factors drive both
treated and control outcomes. Treatment effects arise through structural
changes in treated units' factor loadings and, possibly, slope
coefficients. In this benchmark case, the common factors can be estimated
from the control group and treated parameters are identified from
time-series variation. The method can therefore be applied even when the
treated group is small, including the single-treated-unit setting.

In the second environment, the policy intervention may also affect the
common-shock process. Allowing both potential factors and potential
factor loadings provides a flexible representation of post-treatment
dynamics. When the treated cross section is large, the framework can
accommodate an unrestricted post-treatment factor process and estimate
treated-specific post-treatment factors directly. When the treated group
is small, a fully flexible post-treatment factor process is generally not
separately identifiable; we therefore propose a restricted
post-treatment perturbation that preserves the potential-factor
interpretation while restoring tractable estimation and inference.

Monte Carlo experiments indicate that the proposed confidence intervals
attain coverage close to nominal levels across a range of designs. In
applications to California's Proposition 99 and German reunification, the
causal factor model produces counterfactual paths and treatment-effect
estimates broadly consistent with synthetic control, while offering
additional structure for formal inference and interpretable diagnostics,
including evidence on structural breaks in treated loadings. Overall, the
framework complements synthetic control and difference-in-differences
methods by providing a unified factor-based representation of potential
outcomes and inference for systematic effects in panels with pervasive
comovement.


\begin{appendix}
	
	\section{Proof of Proposition 1}
	
	\subsection{Notation used in the proof}
	
	The factors are estimated using only the control units.  Thus the cross-sectional
	sample size in the factor-estimation problem is
	\begin{equation}
		n_1=n-n_0,
		\label{eq:control-sample-size}
	\end{equation}
	whereas the time dimension is the full sample length $T$.  The sub-block length
	in the improved factor-error note should be read here as the length of the
	corresponding pre- or post-treatment block.  Write
	\begin{equation}
		\cT_0=\{1,\ldots,T_0\},
		\qquad
		\cT_1=\{T_0+1,\ldots,T\},
		\label{eq:time-blocks}
	\end{equation}
	and let
	\begin{equation}
		T_d=|\cT_d|,
		\qquad d\in\{0,1\}.
		\label{eq:block-lengths}
	\end{equation}
	Thus $T_0$ is the pre-treatment block length and $T_1=T-T_0$ is the
	post-treatment block length.
	
	For each block $d\in\{0,1\}$, define the true and estimated factor matrices
	\begin{equation}
		F_d=(f_t')_{t\in\cT_d},
		\qquad
		\widehat F_d=(\widehat f_t')_{t\in\cT_d}.
		\label{eq:Fd-hatFd-def}
	\end{equation}
	All factor errors below are understood after applying the population rotation, or
	equivalently after using the local nonsingular normalization from the improved
	factor-rotation argument.  Thus
	\begin{equation}
		\Delta=\widehat F-F,
		\qquad
		\Delta_d=\widehat F_d-F_d,
		\label{eq:Delta-def}
	\end{equation}
	means the rotated factor-estimation error.
	
	Let the treated units be indexed by $i=1,\ldots,n_0$ and the control units by
	$i=n_0+1,\ldots,n$, and write $n_1=n-n_0$.  For treated units, suppose that in
	block $d\in\{0,1\}$,
	\begin{equation}
		Y_{it}(d)=\lambda_i(d)'f_t+X_{it}'\beta(d)+\varepsilon_{it}(d),
		\qquad t\in\cT_d,
		\label{eq:treated-potential-model}
	\end{equation}
	where $\lambda_i(d)\in\R^r$ is unit specific and $\beta(d)\in\R^p$ is common
	across treated units.
	
	For treated unit $i\le n_0$ and block $d \in \{0,1\}$, define
	\begin{equation}
		Y_{i,d}=(Y_{it})_{t\in\cT_d},
		\qquad
		X_{i,d}=(X_{it}')_{t\in\cT_d},
		\qquad
		\varepsilon_{i,d}=(\varepsilon_{it}(d))_{t\in\cT_d}.
		\label{eq:block-data-def}
	\end{equation}
	The true block model is
	\begin{equation}
		Y_{i,d}=F_d\lambda_i(d)+X_{i,d}\beta(d)+\varepsilon_{i,d}.
		\label{eq:block-model}
	\end{equation}
	Stack the treated observations as
	\[
	Y_d
	=
	\begin{pmatrix}
		Y_{1,d}\\
		Y_{2,d}\\
		\vdots\\
		Y_{n_0,d}
	\end{pmatrix},
	\qquad
	\varepsilon_d
	=
	\begin{pmatrix}
		\varepsilon_{1,d}\\
		\varepsilon_{2,d}\\
		\vdots\\
		\varepsilon_{n_0,d}
	\end{pmatrix},
	\qquad
	\theta_d
	=
	\bigl(
	\lambda_1(d)',\ldots,\lambda_{n_0}(d)',\beta(d)'
	\bigr)' .
	\]
	Then the infeasible stacked treated regression is
	\begin{equation}
		Y_d
		=
		Z_d\theta_d+\varepsilon_d,
		\label{eq:stacked-regression-true-factor}
	\end{equation}
	where
	\[
	Z_d
	=
	\begin{pmatrix}
		F_d & 0 & \cdots & 0 & X_{1,d} \\
		0 & F_d & \cdots & 0 & X_{2,d} \\
		\vdots & \vdots & \ddots & \vdots & \vdots \\
		0 & 0 & \cdots & F_d & X_{n_0,d}
	\end{pmatrix}.
	\]
	By the Schur-complement criterion applied to the block form of $Z_d'Z_d$,
	\[
	Z_d'Z_d>0
	\quad\Longleftrightarrow\quad
	F_d'F_d>0
	\ \text{and}\
	\sum_{j=1}^{n_0}
	X_{j,d}'M_{F_d}X_{j,d}>0,
	\]
	where  $M_{F_d}
	:=
	I_{T_d}
	-
	F_d(F_d'F_d)^{-1}F_d'$.
	The feasible version replaces the unknown factor block $F_d$ by the
	control-unit estimator $\widehat F_d$. Thus,
	\[
	Y_{i,d}
	=
	\widehat F_d\lambda_i(d)
	+
	X_{i,d}\beta(d)
	+
	\widetilde\varepsilon_{i,d},
	\qquad
	\widetilde\varepsilon_{i,d}
	:=
	\varepsilon_{i,d}
	-
	(\widehat F_d-F_d)\lambda_i(d).
	\]
	The feasible stacked version of \eqref{eq:stacked-regression-true-factor} is
	\[  Y_d = \widehat Z_d \theta_d  + \widetilde \varepsilon_d. \]
	Thus the feasible least-squares estimator for $\theta_d$ is
	\[
	\widehat\theta_d
	=
	(\widehat Z_d'\widehat Z_d)^{-1}
	\widehat Z_d'Y_d,
	\qquad
	\widehat\theta_d
	=
	\bigl(
	\widehat\lambda_1(d)',
	\ldots,
	\widehat\lambda_{n_0}(d)',
	\widehat\beta(d)'
	\bigr)' .
	\]
	Equivalently, defining
	\[
	M_{\widehat F_d}
	:=
	I_{T_d}
	-
	\widehat F_d
	(\widehat F_d'\widehat F_d)^{-1}
	\widehat F_d' ,
	\]
	the slope estimator is
	\begin{equation}
		\widehat\beta(d)
		=
		\left(\sum_{j=1}^{n_0}X_{j,d}'M_{\widehat F_d}X_{j,d}\right)^{-1}
		\left(\sum_{j=1}^{n_0}X_{j,d}'M_{\widehat F_d}Y_{j,d}\right),
		\label{eq:beta-hat}
	\end{equation}
	and, for $i=1,\ldots,n_0$,
	\begin{equation}
		\widehat\lambda_i(d)
		=
		(\widehat F_d'\widehat F_d)^{-1}
		\widehat F_d'\{Y_{i,d}-X_{i,d}\widehat\beta(d)\}.
		\label{eq:lambda-hat}
	\end{equation}
	These expressions
	do not require inversion of the full high-dimensional matrix.
	
	For  a fixed treated unit $i$ and a fixed post-treatment date $t\in\cT_1$,
	the systematic treatment effect is defined by
	\begin{equation}
		\tau_{it}^{*}
		=
		f_t'\{\lambda_i(1)-\lambda_i(0)\}
		+
		X_{it}'\{\beta(1)-\beta(0)\},
		\label{eq:tau-star-pop}
	\end{equation}
	and the estimated systematic effect is
	\begin{equation}
		\widehat\tau_{it}^{*}
		=
		\widehat f_t'
		\{\widehat\lambda_i(1)-\widehat\lambda_i(0)\}
		+
		X_{it}'\{\widehat\beta(1)-\widehat\beta(0)\}.
		\label{eq:tau-hat-star}
	\end{equation}
	The next section provides an estimator for the variance of $\widehat\tau_{it}^{*}$.

	We use the following notation throughout. For two positive deterministic
	sequences $a_{nT}$ and $b_{nT}$, $a_{nT}=O(b_{nT})$ means that
	$|a_{nT}|/b_{nT}$ is bounded, while $a_{nT}=o(b_{nT})$ means that
	$|a_{nT}|/b_{nT}\to 0$. We write $a_{nT}\asymp b_{nT}$ when both
	$a_{nT}=O(b_{nT})$ and $b_{nT}=O(a_{nT})$ hold. The stochastic analogues
	$O_p(\cdot)$ and $o_p(\cdot)$ are used in their usual sense. We use $A := B$ to denote that $A$ is defined to be $B$.
	We use $\|\cdot\|$ and $\|\cdot\|_F$ to denote the spectral norm  and Frobenius norm, respectively. The two norms are equivalent for any matrix of finite rank. In the special case of vectors, the two norms are identical.
	The arrow $\Rightarrow$  denotes convergence in distribution, $\overset{p}{\to}$
	denotes convergence in probability, and $\plim$ denotes probability limit.

	\subsection{Standard errors of the estimated treatment effects}
	
	The sampling uncertainty of $\widehat\tau_{it}^{*}$ has two components. The first comes from
	estimating the treated-unit regression parameters
	$\lambda_i(d)$ and $\beta(d)$ in the two blocks $d=0,1$; this is denoted
	by $V_{it}^{\mathrm{reg}}$. The second comes from estimating the common
	factor $f_t$ from the control units; this is denoted by $V_{it}^f$.
	Under the maintained cross-block and treated-control independence
	conditions, these components are asymptotically uncorrelated, so the
	total variance is estimated by
	\[
	\widehat V_{it}
	=
	\widehat V_{it}^{\mathrm{reg}}
	+
	\widehat V_{it}^{f}.
	\]

	\subsubsection{Treated-regression variance}
	
	To estimate the regression component of uncertainty,
	$V_{it}^{\mathrm{reg}}$, it is convenient to use the compact stacked
	representation. This yields a succinct sandwich-form variance estimator
	based on the stacked treated regression.
	
	The regression variance for block $d$ is estimated by the sandwich formula
	\begin{equation}
		\widehat{\Var}(\widehat\theta_d)=
		(\widehat Z_d'\widehat Z_d)^{-1}
		\widehat Z_d'\widehat\Omega_d\widehat Z_d
		(\widehat Z_d'\widehat Z_d)^{-1},
		\label{eq:feasible-sandwich-cov-d}
	\end{equation}
	where
	\begin{equation}
		\widehat\Omega_d  =
		\operatorname{diag}\!\left(\widehat \varepsilon_d\widehat \varepsilon_d'\right),
		\qquad
		\widehat \varepsilon_d
		=
		Y_d-\widehat Z_d\widehat\theta_d .
		\label{eq:omega-diagonal-residuals}
	\end{equation}
	Thus, $\widehat\Omega_d$ is the diagonal matrix whose diagonal entries
	are the squared residuals from the feasible treated regression.
	More generally, the matrix
	$\widehat Z_d'\widehat\Omega_d\widehat Z_d$
	can be replaced by a clustered or HAC estimator of the same object,
	depending on the dependence structure allowed for the treated-regression
	errors.
	
	For treated unit $i$ and date $t$, define
	\begin{equation}
		c_{i,t,d}
		:=
		\bigl(0',\ldots,0',f_t',0',\ldots,0',X_{it}'\bigr)',
		\label{eq:contrast-vector-true-factor}
	\end{equation}
	where $f_t$ appears in the block corresponding to $\lambda_i(d)$ and
	$X_{it}$ appears in the common $\beta(d)$ block. Its feasible analogue is
	\[
	\widehat c_{i,t,d}
	:=
	\bigl(
	0',\ldots,0',\widehat f_t',0',\ldots,0',X_{it}'
	\bigr)' .
	\]
	Then
	\begin{equation}
		\widehat c_{i,t,d}'\widehat\theta_d
		=
		\widehat f_t'\widehat\lambda_i(d)
		+
		X_{it}'\widehat\beta(d).
		\label{eq:contrast-equals-block-fit}
	\end{equation}
	The corresponding block-specific regression variance is
	\begin{equation}
		\widehat V_{it,d}^{\mathrm{reg}}
		=
		\widehat c_{i,t,d}'
		\widehat{\Var}(\widehat\theta_d)
		\widehat c_{i,t,d}.
		\label{eq:block-reg-var}
	\end{equation}
	This variance includes the uncertainty from estimating
	$\lambda_i(d)$, the uncertainty from estimating $\beta(d)$, and their
	covariance within block $d$. But it does not include the uncertainty from the estimated factors, which is considered in
	the next subsection.

	Because the pre- and post-treatment treated-regression score vectors are
	asymptotically independent under the maintained cross-block condition, the
	blockwise regression variances add.
	Hence the
	regression-estimation component of the variance of
	$\widehat\tau_{it}^{*}$ is estimated by
	\begin{equation}
		\widehat V_{it}^{\mathrm{reg}}
		=
		\widehat V_{it,1}^{\mathrm{reg}}
		+
		\widehat V_{it,0}^{\mathrm{reg}}.
		\label{eq:reg-var-sum}
	\end{equation}
	This additive structure is made explicit by the expansion of
	$\widehat\tau_{it}^{*}-\tau_{it}^{*}$ in \eqref{eq:tau-expansion} below.

	\subsubsection{Variance of estimated factors}

	The estimated factors are obtained from the control units using the
	interactive-effects estimator. For each fixed $t$, the factor estimator
	has the first-order representation
	\begin{equation}
		\widehat f_t-f_t
		=
		\left(
		\frac{1}{n_1}
		\sum_{k=n_0+1}^{n}
		\lambda_k(0)\lambda_k(0)'
		\right)^{-1}
		\frac{1}{n_1}
		\sum_{k=n_0+1}^{n}
		\lambda_k(0)\varepsilon_{kt}(0)
		+O_p\left(\frac 1 {n_1} +\frac 1 T\right)
		\label{eq:factor-representation}
	\end{equation}
	Thus, under the assumptions,
	$\widehat f_t$ is asymptotically normal. The estimation error in the
	common slope coefficient does not affect this leading
	representation, because both $n_1$ and $T$ diverge and the slope
	estimator converges at the faster rate $\sqrt{n_1T}$.
	
	Under the maintained cross-sectional independence condition for control-unit
	errors, define
	\[
	Q_\lambda
	=
	\frac{1}{n_1}
	\sum_{k=n_0+1}^{n}
	\lambda_k(0)\lambda_k(0)',
	\qquad
	S_t
	=E \Big[
	\frac{1}{n_1}
	\sum_{k=n_0+1}^{n}
	\lambda_k(0)\lambda_k(0)'
	\varepsilon_{kt}(0)^2 \Big] .
	\]
	The factor leading expansion has variance
	\[ \operatorname{Var}(\widehat f_t)= \frac 1 {n_1}
	Q_\lambda^{-1}S_t Q_\lambda^{-1}.
	\]
	The feasible estimator is
	\[
	\widehat{\operatorname{Var}}(\widehat f_t)
	=
	\frac{1}{n_1}
	\widehat Q_\lambda^{-1}
	\widehat S_t
	\widehat Q_\lambda^{-1},
	\]
	where
	\[
	\widehat Q_\lambda
	=
	\frac{1}{n_1}
	\sum_{k=n_0+1}^{n}
	\widehat\lambda_k(0)\widehat\lambda_k(0)',
	\qquad
	\widehat S_t
	=
	\frac{1}{n_1}
	\sum_{k=n_0+1}^{n}
	\widehat\lambda_k(0)\widehat\lambda_k(0)'
	\widehat\varepsilon_{kt}(0)^2 .
	\]
	Here,
	\[
	\widehat\varepsilon_{kt}(0)
	=
	Y_{kt}
	-
	\widehat\lambda_k(0)'\widehat f_t
	-
	X_{kt}'\widehat\beta_c,
	\qquad
	k=n_0+1,\ldots,n.
	\]
	and $\widehat \beta_c$ is the estimated control-unit slope coefficient.
	
	By \eqref{eq:tau-expansion} below, the contribution of factor estimation error to the treatment-effect estimator is
	$[\lambda_i(1)-\lambda_i(0)]'(\widehat f_t-f_t)$.
	Let \[
	\Delta\lambda_i
	:=
	\lambda_i(1)-\lambda_i(0),
	\qquad
	\widehat{\Delta\lambda}_i
	:=
	\widehat\lambda_i(1)-\widehat\lambda_i(0).
	\]
	The population and feasible factor-variance components are
	\begin{equation}
		V_{it}^{f}
		=
		\Delta\lambda_i'
		\operatorname{Var}(\widehat f_t-f_t)
		\Delta\lambda_i,
		\qquad
		\widehat V_{it}^{f}
		=
		\widehat{\Delta\lambda}_i'
		\widehat{\operatorname{Var}}(\widehat f_t)
		\widehat{\Delta\lambda}_i.
		\label{eq:factor-var}
	\end{equation}
	
	Combining the regression-estimation and factor-estimation components, the
	estimated variance of the systematic treatment effect is
	\[
	\widehat V_{it}
	=
	\widehat V_{it}^{\mathrm{reg}}
	+
	\widehat V_{it}^{f}.
	\]
	The corresponding standard error is
	\[
	\widehat{\operatorname{se}}\!\left(\widehat\tau_{it}^{*}\right)
	=
	\widehat V_{it}^{1/2}.
	\]

	\subsection{Assumptions}

	In terms of inferential theory, we assume $T_d\to\infty$ for $d=0,1$ and
	$n_1\to\infty$, with
	\begin{equation}
		\frac{\sqrt{n_1}}{T}\to0,
		\qquad
		\frac{\sqrt{T}}{n_1}\to0.
		\label{eq:main-rate-controls}
	\end{equation}
	The number of treated units $n_0$ may be either fixed or diverging. When
	$n_0\to\infty$, the following additional rate condition will be used in the
	case where the post-treatment factor at date $t$ is zero $(f_t=0)$:
	\begin{equation}
		\label{eq:pooled-slope-rate}
		\frac{\sqrt{n_0T_d}}{T}\to0,
		\qquad
		\frac{n_0}{n_1}\to0,
		\qquad
		\frac{\sqrt{n_0T_d}}{n_1}\to0,
		\qquad d=0,1 .
	\end{equation}
	
	This restriction is invoked only in the knife-edge case $f_t=0$ with
	$n_0\to\infty$. In that case, the cross-sectionally pooled slope determines the
	leading regression-side variance, while \eqref{eq:pooled-slope-rate} makes the
	higher-order terms negligible; this point will be made more precise in the
	proof.
	
	Let $0<M<\infty$ denote a generic
	constant, not depending on $n$ or $T$.

	\begin{assumption}[Factor moment nonsingularity]
		\label{ass:factor-moments}
		The factors satisfy $E\|f_t\|^4\leq M<\infty$ and
		\[
		\frac1T\sum_{t=1}^T f_tf_t'
		\overset{p}{\to}
		\Sigma_f,
		\qquad
		\Sigma_f>0 .
		\]
		Moreover, for each $d\in\{0,1\}$,
		\[
		Q_d
		:=
		\frac{F_d'F_d}{T_d}
		\overset{p}{\to}
		Q_{f,d},
		\qquad
		Q_{f,d}>0 .
		\]
	\end{assumption}
	
	\begin{assumption}[Treated projected-design moment]
		\label{ass:treated-design}
		For each $d\in\{0,1\}$, let
		\[
		M_{F_d}
		:=
		I_{T_d}
		-
		F_d(F_d'F_d)^{-1}F_d'
		\]
		and
		\[
		S_{xx,F,d}
		:=
		\frac1{n_0T_d}
		\sum_{j=1}^{n_0}
		X_{j,d}'M_{F_d}X_{j,d}.
		\]
		Then $S_{xx,F,d}\overset{p}{\to}S_{xx,F,d}^{0}$, where $S_{xx,F,d}^{0}>0$.
	\end{assumption}

	\begin{assumption}[Weak serial correlation]
		\label{ass:weak-correlation}
		The error term $\varepsilon_{it}(d)$ is independent over $i$. For each block $d\in\{0,1\}$,
		\[
		E[\varepsilon_{i,d}\mid X_{i,d},F_d,\lambda_i(d)]=0 .
		\]
		In addition, the following conditions hold.
		\begin{enumerate}
			\item $E|\varepsilon_{it}(d)|^8\leq M$, uniformly over $(i,t,d)$.
			
			\item Let $\sigma_{i,ts}(d):=E[\varepsilon_{it}(d)\varepsilon_{is}(d)]$.
			There exist constants $c_{ts}$ such that
			$|\sigma_{i,ts}(d)|\leq c_{ts}$ for all $(i,t,s,d)$, and
			$T^{-1}\sum_{t,s=1}^{T}c_{ts}\leq M$.

			\item For the control-unit errors $\varepsilon_{it}(0)$, for all $(t,s)$,
			\[
			E\left|
			\frac1{\sqrt{n_1}}
			\sum_{i=n_0+1}^{n}
			\left[
			\varepsilon_{it}(0)\varepsilon_{is}(0)
			-
			E\{\varepsilon_{it}(0)\varepsilon_{is}(0)\}
			\right]
			\right|^4
			\leq M .
			\]
		\end{enumerate}
		The  fourth-moment bound over $i=n_0+1,\ldots,n$ is used
		for the control-unit factor-estimation expansion.
	\end{assumption}
	
	\begin{assumption}[Treated-unit score CLT]
		\label{ass:treated-score-clt}
		For each fixed treated unit $i$ and block $d\in\{0,1\}$, define
		\[
		\mathcal S_{i,d}^{\lambda}
		:=
		\frac1{\sqrt{T_d}}F_d'\varepsilon_{i,d},
		\qquad
		\mathcal S_d^{\beta}
		:=
		\frac1{\sqrt{n_0T_d}}
		\sum_{j=1}^{n_0}
		X_{j,d}'M_{F_d}\varepsilon_{j,d}.
		\]
		Then $\mathcal S_{i,d}^{\lambda}\Rightarrow N(0,\Phi_{i,d})$ and
		$\mathcal S_d^{\beta}\Rightarrow N(0,\Psi_d)$. Moreover,
		\[
		\begin{pmatrix}
			\mathcal S_{i,d}^{\lambda}\\
			\mathcal S_d^{\beta}
		\end{pmatrix}
		\Rightarrow
		N(0,\Xi_{i,d}),
		\qquad
		\Xi_{i,d}>0.
		\]
		For $d=0,1$, the block score vectors are asymptotically independent
		across $d$; equivalently, the joint limit of the $d=0$ and $d=1$
		score vectors is Gaussian with block-diagonal covariance.
	\end{assumption}
	
	\begin{remark} We note that the block structure of $\Xi_{i,d}$ depends on the growth of
		$n_0$. When $n_0$ is fixed, the two score components are asymptotically
		correlated. For example, when $n_0=1$, both
		$\mathcal S_{i,d}^{\lambda}$ and $\mathcal S_d^{\beta}$ are functions of
		the same error vector $\varepsilon_{i,d}$ (with $i=1$). When $n_0\to\infty$, the
		contribution of any fixed treated unit to $\mathcal S_d^{\beta}$ is
		negligible, so the asymptotic covariance between
		$\mathcal S_{i,d}^{\lambda}$ and $\mathcal S_d^{\beta}$ vanishes. In that
		case, $\Xi_{i,d}$ is block diagonal. Also, for treated units, the pre- and post-treatment periods are
		non-overlapping blocks, and their score vectors are assumed to be
		asymptotically independent.
	\end{remark}

	\begin{assumption}[Control-unit moment nonsingularity and score CLT]
		\label{ass:control-score-clt}
		Define
		\[
		Q_{\lambda}
		:=
		\frac1{n_1}
		\sum_{k=n_0+1}^{n}
		\lambda_k(0)\lambda_k(0)' .
		\]
		Then $Q_{\lambda}\to Q_\lambda^0$, where $Q_\lambda^0>0$.
		For each fixed $t$, define
		\[
		\mathcal S_t^f
		:=
		\frac1{\sqrt{n_1}}
		\sum_{k=n_0+1}^{n}
		\lambda_k(0)\varepsilon_{kt}(0).
		\]
		Then $\mathcal S_t^f\Rightarrow N(0,\Gamma_t)$, where $\Gamma_t>0$.
	\end{assumption}

	\subsection{Proof of Proposition~\ref{prop:prop-1}}
	To prove Proposition 1, we need a few lemmas.
	
	\begin{lemma}[Blockwise factor-error bounds used in the proof]
		\label{lem:blockwise-factor-error-bounds}
		Suppose the improved factor-rotation conditions hold for the control-unit factor
		estimator, with $n_1$ units and $T$ time periods.  Then, uniformly for $d\in\{0,1\}$,
		\begin{align}
			\frac{\|\Delta_d\|_F^2}{T_d}
			&=
			\Op\left(\frac1{n_1}+\frac1{T^2}\right),
			\label{eq:block-sq-norm}\\
			\frac{F_d'\Delta_d}{T_d}
			&=
			\Op\left(
			\frac1T+
			\frac1{\sqrt{n_1T_d}}+
			\frac1{n_1}
			\right),
			\label{eq:block-FDelta}\\
			\frac{\widehat F_d'\Delta_d}{T_d}
			&=
			\Op\left(
			\frac1T+
			\frac1{\sqrt{n_1T_d}}+
			\frac1{n_1}
			\right).
			\label{eq:block-hatFDelta}
		\end{align}
		More generally, if $A_d=(a_t')_{t\in\cT_d}$ is a block of bounded deterministic
		weights, or a block of regressors/residualized regressors satisfying $E\|a_t\|^2 \le M$ for all $t$, then
		\begin{equation}
			\frac{A_d'\Delta_d}{T_d}
			=
			\Op\left(
			\frac1T+
			\frac1{\sqrt{n_1T_d}}+
			\frac1{n_1}
			\right).
			\label{eq:block-ADelta}
		\end{equation}
		Let $\eta_d=(\eta_{d,t}')_{t\in\cT_d}$ be a block of zero mean random
		vectors, weakly correlated with bounded absolute sum of autocovariances
		and independent of the control-unit variables.
		Then a slightly sharper bound is
		\begin{equation}
			\frac{\eta_d'\Delta_d}{T_d}
			=
			\Op\left(
			\frac1{\sqrt{n_1T_d}}+
			\frac1{T\sqrt{T_d}}
			\right).
			\label{eq:block-zero-mean}
		\end{equation}
	\end{lemma}
	
	\begin{remark}
		The rate in \eqref{eq:block-sq-norm} is governed by the dimensions used to
		estimate the factors, namely the number of control units $n_1$ and the full
		time dimension $T$, rather than by the length $T_d$ of the sub-block. This is the blockwise analogue of the pointwise result
		in \citet{Bai_2003}: for each fixed $t$,
		$
		\widehat f_t-f_t = O_p(n_1^{-1/2}+T^{-1}).
		$
		Thus the average squared factor-estimation error over a block $\mathcal T_d$
		has the corresponding order $O_p(n_1^{-1}+T^{-2})$.
		Bai's corresponding bound is often stated with $1/T$ rather than $1/T^2$.
		The weaker $1/T$ rate is also sufficient for the present paper.
		A direct proof of \eqref{eq:block-sq-norm} is straightforward.
		
		The cross-product bounds \eqref{eq:block-FDelta} and
		\eqref{eq:block-hatFDelta} follow from \citet{Bai_2003}, while
		\eqref{eq:block-ADelta} follows from \citet{Bai_Ng_2006}. Although
		\eqref{eq:block-FDelta} and \eqref{eq:block-hatFDelta} are special cases of
		\eqref{eq:block-ADelta}, we state them separately because they are used
		repeatedly throughout the proof.

		The cross-product bounds in \eqref{eq:block-FDelta}-\eqref{eq:block-zero-mean} are sharper
		than what would follow from a direct Cauchy--Schwarz argument.
		They are weighted averages of $\hat f_t-f_t$. The bound improves with block size by  averaging out errors, so
		it depends on $T_d$  via $1/\sqrt{n_1 T_d}$.
		Since the required bounds follow directly from the
		existing factor-estimation literature, we omit the proof.
		In our application, $A_d$ will be the treated-unit covariate block
		$X_{i,d}$, while $\eta_d$ will be the treated-unit regression error block
		$\varepsilon_{i,d}$.
	\end{remark}
	
	Define
	\begin{equation}
		b_{nT,d}
		=
		\frac1T+
		\frac1{\sqrt{n_1T_d}}+
		\frac1{n_1},
		\label{eq:bntd}
	\end{equation}
	This is the rate that we use most of the time. Also define
	\begin{equation}
		\widetilde b_{nT,d}
		=
		\frac1{\sqrt{n_1T_d}}+
		\frac1{T\sqrt{T_d}}.
		\label{eq:bntd-zero-mean}
	\end{equation}

	\begin{lemma}[Pooled slope expansion]
		\label{lem:pooled-slope-expansion}
		Fix $d\in\{0,1\}$. Under the conditions of
		Proposition~\ref{prop:prop-1}, we have
		\begin{equation}
			\widehat\beta(d)-\beta(d)
			=
			\left\{
			\sum_{j=1}^{n_0}X_{j,d}'M_{F_d}X_{j,d}
			\right\}^{-1}
			\sum_{j=1}^{n_0}X_{j,d}'M_{F_d}\varepsilon_{j,d}
			+
			\Op(b_{nT,d}),
			\label{eq:beta-representation-3}
		\end{equation}
		where $b_{nT,d}$ is defined in \eqref{eq:bntd}.
	\end{lemma}
	
	This lemma states that, up to an $\Op(b_{nT,d})$ remainder,
	$\widehat \beta(d)-\beta(d)$ has the same expansion as in the case with
	known $F_d$.

	\begin{remark}
		\label{rem:pooled-slope-remainder}
		We make the following observations about
		Lemma~\ref{lem:pooled-slope-expansion}. The leading stochastic term in the
		pooled-slope expansion is of order $1/\sqrt{n_0T_d}$, while the remainder is
		of order $O_p(b_{nT,d})$. Therefore, if one wanted the leading term in the
		pooled-slope expansion itself to dominate the remainder, one would require
		\[
		\sqrt{n_0T_d}\,b_{nT,d}\to0 .
		\]
		By the definition of $b_{nT,d}$, this is implied by
		\eqref{eq:pooled-slope-rate}.
		
		Condition \eqref{eq:pooled-slope-rate} is used only for the knife-edge case
		$n_0\to\infty$ and $f_t=0$. In that case, the loading-estimation component of
		the treated-regression variance vanishes, and the relevant variance scale is
		generally $(n_0T_d)^{-1}$, provided the covariate component is nondegenerate.
		Even when $n_0\to\infty$, if $f_t\neq0$, condition
		\eqref{eq:pooled-slope-rate} is not needed for pointwise inference on
		$\tau_{it}^*$. The details are given in the proof.
	\end{remark}
	
	\begin{remark}
		\label{rem:pooled-slope-not-essential}
		For pointwise inference on a fixed treated unit, it is not essential to
		estimate the treated-unit slope $\beta(d)$ by pooling across treated
		units. One could instead estimate the slope separately for that treated
		unit, in which case the argument reduces to the fixed-$n_0$, single-unit
		version.
	\end{remark}

	\noindent
	\begin{proof}[Proof of Lemma \ref{lem:pooled-slope-expansion}]
		Let \begin{equation}
			\widehat S_{xx,d}
			=
			\frac1{n_0T_d}
			\sum_{j=1}^{n_0}X_{j,d}'M_{\widehat F_d}X_{j,d},
			\qquad
			S_{xx,F,d}
			=
			\frac1{n_0T_d}
			\sum_{j=1}^{n_0}X_{j,d}'M_{F_d}X_{j,d}.
			\label{eq:Sxx-def}
		\end{equation}
		Substituting
		$Y_{j,d}=X_{j,d}\beta(d)+F_d\lambda_j(d)+\varepsilon_{j,d}$ into the
		estimator $\widehat \beta(d)$ gives
		\[
		\widehat\beta(d)-\beta(d)
		=
		\widehat S_{xx,d}^{-1} \,
		\frac{1}{n_0T_d}
		\sum_{j=1}^{n_0}
		X_{j,d}'M_{\widehat F_d}
		\{F_d\lambda_j(d)+\varepsilon_{j,d}\}.
		\]
		Since $F_d=\widehat F_d-\Delta_d$ and
		$M_{\widehat F_d}\widehat F_d=0$, we have
		$M_{\widehat F_d}F_d=-M_{\widehat F_d}\Delta_d$. Hence
		\[
		\widehat\beta(d)-\beta(d)
		=
		\widehat S_{xx,d}^{-1}
		\frac{1}{n_0T_d}
		\sum_{j=1}^{n_0}
		X_{j,d}'M_{\widehat F_d}\varepsilon_{j,d}
		-
		\widehat S_{xx,d}^{-1}
		\frac{1}{n_0T_d}
		\sum_{j=1}^{n_0}
		X_{j,d}'M_{\widehat F_d}\Delta_d\lambda_j(d).
		\]
		
		Applying Lemma~\ref{lem:blockwise-factor-error-bounds} with the averaged
		array
		$
		\bar A_d=\frac1{n_0}\sum_{j=1}^{n_0}X_{j,d}\lambda_j(d)'$ componentwise
		gives
		\[
		\frac1{n_0T_d}
		\sum_{j=1}^{n_0}X_{j,d}'M_{\widehat F_d}\Delta_d\lambda_j(d)
		=
		O_p(b_{nT,d}).
		\]
		Together with $\widehat S_{xx,d}^{-1}= O_p(1)$, we have
		\[
		\widehat\beta(d)-\beta(d)
		=
		\widehat S_{xx,d}^{-1}
		\frac{1}{n_0T_d}
		\sum_{j=1}^{n_0}
		X_{j,d}'M_{\widehat F_d}\varepsilon_{j,d} +O_p(b_{nT,d}).
		\]
		Next, we show $M_{\widehat F_d}$ can be replaced by $M_{F_d}$.
		Let
		\[
		R_d:=
		\frac1{n_0T_d}
		\sum_{j=1}^{n_0}
		X_{j,d}'(M_{\widehat F_d}-M_{F_d})\varepsilon_{j,d}.
		\]
		Since $M_{\widehat F_d}-M_{F_d}=P_{F_d}-P_{\widehat F_d}$, $\widehat F_d=F_d+\Delta_d$, and $Q_d=F_d'F_d/T_d$, the standard projector expansion gives
		\[
		P_{F_d}-P_{\widehat F_d}
		=
		-
		M_{F_d}\Delta_dQ_d^{-1}\frac{F_d'}{T_d}
		-
		F_dQ_d^{-1}\frac{\Delta_d'M_{F_d}}{T_d}
		+
		\Upsilon_d,\]
		where $
		\|\Upsilon_d\|
		= O_p\!\left(\frac{\|\Delta_d\|^2}{T_d}\right),$
		here $\|\cdot\|$ denotes the spectral norm (the largest singular value). Therefore
		\[
		R_d=R_{1d}+R_{2d}+R_{3d},
		\]
		where
		\[
		R_{1d}
		=
		-\frac1{n_0}
		\sum_{j=1}^{n_0}
		\left(
		\frac{X_{j,d}'M_{F_d}\Delta_d}{T_d}
		\right)
		Q_d^{-1}
		\left(
		\frac{F_d'\varepsilon_{j,d}}{T_d}
		\right),
		\]
		\[
		R_{2d}
		=
		-\frac1{n_0}
		\sum_{j=1}^{n_0}
		\left(
		\frac{X_{j,d}'F_d}{T_d}
		\right)
		Q_d^{-1}
		\left(
		\frac{\Delta_d'M_{F_d}\varepsilon_{j,d}}{T_d}
		\right),
		\]
		and
		\[
		R_{3d}
		=
		\frac1{n_0T_d}
		\sum_{j=1}^{n_0}
		X_{j,d}'\Upsilon_d\varepsilon_{j,d}.
		\]
		Since $Q_d^{-1}=O_p(1)$,   Lemma~\ref{lem:blockwise-factor-error-bounds}
		gives
		$X_{j,d}'M_{F_d}\Delta_d/T_d=O_p(b_{nT,d})$, while
		$F_d'\varepsilon_{j,d}/T_d=O_p(T_d^{-1/2})$. Hence
		\[
		R_{1d}
		=
		O_p\!\left(\frac{b_{nT,d}}{\sqrt{T_d}}\right)
		=
		O_p(b_{nT,d}).
		\]
		
		For $R_{2d}$, use $M_{F_d}=I_{T_d}-P_{F_d}$ to write
		\[
		\frac{\Delta_d'M_{F_d}\varepsilon_{j,d}}{T_d}
		=
		\frac{\Delta_d'\varepsilon_{j,d}}{T_d}
		-
		\frac{\Delta_d'F_d}{T_d}
		Q_d^{-1}
		\frac{F_d'\varepsilon_{j,d}}{T_d}.
		\]
		By the last part of Lemma~\ref{lem:blockwise-factor-error-bounds},
		$\Delta_d'\varepsilon_{j,d}/T_d
		=
		O_p\!\bigl((n_1T_d)^{-1/2}+(T\sqrt{T_d})^{-1}\bigr)$.
		Note $F_d'\varepsilon_{j,d}/T_d=O_p(T_d^{-1/2})$, and  by Lemma~\ref{lem:blockwise-factor-error-bounds},
		$\Delta_d'F_d/T_d=O_p(b_{nT,d})$.
		Therefore
		\[
		\frac{\Delta_d'M_{F_d}\varepsilon_{j,d}}{T_d}
		=
		O_p\!\left(
		\frac1{\sqrt{n_1T_d}}
		+
		\frac1{T\sqrt{T_d}}
		+
		\frac{b_{nT,d}}{\sqrt{T_d}}
		\right)
		=
		O_p(b_{nT,d}).
		\]
		Since $X_{j,d}'F_d/T_d=O_p(1)$, it follows that
		$R_{2d}=O_p(b_{nT,d})$.
		
		Finally,  Lemma~\ref{lem:blockwise-factor-error-bounds} gives
		$\|\Delta_d\|_F^2/T_d
		=
		O_p(n_1^{-1}+T^{-2})$. Since $\|\Delta_d\|\le \|\Delta_d\|_F$,
		\[
		\|\Upsilon_d\|
		=
		O_p\!\left(\frac{\|\Delta_d\|^2}{T_d}\right)
		\le
		O_p\!\left(\frac{\|\Delta_d\|_F^2}{T_d}\right)
		=
		O_p\!\left(\frac1{n_1}+\frac1{T^2}\right)
		=
		O_p(b_{nT,d}).
		\]
		Under the average moment bound
		$(n_0T_d)^{-1}\sum_{j=1}^{n_0}\|X_{j,d}\|\,\|\varepsilon_{j,d}\|=O_p(1)$,
		we obtain $R_{3d}=O_p(b_{nT,d})$. Combining the three bounds yields
		$R_d=O_p(b_{nT,d})$, and hence
		\[
		\frac1{n_0T_d}
		\sum_{j=1}^{n_0}
		X_{j,d}'M_{\widehat F_d}\varepsilon_{j,d}
		=
		\frac1{n_0T_d}
		\sum_{j=1}^{n_0}
		X_{j,d}'M_{F_d}\varepsilon_{j,d}
		+
		O_p(b_{nT,d}).
		\]
		Next we show $\widehat S_{xx,d}^{-1}$ can be replaced by $S_{xx,F,d}^{-1}$. Notice
		\[
		\widehat S_{xx,d}-S_{xx,F,d}
		=
		\frac1{n_0T_d}
		\sum_{j=1}^{n_0}
		X_{j,d}'(M_{\widehat F_d}-M_{F_d})X_{j,d}.
		\]
		Using the projector-difference expansion derived above, together with Lemma~\ref{lem:blockwise-factor-error-bounds},
		the same argument as for the score replacement gives
		\[
		\widehat S_{xx,d}-S_{xx,F,d}
		=
		O_p(b_{nT,d}).
		\]
		Indeed, the two first-order projector terms are controlled by the
		Lemma~\ref{lem:blockwise-factor-error-bounds} bounds for $X_{j,d}'M_{F_d}\Delta_d/T_d$ and
		$\Delta_d'M_{F_d}X_{j,d}/T_d$, while the projector remainder is controlled by
		$\|\Upsilon_d\|=O_p(\|\Delta_d\|^2/T_d)$ and
		$\|\Delta_d\|_F^2/T_d=O_p(n_1^{-1}+T^{-2})$.
		
		Assume that $S_{xx,F,d}$ is nonsingular with eigenvalues bounded away from
		zero and that $b_{nT,d}=o(1)$.  Since
		\[
		\widehat S_{xx,d}^{-1}-S_{xx,F,d}^{-1}
		=
		\widehat S_{xx,d}^{-1}
		\left(S_{xx,F,d}-\widehat S_{xx,d}\right)
		S_{xx,F,d}^{-1},
		\]
		we have
		\[
		\left\|
		\widehat S_{xx,d}^{-1}-S_{xx,F,d}^{-1}
		\right\|
		\le
		\left\|\widehat S_{xx,d}^{-1}\right\|
		\left\|
		\widehat S_{xx,d}-S_{xx,F,d}
		\right\|
		\left\|S_{xx,F,d}^{-1}\right\|.
		\]
		Because $\|S_{xx,F,d}^{-1}\|=O_p(1)$ and
		$\|\widehat S_{xx,d}-S_{xx,F,d}\|=O_p(b_{nT,d})=o_p(1)$, Weyl's inequality
		implies that $\widehat S_{xx,d}$ is nonsingular with probability approaching
		one and $\|\widehat S_{xx,d}^{-1}\|=O_p(1)$.  Therefore
		\[
		\widehat S_{xx,d}^{-1}-S_{xx,F,d}^{-1}
		=
		O_p(b_{nT,d}).
		\]
		This implies that
		\[ \widehat S_{xx,d}^{-1}
		\frac1{n_0T_d}
		\sum_{j=1}^{n_0}X_{j,d}'M_{F_d}\varepsilon_{j,d}=
		S_{xx,F,d}^{-1}
		\frac1{n_0T_d}
		\sum_{j=1}^{n_0}X_{j,d}'M_{F_d}\varepsilon_{j,d}+ \Op(b_{nT,d})\cdot O_p((n_0 T_d)^{-1/2}) \]
		Thus replacing $\widehat S_{xx,d}^{-1}$ by $S_{xx,F,d}^{-1}$ adds another smaller order term. Combining results we have proved the lemma.
	\end{proof}

	\begin{lemma}[Unit-specific loading expansion]
		\label{lem:loading-expansion}
		Fix $d\in\{0,1\}$ and treated unit $i$.  Suppose the conditions of
		Proposition~\ref{prop:prop-1} hold. Then
		\begin{equation}
			\widehat\lambda_i(d)-\lambda_i(d)
			=
			(F_d' F_d)^{-1}F_d'\varepsilon_{i,d}
			-
			(F_d' F_d)^{-1} F_d' X_{i,d}\{\widehat\beta(d)-\beta(d)\}
			+ O_p(b_{nT,d}).
			\label{eq:lambda-representation}
		\end{equation}
	\end{lemma}

	Similar to the previous lemma, this lemma states that, up to an $\Op(b_{nT,d})$ remainder,
	$\widehat \lambda_i(d)-\lambda_i(d)$ has the same expansion as in the case with
	known $F_d$.
	
	\begin{proof}[Proof of Lemma \ref{lem:loading-expansion}]
		Substituting \eqref{eq:block-model} into \eqref{eq:lambda-hat} and using
		$F_d=\widehat F_d-\Delta_d$ gives
		\begin{align}
			\widehat\lambda_i(d)-\lambda_i(d)
			={}&
			(\widehat F_d'\widehat F_d)^{-1}\widehat F_d'\varepsilon_{i,d}
			-
			(\widehat F_d'\widehat F_d)^{-1}\widehat F_d'\Delta_d\lambda_i(d)
			\nonumber\\
			&-
			(\widehat F_d'\widehat F_d)^{-1}\widehat F_d'X_{i,d}
			\{\widehat\beta(d)-\beta(d)\}.
			\label{eq:lambda-expansion-raw}
		\end{align}
		By \eqref{eq:block-hatFDelta},       $\widehat F_d'\Delta_d/T_d = O_p(b_{nT,d})$, and $(\widehat F_d'\widehat F_d/T_d)^{-1}=O_p(1)$,  thus the second term  is $O_p(b_{nT,d})$.
		
		Next
		\begin{equation}
			\Big(\frac{\widehat F_d'\widehat F_d}{T_d}\Big)^{-1}
			-
			\Big(\frac{F_d'F_d}{T_d}\Big)^{-1} = \Big(\frac{\widehat F_d'\widehat F_d}{T_d}\Big)^{-1}
			\Big[\frac{F_d'F_d}{T_d}-\frac{\widehat F_d'\widehat F_d}{T_d} \Big]\Big(\frac{F_d'F_d}{T_d}\Big)^{-1} =O_p(b_{nT,d}).
			\label{eq:hatFhatF-FF}
		\end{equation}
		For the first term in \eqref{eq:lambda-expansion-raw},
		\begin{equation}
			\frac{\widehat F_d'\varepsilon_{i,d}}{T_d}
			=
			\frac{F_d'\varepsilon_{i,d}}{T_d}
			+
			\frac{\Delta_d'\varepsilon_{i,d}}{T_d}.
			\label{eq:hatF-eps-decomp}
		\end{equation}
		The second term in \eqref{eq:hatF-eps-decomp} is $O_p(\widetilde b_{nT,d})$ by \eqref{eq:block-zero-mean}. Thus we can rewrite
		the first term in \eqref{eq:lambda-expansion-raw} as $(\widehat F_d'\widehat F_d/T_d)^{-1} F_d' \varepsilon_{i,d}/T_d + O_p(\widetilde b_{nT,d})$.
		In view of \eqref{eq:hatFhatF-FF}, the first term in \eqref{eq:lambda-expansion-raw} can be written as
		\[ (\widehat F_d'\widehat F_d/T_d)^{-1} \widehat F_d' \varepsilon_{i,d}/T_d =
		(F_d'F_d/T_d)^{-1} F_d' \varepsilon_{i,d}/T_d +  O_p(b_{nT,d})O_p(T_d^{-1/2})+O_p(\widetilde b_{nT,d}) \]
		The sum of the two $O_p$ terms is dominated by $O_p(b_{nT,d})$.

		For the third term in \eqref{eq:lambda-expansion-raw},
		\[ \frac{ \widehat F_d'X_{i,d}}{T_d} = \frac{  F_d'X_{i,d}}{T_d}+  \frac{ \Delta_d'X_{i,d}}{T_d} \]
		the second term above is $O_p(b_{nT,d})$ by Lemma \ref{lem:blockwise-factor-error-bounds} with $A_d=X_{i,d}$.
		Thus
		\[
		\left[
		\left(\frac{\widehat F_d'\widehat F_d}{T_d}\right)^{-1}
		\frac{\widehat F_d'X_{i,d}}{T_d}
		-
		\left(\frac{F_d'F_d}{T_d}\right)^{-1}
		\frac{F_d'X_{i,d}}{T_d}
		\right]
		=
		O_p(b_{nT,d}).
		\]
		Moreover, by Lemma~\ref{lem:pooled-slope-expansion},
		\[
		\widehat\beta(d)-\beta(d)
		=
		O_p\left((n_0T_d)^{-1/2}+b_{nT,d}\right).
		\]
		Therefore,
		\begin{align*}
			&
			(\widehat F_d'\widehat F_d)^{-1}\widehat F_d'X_{i,d}
			\{\widehat\beta(d)-\beta(d)\}
			\\
			&\qquad =
			(F_d'F_d)^{-1}F_d'X_{i,d}
			\{\widehat\beta(d)-\beta(d)\}
			+
			O_p(b_{nT,d})
			O_p\left((n_0T_d)^{-1/2}+b_{nT,d}\right).
		\end{align*}
		The product of the two $O_p$ terms is dominated by $O_p(b_{nT,d})$. Hence
		\[
		(\widehat F_d'\widehat F_d)^{-1}\widehat F_d'X_{i,d}
		\{\widehat\beta(d)-\beta(d)\}
		=
		(F_d'F_d)^{-1}F_d'X_{i,d}
		\{\widehat\beta(d)-\beta(d)\}
		+
		O_p(b_{nT,d}).
		\]
		Combining results, we proved Lemma \ref{lem:loading-expansion}.
	\end{proof}

	Let $b_{nT}=\sum_{d=0}^1 b_{nT,d}$ then
	\begin{equation}
		b_{nT}
		=
		\frac{2}{T}
		+
		\frac{1}{\sqrt{n_1T_0}}
		+
		\frac{1}{\sqrt{n_1T_1}}
		+
		\frac{2}{n_1}.
		\label{eq:bnt-total}
	\end{equation}
	Assume $T_0\to\infty$, $T_1\to\infty$, and $T/n_1^2\to 0$. Then
	$T^{-1}=o\{(T_0^{-1}+T_1^{-1})^{1/2}\}$. Also, for each
	$d\in\{0,1\}$, $(n_1T_d)^{-1/2}=o(T_d^{-1/2})$, and therefore
	$(n_1T_d)^{-1/2}=o(T_0^{-1}+T_1^{-1})^{1/2}$. Finally, since
	$T_d\le T$, the condition $T/n_1^2\to0$ implies $T_d/n_1^2\to0$, or
	equivalently $n_1^{-1}=o(T_d^{-1/2})$. Hence
	$n_1^{-1}=o\{(T_0^{-1}+T_1^{-1})^{1/2}\}$. Combining these bounds,
	\begin{equation}
		b_{nT}
		=
		o\{(T_0^{-1}+T_1^{-1})^{1/2}\}.
		\label{eq:bnt-small-relative-time}
	\end{equation}
	This bound will be used below when $n_0$ is fixed or when $f_t\neq0$; in
	these cases the treated-regression variance has the usual $T_d^{-1}$ scale.
	When $n_0\to\infty$ and $f_t=0$, the leading regression variance is instead
	of order $(n_0T_d)^{-1}$, and the stronger condition in
	\eqref{eq:pooled-slope-rate} is required.
	Under \eqref{eq:pooled-slope-rate}, we have, with
	$b_{nT}=b_{nT,0}+b_{nT,1}$,
	\begin{equation}
		b_{nT}
		=
		o\left(
		\left\{(n_0T_0)^{-1}+(n_0T_1)^{-1}\right\}^{1/2}
		\right).
		\label{eq:bnt-pooled-slope-rate}
	\end{equation}

	\begin{proof}[\bf Proof of Proposition~\ref{prop:prop-1}]
		For $t\in\cT_1$, by adding and subtracting, we have the identity
		\begin{align}
			\widehat\tau_{it}^{*}-\tau_{it}^{*}
			={}&
			\Bigl[
			f_t'\{\widehat\lambda_i(1)-\lambda_i(1)\}
			+X_{it}'\{\widehat\beta(1)-\beta(1)\}
			\Bigr]
			\nonumber\\
			-{}&
			\Bigl[
			f_t'\{\widehat\lambda_i(0)-\lambda_i(0)\}
			+X_{it}'\{\widehat\beta(0)-\beta(0)\}
			\Bigr]
			\nonumber\\
			&+(\widehat f_t-f_t)'\{\lambda_i(1)-\lambda_i(0)\}+R_{it},
			\label{eq:tau-expansion}
		\end{align}
		where
		\begin{equation}
			R_{it}
			=
			(\widehat f_t-f_t)'
			\Bigl[
			\{\widehat\lambda_i(1)-\lambda_i(1)\}
			-\{\widehat\lambda_i(0)-\lambda_i(0)\}
			\Bigr].
			\label{eq:R-it-def}
		\end{equation}
		
		For each $d\in\{0,1\}$, Lemmas~\ref{lem:pooled-slope-expansion}
		and~\ref{lem:loading-expansion} imply
		\begin{equation}
			f_t'\{\widehat\lambda_i(d)-\lambda_i(d)\}
			+X_{it}'\{\widehat\beta(d)-\beta(d)\}
			= A_{it,d}+O_p(b_{nT,d}),
			\label{eq:treated-regression-Aitd}
		\end{equation}
		where
		\begin{align*}
			A_{it,d}
			={}&
			f_t'
			\left(\frac{F_d'F_d}{T_d}\right)^{-1}
			\frac{F_d'\varepsilon_{i,d}}{T_d}
			+h_{it,d}'S_{xx,F,d}^{-1}
			\frac{1}{n_0T_d}
			\sum_{j=1}^{n_0}X_{j,d}'M_{F_d}\varepsilon_{j,d},\\
			h_{it,d}'
			={}&
			X_{it}'
			-
			f_t'
			\left(\frac{F_d'F_d}{T_d}\right)^{-1}
			\frac{F_d'X_{i,d}}{T_d},
			\qquad
			S_{xx,F,d}
			=
			\frac{1}{n_0T_d}
			\sum_{j=1}^{n_0}X_{j,d}'M_{F_d}X_{j,d}.
		\end{align*}
		The term $A_{it,d}$ contains no estimated quantities.
		
		Substituting \eqref{eq:treated-regression-Aitd} into
		\eqref{eq:tau-expansion} gives
		\begin{equation}
			\widehat\tau_{it}^{*}-\tau_{it}^{*}
			=
			A_{it,1}-A_{it,0}
			+(\widehat f_t-f_t)'\{\lambda_i(1)-\lambda_i(0)\}
			+R_{it}+O_p(b_{nT}),
			\label{eq:hattau-tau-preR}
		\end{equation}
		where $b_{nT}=b_{nT,1}+b_{nT,0}$. By the pointwise factor expansion,
		\begin{equation}
			\widehat f_t-f_t=O_p(n_1^{-1/2})+O_p(T^{-1}),
			\label{eq:pointwise-factor-rate}
		\end{equation}
		and by Lemma~\ref{lem:loading-expansion},
		\begin{equation}
			\widehat\lambda_i(d)-\lambda_i(d)
			=O_p(T_d^{-1/2})+O_p\{(n_0T_d)^{-1/2}\}+O_p(b_{nT,d}).
			\label{eq:lambda-rate-for-remainder}
		\end{equation}
		The product in $R_{it}$ is therefore of smaller order than $O_p(b_{nT,d})$,
		and is thus absorbed into the same remainder bound. Hence
		\begin{equation}
			\widehat\tau_{it}^{*}-\tau_{it}^{*}
			=
			A_{it,1}-A_{it,0}
			+(\widehat f_t-f_t)'\{\lambda_i(1)-\lambda_i(0)\}
			+O_p(b_{nT}).
			\label{eq:hattau-tau}
		\end{equation}
		
		We next justify the asymptotic normality of $A_{it,d}$. In compact notation,
		\begin{equation*}
			A_{it,d}=c_{it,d}'(Z_d'Z_d)^{-1}Z_d'\varepsilon_d,
			\qquad
			V_{it,d}^{\rm reg}=\Var(A_{it,d}).
		\end{equation*}
		If $n_0$ is fixed, $A_{it,d}$ is a scalar contrast from a fixed-dimensional
		least-squares regression. By the finite-dimensional score CLT in the
		maintained assumptions,
		\begin{equation}
			\frac{A_{it,d}}{\{V_{it,d}^{\rm reg}\}^{1/2}}
			\Rightarrow N(0,1).
			\label{eq:Aitd-clt-fixed}
		\end{equation}
		By the cross-block asymptotic independence condition in
		Assumption~\ref{ass:treated-score-clt}, $A_{it,0}$ and $A_{it,1}$ are
		asymptotically independent. Therefore,
		\[
		\frac{A_{it,1}-A_{it,0}}
		{
			\{V_{it,1}^{\mathrm{reg}}+V_{it,0}^{\mathrm{reg}}\}^{1/2}
		}
		\Rightarrow N(0,1).
		\]
		Thus, in the fixed-$n_0$ case,
		\[
		V_{it}^{\mathrm{reg}}
		=
		V_{it,1}^{\mathrm{reg}}
		+
		V_{it,0}^{\mathrm{reg}} .
		\]

		When $n_0\to\infty$, write
		\begin{equation}
			A_{it,d}=A_{it,d}^{\lambda}+A_{it,d}^{\beta},
			\label{eq:Aitd-lambda-beta}
		\end{equation}
		where
		\begin{align*}
			A_{it,d}^{\lambda}
			&=
			f_t'(F_d'F_d)^{-1}F_d'\varepsilon_{i,d},\\
			A_{it,d}^{\beta}
			&=
			h_{it,d}'S_{xx,F,d}^{-1}
			\frac{1}{n_0T_d}
			\sum_{j=1}^{n_0}X_{j,d}'M_{F_d}\varepsilon_{j,d}.
		\end{align*}
		If $f_t\neq0$, then by $Q_d\overset{p}{\to}Q_{f,d}>0$ and $\Phi_{i,d}>0$,
		\begin{equation*}
			\Var(A_{it,d}^{\lambda})\asymp T_d^{-1},
		\end{equation*}
		and the time-series CLT for $A_{it,d}^{\lambda}$ applies. Since
		$A_{it,d}^{\beta}=O_p\{(n_0T_d)^{-1/2}\}=o_p(T_d^{-1/2})$, Slutsky's theorem
		gives
		\begin{equation}
			\frac{A_{it,d}}{\{V_{it,d}^{\rm reg}\}^{1/2}}
			\Rightarrow N(0,1).
			\label{eq:Aitd-clt-ft-nonzero}
		\end{equation}
		If $f_t=0$, then $A_{it,d}^{\lambda}=0$ and $A_{it,d}=A_{it,d}^{\beta}$.
		Under the covariate nondegeneracy condition,
		\begin{equation*}
			V_{it,d}^{\rm reg}=\Var(A_{it,d}^{\beta})\asymp (n_0T_d)^{-1},
		\end{equation*}
		and the pooled cross-sectional/time-series CLT gives
		\begin{equation}
			\frac{A_{it,d}}{\{V_{it,d}^{\rm reg}\}^{1/2}}
			\Rightarrow N(0,1).
			\label{eq:Aitd-clt-ft-zero}
		\end{equation}
		As in the fixed-$n_0$ case, the cross-block asymptotic independence
		condition in Assumption~\ref{ass:treated-score-clt} implies that
		$A_{it,0}$ and $A_{it,1}$ are asymptotically independent. Hence the
		regression variance remains
		\[
		V_{it}^{\mathrm{reg}}
		=
		V_{it,1}^{\mathrm{reg}}
		+
		V_{it,0}^{\mathrm{reg}} .
		\]
		Consequently, in the diverging-$n_0$ case as well,
		\[
		\frac{A_{it,1}-A_{it,0}}
		{\{V_{it}^{\mathrm{reg}}\}^{1/2}}
		\Rightarrow N(0,1).
		\]
		In addition, if $f_t=0$, we have $V_{it}^{\mathrm{reg}}\asymp (n_0T_1)^{-1} +(n_0T_0)^{-1}$.
		
		We now compare the remainders with that scale. If $f_t\neq0$, the preceding
		argument yields
		\begin{equation*}
			V_{it}^{\rm reg}\ge c(T_0^{-1}+T_1^{-1})
		\end{equation*}
		for some $c>0$. Therefore \eqref{eq:bnt-small-relative-time} implies
		\begin{equation}
			b_{nT}=o\{(V_{it}^{\rm reg})^{1/2}\}=o(V_{it}^{1/2}).
			\label{eq:bnt-negligible-ft-nonzero}
		\end{equation}
		This is the case in which no relative growth restriction between $n_0$ and
		$n_1$ is needed.

		If $f_t=0$, then the comparison depends on whether $n_0$ is fixed or
		diverging. When $n_0$ is fixed,
		\[
		V_{it}^{\rm reg}\ge c(T_0^{-1}+T_1^{-1})
		\]
		under the covariate nondegeneracy condition, because
		$(n_0T_d)^{-1}\asymp T_d^{-1}$. Hence \eqref{eq:bnt-small-relative-time}
		implies
		\begin{equation}
			b_{nT}=o\{(V_{it}^{\rm reg})^{1/2}\}=o(V_{it}^{1/2}).
			\label{eq:bnt-negligible-ft-zero-fixed}
		\end{equation}

		When $n_0\to\infty$ and $f_t=0$,
		\[
		V_{it}^{\rm reg}\ge c\{(n_0T_0)^{-1}+(n_0T_1)^{-1}\}
		\]
		under the covariate nondegeneracy condition. Combining this lower bound with
		\eqref{eq:bnt-pooled-slope-rate} gives
		\begin{equation}
			b_{nT}=o\{(V_{it}^{\rm reg})^{1/2}\}=o(V_{it}^{1/2}).
			\label{eq:bnt-negligible-ft-zero-diverging}
		\end{equation}
		Thus the same feasible expansion is valid in the $f_t=0$ case: for fixed
		$n_0$ under the main rate conditions \eqref{eq:main-rate-controls}, and for $n_0\to\infty$ under the
		additional restriction \eqref{eq:pooled-slope-rate}.

		Finally, by the factor representation in \eqref{eq:factor-representation},
		\begin{equation}
			\{\lambda_i(1)-\lambda_i(0)\}'(\widehat f_t-f_t)
			=
			\{\lambda_i(1)-\lambda_i(0)\}'Q_\lambda^{-1}
			\frac{1}{n_1}
			\sum_{k=n_0+1}^{n}
			\lambda_k(0)\varepsilon_{kt}(0)
			+O_p\left(\frac1{n_1}+\frac1T\right).
			\label{eq:factor-leading-in-proof}
		\end{equation}
		The leading factor-estimation term is asymptotically normal with variance
		\begin{equation*}
			V_{it}^f
			=
			\{\lambda_i(1)-\lambda_i(0)\}'
			\Var(\widehat f_t-f_t)
			\{\lambda_i(1)-\lambda_i(0)\}.
		\end{equation*}
		Its remainder in \eqref{eq:factor-leading-in-proof} is bounded by $O_p(b_{nT})$. Since the
		factor-estimation component uses the control-unit errors and the
		regression-estimation component uses the treated-unit errors, the maintained
		cross-sectional independence condition gives zero asymptotic covariance.
		Therefore the leading term in \eqref{eq:hattau-tau} is asymptotically normal
		with variance
		\begin{equation*}
			V_{it}=V_{it}^{\rm reg}+V_{it}^f.
		\end{equation*}
		Together with \eqref{eq:bnt-negligible-ft-nonzero} when $f_t\neq0$,
		\eqref{eq:bnt-negligible-ft-zero-fixed} when $f_t=0$ and $n_0$ is fixed,
		and \eqref{eq:bnt-negligible-ft-zero-diverging} when $f_t=0$ and
		$n_0\to\infty$, this yields
		\begin{equation*}
			\frac{\widehat\tau_{it}^{*}-\tau_{it}^{*}}{V_{it}^{1/2}}
			\Rightarrow N(0,1).
		\end{equation*}
		
		It remains only to replace $V_{it}$ by its feasible estimator. By
		Lemma~\ref{lem:variance-consistency} below,
		\begin{equation*}
			\frac{\widehat V_{it}}{V_{it}}=1+o_p(1).
		\end{equation*}
		Slutsky's theorem gives
		\begin{equation*}
			\frac{\widehat\tau_{it}^{*}-\tau_{it}^{*}}{\widehat V_{it}^{1/2}}
			\Rightarrow N(0,1).
		\end{equation*}
		This proves Proposition~\ref{prop:prop-1}.
	\end{proof}

	\begin{lemma}[Consistency of the feasible variance estimator]
		\label{lem:variance-consistency}
		Suppose the assumptions of Proposition~1 hold. Then, for each fixed
		treated unit $i$ and fixed time $t$,
		\[
		\widehat V_{it}
		=
		V_{it}
		+
		o_p(V_{it}),
		\]
		or equivalently,
		\[
		\frac{\widehat V_{it}}{V_{it}}
		=
		1+o_p(1).
		\]
	\end{lemma}

	\begin{proof}[Proof of Lemma \ref{lem:variance-consistency}]
		Fix $d\in\{0,1\}$. Recall the decomposition
		\[
		A_{it,d}=A_{it,d}^{\lambda}+A_{it,d}^{\beta},
		\]
		where
		\[
		A_{it,d}^{\lambda}
		=
		f_t'(F_d'F_d)^{-1}F_d'\varepsilon_{i,d},
		\]
		and
		\[
		A_{it,d}^{\beta}
		=
		h_{it,d}'S_{xx,F,d}^{-1}
		\frac{1}{n_0T_d}
		\sum_{j=1}^{n_0}X_{j,d}'M_{F_d}\varepsilon_{j,d}.
		\]
		Here
		\[
		h_{it,d}
		:=
		X_{it}
		-
		X_{i,d}'F_d(F_d'F_d)^{-1}f_t
		=
		X_{it}
		-
		\frac{ X_{i,d}'F_d} {T_d}
		\left(\frac{F_d'F_d}{T_d}\right)^{-1}
		f_t,
		\]
		and
		\[
		S_{xx,F,d}
		:=
		\frac{1}{n_0T_d}
		\sum_{j=1}^{n_0}X_{j,d}'M_{F_d}X_{j,d}.
		\]
		Let
		\[
		\Omega_{j,d}
		:=
		E(\varepsilon_{j,d}\varepsilon_{j,d}')
		\]
		denote the block covariance matrix of $\varepsilon_{j,d}$. In the
		diagonal heteroskedastic case,
		\[
		\Omega_{j,d}
		=
		\operatorname{diag}
		\left(
		E\varepsilon_{js}(d)^2:s\in\mathcal T_d
		\right).
		\]
		
		The theoretical treated-regression variance for block $d$ is
		\[
		V_{it,d}^{\mathrm{reg}}
		=
		V_{it,d}^{\lambda}
		+
		V_{it,d}^{\beta}
		+
		2V_{it,d}^{\lambda\beta},
		\]
		where
		\[
		V_{it,d}^{\lambda}
		=
		\operatorname{Var}(A_{it,d}^{\lambda}),
		\qquad
		V_{it,d}^{\beta}
		=
		\operatorname{Var}(A_{it,d}^{\beta}),
		\qquad
		V_{it,d}^{\lambda\beta}
		=
		\operatorname{Cov}(A_{it,d}^{\lambda},A_{it,d}^{\beta}).
		\]
		The three components are
		\[
		V_{it,d}^{\lambda}
		=
		\frac{1}{T_d}
		f_t'
		\left(\frac{F_d'F_d}{T_d}\right)^{-1}
		\left(\frac{F_d'\Omega_{i,d}F_d}{T_d}\right)
		\left(\frac{F_d'F_d}{T_d}\right)^{-1}
		f_t,
		\]
		\[
		V_{it,d}^{\beta}
		=
		\frac{1}{n_0T_d}
		h_{it,d}'S_{xx,F,d}^{-1}
		\left[
		\frac{1}{n_0T_d}
		\sum_{j=1}^{n_0}
		X_{j,d}'M_{F_d}\Omega_{j,d}M_{F_d}X_{j,d}
		\right]
		S_{xx,F,d}^{-1}h_{it,d},
		\]
		and
		\[
		V_{it,d}^{\lambda\beta}
		=
		\frac{1}{n_0T_d}
		f_t'
		\left(\frac{F_d'F_d}{T_d}\right)^{-1}
		\left[
		\frac{F_d'\Omega_{i,d}M_{F_d}X_{i,d}}{T_d}
		\right]
		S_{xx,F,d}^{-1}h_{it,d}.
		\]
		
		Now consider the feasible block-specific regression variance estimator.
		Let
		\[
		\widehat S_{xx,d}
		:=
		\frac{1}{n_0T_d}
		\sum_{j=1}^{n_0}
		X_{j,d}'M_{\widehat F_d}X_{j,d},
		\]
		and
		\[
		\widehat h_{it,d}
		:=
		X_{it}
		-
		X_{i,d}'\widehat F_d
		(\widehat F_d'\widehat F_d)^{-1}\widehat f_t
		=
		X_{it}
		-
		X_{i,d}'\widehat F_d
		\left(\frac{\widehat F_d'\widehat F_d}{T_d}\right)^{-1}
		\frac{\widehat f_t}{T_d}.
		\]
		Let $\widehat\Omega_{j,d}$ denote the residual covariance estimator. In
		the diagonal heteroskedastic case,
		\[
		\widehat\Omega_{j,d}
		=
		\operatorname{diag}
		\left(
		\widehat\varepsilon_{js}(d)^2:s\in\mathcal T_d
		\right).
		\]
		For notational simplicity, we focus on the heteroskedastic case. Under serial
		dependence, covariance blocks such as
		$\widehat F_d'\widehat\Omega_{i,d}\widehat F_d/T_d$ are estimated by
		consistent HAC estimators of the corresponding long-run covariance blocks.
		The feasible regression variance for block $d$ is
		\[
		\widehat V_{it,d}^{\mathrm{reg}}
		=
		\widehat V_{it,d}^{\lambda}
		+
		\widehat V_{it,d}^{\beta}
		+
		2\widehat V_{it,d}^{\lambda\beta},
		\]
		where
		\[
		\widehat V_{it,d}^{\lambda}
		=
		\frac{1}{T_d}
		\widehat f_t'
		\left(\frac{\widehat F_d'\widehat F_d}{T_d}\right)^{-1}
		\left(\frac{\widehat F_d'\widehat\Omega_{i,d}\widehat F_d}{T_d}\right)
		\left(\frac{\widehat F_d'\widehat F_d}{T_d}\right)^{-1}
		\widehat f_t,
		\]
		\[
		\widehat V_{it,d}^{\beta}
		=
		\frac{1}{n_0T_d}
		\widehat h_{it,d}'\widehat S_{xx,d}^{-1}
		\left[
		\frac{1}{n_0T_d}
		\sum_{j=1}^{n_0}
		X_{j,d}'M_{\widehat F_d}
		\widehat\Omega_{j,d}
		M_{\widehat F_d}X_{j,d}
		\right]
		\widehat S_{xx,d}^{-1}\widehat h_{it,d},
		\]
		and
		\[
		\widehat V_{it,d}^{\lambda\beta}
		=
		\frac{1}{n_0T_d}
		\widehat f_t'
		\left(\frac{\widehat F_d'\widehat F_d}{T_d}\right)^{-1}
		\left[
		\frac{\widehat F_d'\widehat\Omega_{i,d}
			M_{\widehat F_d}X_{i,d}}{T_d}
		\right]
		\widehat S_{xx,d}^{-1}\widehat h_{it,d}.
		\]
		The estimator $\widehat V_{it,d}^{\mathrm{reg}}$ displayed above is
		identical to the compact sandwich estimator in \eqref{eq:block-reg-var}.
		The expanded form is used only to simplify the consistency proof, avoiding
		a direct high-dimensional sandwich-matrix argument when $n_0$ grows.

		Using Lemma~\ref{lem:blockwise-factor-error-bounds}, the following are $o_p(1)$:
		\[
		\frac{\widehat F_d'\widehat F_d}{T_d}
		-
		\frac{F_d'F_d}{T_d}
		=o_p(1),
		\qquad
		\left(\frac{\widehat F_d'\widehat F_d}{T_d}\right)^{-1}
		-
		\left(\frac{F_d'F_d}{T_d}\right)^{-1}
		=o_p(1),
		\]
		\[
		\frac{\widehat F_d'\widehat\Omega_{i,d}\widehat F_d}{T_d}
		-
		\frac{F_d'\Omega_{i,d}F_d}{T_d}
		=o_p(1),
		\]
		\[
		\widehat S_{xx,d}-S_{xx,F,d}=o_p(1),
		\qquad
		\widehat S_{xx,d}^{-1}-S_{xx,F,d}^{-1}=o_p(1),
		\]
		\[
		\frac{1}{n_0T_d}
		\sum_{j=1}^{n_0}
		X_{j,d}'M_{\widehat F_d}
		\widehat\Omega_{j,d}
		M_{\widehat F_d}X_{j,d}
		-
		\frac{1}{n_0T_d}
		\sum_{j=1}^{n_0}
		X_{j,d}'M_{F_d}
		\Omega_{j,d}
		M_{F_d}X_{j,d}
		=o_p(1),
		\]
		\[
		\frac{\widehat F_d'\widehat\Omega_{i,d}
			M_{\widehat F_d}X_{i,d}}{T_d}
		-
		\frac{F_d'\Omega_{i,d}M_{F_d}X_{i,d}}{T_d}
		=o_p(1).
		\]
		Also $\widehat h_{it,d}-h_{it,d}=o_p(1)$.
		
		We now separate the two possible variance scales.
		
		First suppose $f_t\neq0$. Then, by the maintained nondegeneracy
		condition,
		\[
		V_{it,d}^{\lambda}\asymp T_d^{-1},
		\qquad
		V_{it,d}^{\mathrm{reg}}\asymp T_d^{-1}.
		\]
		Using the preceding normalized convergences and $\widehat f_t-f_t=o_p(1)$,
		\[
		\widehat V_{it,d}^{\lambda}-V_{it,d}^{\lambda}
		=
		T_d^{-1}o_p(1)
		=
		o_p(T_d^{-1}).
		\]
		Similarly,
		\[
		\widehat V_{it,d}^{\beta}-V_{it,d}^{\beta}
		=
		(n_0T_d)^{-1}o_p(1)
		=
		o_p(T_d^{-1}),
		\]
		and
		\[
		\widehat V_{it,d}^{\lambda\beta}-V_{it,d}^{\lambda\beta}
		=
		(n_0T_d)^{-1}o_p(1)
		=
		o_p(T_d^{-1}).
		\]
		Therefore
		\[
		\widehat V_{it,d}^{\mathrm{reg}}
		-
		V_{it,d}^{\mathrm{reg}}
		=
		o_p(T_d^{-1})
		=
		o_p(V_{it,d}^{\mathrm{reg}}).
		\]
		
		Next suppose $f_t=0$. Then $V_{it,d}^{\lambda}=0$ and
		$V_{it,d}^{\lambda\beta}=0$. The leading regression variance is the
		pooled-slope component, and by the maintained nondegeneracy condition,
		\[
		V_{it,d}^{\mathrm{reg}}
		=
		V_{it,d}^{\beta}
		\asymp
		(n_0T_d)^{-1}.
		\]
		For the slope component, the normalized sandwich convergence gives
		\[
		\widehat V_{it,d}^{\beta}-V_{it,d}^{\beta}
		=
		(n_0T_d)^{-1}o_p(1)
		=
		o_p((n_0T_d)^{-1}).
		\]
		It remains to check that the feasible loading and covariance pieces,
		which are zero in population when $f_t=0$, are negligible at the
		$(n_0T_d)^{-1}$ scale. Since $f_t=0$,
		\[
		\widehat f_t=O_p(n_1^{-1/2}+T^{-1}).
		\]
		Hence
		\[
		\widehat V_{it,d}^{\lambda}
		=
		T_d^{-1}O_p(\|\widehat f_t\|^2)
		=
		O_p\left(\frac{1}{T_dn_1}+\frac{1}{T_dT^2}\right).
		\]
		Therefore
		\[
		\frac{\widehat V_{it,d}^{\lambda}}
		{(n_0T_d)^{-1}}
		=
		O_p\left(\frac{n_0}{n_1}+\frac{n_0}{T^2}\right)
		=
		o_p(1).
		\]
		The last equality holds because, if $n_0$ is fixed, then
		$n_0/n_1\to0$ and $n_0/T^2\to0$ follow from $n_1\to\infty$ and
		$T\to\infty$; while if $n_0\to\infty$ in the $f_t=0$ case, it follows from
		\eqref{eq:pooled-slope-rate}.
		Thus
		\[
		\widehat V_{it,d}^{\lambda}
		=
		o_p((n_0T_d)^{-1}).
		\]
		Similarly,
		\[
		\widehat V_{it,d}^{\lambda\beta}
		=
		(n_0T_d)^{-1}O_p(\|\widehat f_t\|)
		=
		o_p((n_0T_d)^{-1}).
		\]
		Combining the three pieces gives
		\[
		\widehat V_{it,d}^{\mathrm{reg}}
		-
		V_{it,d}^{\mathrm{reg}}
		=
		o_p((n_0T_d)^{-1})
		=
		o_p(V_{it,d}^{\mathrm{reg}}).
		\]
		
		Thus, in either case,
		\[
		\widehat V_{it,d}^{\mathrm{reg}}
		=
		V_{it,d}^{\mathrm{reg}}
		+
		o_p(V_{it,d}^{\mathrm{reg}}),
		\qquad d=0,1.
		\]
		Summing over $d=0,1$ and using the maintained nondegeneracy of
		$V_{it}^{\mathrm{reg}}=V_{it,1}^{\mathrm{reg}}+V_{it,0}^{\mathrm{reg}}$,
		we obtain
		\[
		\widehat V_{it}^{\mathrm{reg}}
		=
		V_{it}^{\mathrm{reg}}
		+
		o_p(V_{it}^{\mathrm{reg}}).
		\]
		
		It remains to handle the factor-estimation variance component.

		By consistency of the control-unit loading and residual covariance
		estimators,
		\[
		\widehat Q_\lambda-Q_\lambda=o_p(1),
		\qquad
		\widehat Q_\lambda^{-1}-Q_\lambda^{-1}=o_p(1),
		\qquad
		\widehat S_t-S_t=o_p(1).
		\]
		Therefore
		\[
		\widehat Q_\lambda^{-1}
		\widehat S_t
		\widehat Q_\lambda^{-1}
		-
		Q_\lambda^{-1}S_tQ_\lambda^{-1}
		=
		o_p(1),
		\]
		and hence
		\[
		\widehat{\operatorname{Var}}(\widehat f_t)
		-
		\frac{1}{n_1}Q_\lambda^{-1}S_tQ_\lambda^{-1}
		=
		o_p(n_1^{-1}).
		\]
		
		We now split according to whether $\Delta\lambda_i$ is zero.
		
		If $\Delta\lambda_i\neq0$, then by nondegeneracy,
		\[
		V_{it}^{f}
		=
		\frac{1}{n_1}
		\Delta\lambda_i'
		Q_\lambda^{-1}S_tQ_\lambda^{-1}
		\Delta\lambda_i
		\asymp n_1^{-1}.
		\]
		Moreover, by consistency of the treated loading estimators,
		\[
		\widehat{\Delta\lambda}_i-\Delta\lambda_i=o_p(1).
		\]
		Therefore
		\[
		\widehat V_{it}^{f}-V_{it}^{f}
		=
		\frac{1}{n_1}
		\left[
		\widehat{\Delta\lambda}_i'
		\widehat Q_\lambda^{-1}\widehat S_t\widehat Q_\lambda^{-1}
		\widehat{\Delta\lambda}_i
		-
		\Delta\lambda_i'
		Q_\lambda^{-1}S_tQ_\lambda^{-1}
		\Delta\lambda_i
		\right]
		=
		o_p(n_1^{-1}).
		\]
		Since $V_{it}^{f}\asymp n_1^{-1}$,
		\[
		\widehat V_{it}^{f}
		=
		V_{it}^{f}
		+
		o_p(V_{it}^{f}).
		\]
		
		If $\Delta\lambda_i=0$, then $V_{it}^{f}=0$, so relative consistency
		with respect to $V_{it}^{f}$ is not meaningful. We instead show
		\[
		\widehat V_{it}^{f}
		=
		o_p(V_{it}^{\mathrm{reg}}).
		\]
		Since $\Delta\lambda_i=0$,
		\[
		\widehat{\Delta\lambda}_i
		=
		\{\widehat\lambda_i(1)-\lambda_i(1)\}
		-
		\{\widehat\lambda_i(0)-\lambda_i(0)\}.
		\]
		By the treated-loading expansion, and because the first-stage remainder
		is dominated by the usual $T_d^{-1/2}$ loading-estimation rate in this
		comparison,
		\[
		\|\widehat{\Delta\lambda}_i\|^2
		=
		O_p\left(\frac{1}{T_0}+\frac{1}{T_1}\right).
		\]
		Also
		\[
		\widehat Q_\lambda^{-1}\widehat S_t\widehat Q_\lambda^{-1}
		=
		O_p(1).
		\]
		Therefore
		\[
		\widehat V_{it}^{f}
		=
		O_p\left[
		\frac{1}{n_1}
		\left(
		\frac{1}{T_0}+\frac{1}{T_1}
		\right)
		\right].
		\]
		
		If $f_t\neq0$, then
		\[
		V_{it}^{\mathrm{reg}}
		\asymp
		\frac{1}{T_0}+\frac{1}{T_1},
		\]
		and hence, since $n_1\to\infty$,
		\[
		\widehat V_{it}^{f}
		=
		o_p(V_{it}^{\mathrm{reg}}).
		\]
		If $f_t=0$, then
		\[
		V_{it}^{\mathrm{reg}}
		\asymp
		\frac{1}{n_0T_0}+\frac{1}{n_0T_1}.
		\]
		Thus
		\[
		\frac{\widehat V_{it}^{f}}{V_{it}^{\mathrm{reg}}}
		=
		O_p\left(\frac{n_0}{n_1}\right)
		=
		o_p(1).
		\]
		The above conclusion holds automatically when $n_0$ is fixed, since
		$n_0/n_1\to0$. When $n_0\to\infty$ and $f_t=0$, it holds under the additional
		rate restriction \eqref{eq:pooled-slope-rate}.
		Hence again
		\[
		\widehat V_{it}^{f}
		=
		o_p(V_{it}^{\mathrm{reg}}).
		\]
		
		Finally, combine the regression and factor components. If
		$\Delta\lambda_i\neq0$, then
		\[
		\widehat V_{it}-V_{it}
		=
		\left(
		\widehat V_{it}^{\mathrm{reg}}-V_{it}^{\mathrm{reg}}
		\right)
		+
		\left(
		\widehat V_{it}^{f}-V_{it}^{f}
		\right)
		=
		o_p(V_{it}^{\mathrm{reg}})
		+
		o_p(V_{it}^{f})
		=
		o_p(V_{it}).
		\]
		If $\Delta\lambda_i=0$, then $V_{it}^{f}=0$ and
		$V_{it}=V_{it}^{\mathrm{reg}}$, while
		\[
		\widehat V_{it}-V_{it}
		=
		\left(
		\widehat V_{it}^{\mathrm{reg}}-V_{it}^{\mathrm{reg}}
		\right)
		+
		\widehat V_{it}^{f}
		=
		o_p(V_{it}^{\mathrm{reg}})
		+
		o_p(V_{it}^{\mathrm{reg}})
		=
		o_p(V_{it}).
		\]
		Therefore, in all cases,
		\[
		\widehat V_{it}
		=
		V_{it}
		+
		o_p(V_{it}),
		\]
		or equivalently,
		\[
		\frac{\widehat V_{it}}{V_{it}}
		=
		1+o_p(1).
		\]
		This proves the lemma.
	\end{proof}

	\section{Proof of Proposition~\ref{prop:prop-2}}
	
	In addition to Assumptions A.1--A.5, we impose one additional assumption.
	Recall we also assume $n_0/n \to c\in (0,1)$ and $T_0/T \to b\in (0,1)$.

	\begin{assumption}[Potential factors] \label{ass:potential-factors}
		For $d=0,1$, the potential factors and factor loadings satisfy
		$E\|f_t(d)\|^4\leq M<\infty$ and
		$E\|\lambda_i(d)\|^4\leq M<\infty$, for all $t$ and $i$.
		The sample second moments
		$T_0^{-1}\sum_{s=1}^{T_0}f_s(0)f_s(0)'$,
		$T_1^{-1}\sum_{s=T_0+1}^{T}f_s(1)f_s(1)'$,
		$n_0^{-1}\sum_{k=1}^{n_0}\lambda_k(1)\lambda_k(1)'$, and
		$n_1^{-1}\sum_{k=n_0+1}^{n}\lambda_k(0)\lambda_k(0)'$
		converge in probability to positive definite matrices.
		For the fixed treated unit $i$ and fixed post-treatment date $t$ considered in
		Proposition~\ref{prop:prop-2}, the normalized scores
		$n_0^{-1/2}\sum_{k=1}^{n_0}\lambda_k(1)\varepsilon_{kt}(1)$,
		$T_1^{-1/2}\sum_{s=T_0+1}^{T}f_s(1)\varepsilon_{is}(1)$,
		$T_0^{-1/2}\sum_{s=1}^{T_0}f_s(0)\varepsilon_{is}(0)$, and
		$n_1^{-1/2}\sum_{k=n_0+1}^{n}\lambda_k(0)\varepsilon_{kt}(0)$
		satisfy a joint central limit theorem with a nondegenerate covariance matrix.
	\end{assumption}

	\textbf{Proof of Proposition~\ref{prop:prop-2}}. Let
	$T_1=T-T_0$, $n_1=n-n_0$, and define
	$C_{it}(d)=\lambda_i(d)'f_t(d)$, $d=0,1$.  Also let
	\[
	\delta_{nT}=
	\min\{\sqrt{n_0},\sqrt{n_1},\sqrt{T_0},\sqrt{T_1}\}.
	\]
	We first derive the expansion for the post-treatment treated common component.
	Since $\hat C_{it}(1)$ is obtained by principal components using
	$\{Y_{ks}:k\le n_0,\;s>T_0\}$, Bai's common-component expansion gives,
	for fixed $i\le n_0$ and $t>T_0$,
	\begin{align*}
		\hat C_{it}(1)-C_{it}(1)
		&=\lambda_i(1)'
		\left(\frac1{n_0}\sum_{k=1}^{n_0}\lambda_k(1)\lambda_k(1)'\right)^{-1}
		\left(\frac1{n_0}\sum_{k=1}^{n_0}\lambda_k(1)\varepsilon_{kt}(1)\right) \\
		& +f_t(1)'
		\left(\frac1{T_1}\sum_{s=T_0+1}^{T}f_s(1)f_s(1)'\right)^{-1}
		\left(\frac1{T_1}\sum_{s=T_0+1}^{T}f_s(1)\varepsilon_{is}(1)\right)
		+O_p(\delta_{nT}^{-2}).
	\end{align*}
	To make the normalizations explicit, define the sample second-moment matrices
	\[
	Q_{\lambda,1}:=\frac1{n_0}\sum_{k=1}^{n_0}\lambda_k(1)\lambda_k(1)',
	\qquad
	Q_{f,1}:=\frac1{T_1}\sum_{s=T_0+1}^{T}f_s(1)f_s(1)',
	\]
	and the normalized score averages
	\[
	\bar S_{A,t}:=\frac1{n_0}\sum_{k=1}^{n_0}\lambda_k(1)\varepsilon_{kt}(1),
	\qquad
	\bar S_{B,i}:=\frac1{T_1}\sum_{s=T_0+1}^{T}f_s(1)\varepsilon_{is}(1).
	\]
	Let
	\[
	M_A:=\lambda_i(1)'Q_{\lambda,1}^{-1},
	\qquad
	M_B:=f_t(1)'Q_{f,1}^{-1}.
	\]
	Then
	\[
	\hat C_{it}(1)-C_{it}(1)
	= A+B+O_p(\delta_{nT}^{-2}),
	\qquad
	A:=M_A\bar S_{A,t},\quad B:=M_B\bar S_{B,i}.
	\]
	Equivalently,
	\[
	A=n_0^{-1/2}M_A
	\left(\frac1{\sqrt{n_0}}\sum_{k=1}^{n_0}\lambda_k(1)\varepsilon_{kt}(1)\right),
	\qquad
	B=T_1^{-1/2}M_B
	\left(\frac1{\sqrt{T_1}}\sum_{s=T_0+1}^{T}f_s(1)\varepsilon_{is}(1)\right),
	\]
	so Assumption~\ref{ass:potential-factors} applies directly to the normalized sums.
	
	Next consider the untreated common component for the same treated unit.
	The factor estimate $\hat f_t(0)$ is obtained from the control units.
	The product $\lambda_i(d)'f_t(d)$ is invariant to rotations, and under the
	normalization conditions used for the principal-components estimator we can
	write the expansion with the rotation matrix absorbed into the notation.
	Under $\sqrt n/T\to0$,
	\[
	\sqrt{n_1}\{\hat f_t(0)-f_t(0)\}
	=
	\left(\frac1{n_1}\sum_{k=n_0+1}^{n}\lambda_k(0)\lambda_k(0)'\right)^{-1}
	\left(\frac1{\sqrt{n_1}}\sum_{k=n_0+1}^{n}\lambda_k(0)\varepsilon_{kt}(0)\right)
	+o_p(1).
	\]
	The pre-treatment loading estimator for the treated unit is obtained by
	regressing $Y_{is}$ on $\hat f_s(0)$ over $s=1,\ldots,T_0$. Its first-order
	expansion is
	\[
	\sqrt{T_0}\{\hat\lambda_i(0)-\lambda_i(0)\}
	=
	\left(\frac1{T_0}\sum_{s=1}^{T_0}f_s(0)f_s(0)'\right)^{-1}
	\left(\frac1{\sqrt{T_0}}\sum_{s=1}^{T_0}f_s(0)\varepsilon_{is}(0)\right)
	+o_p(1).
	\]
	Define
	\[
	Q_{f,0}:=\frac1{T_0}\sum_{s=1}^{T_0}f_s(0)f_s(0)',
	\qquad
	Q_{\lambda,0}:=\frac1{n_1}\sum_{k=n_0+1}^{n}\lambda_k(0)\lambda_k(0)',
	\]
	\[
	\bar S_{C,i}:=\frac1{T_0}\sum_{s=1}^{T_0}f_s(0)\varepsilon_{is}(0),
	\qquad
	\bar S_{D,t}:=\frac1{n_1}\sum_{k=n_0+1}^{n}\lambda_k(0)\varepsilon_{kt}(0),
	\]
	and
	\[
	M_C:=f_t(0)'Q_{f,0}^{-1},
	\qquad
	M_D:=\lambda_i(0)'Q_{\lambda,0}^{-1}.
	\]
	Then, for fixed $i\le n_0$ and $t>T_0$,
	\begin{align*}
		\hat C_{it}(0)-C_{it}(0)
		&=\{\hat\lambda_i(0)-\lambda_i(0)\}'\hat f_t(0)
		+\lambda_i(0)'\{\hat f_t(0)-f_t(0)\} \\
		&=\{\hat\lambda_i(0)-\lambda_i(0)\}'f_t(0)
		+\lambda_i(0)'\{\hat f_t(0)-f_t(0)\} \\
		&\quad
		+\{\hat\lambda_i(0)-\lambda_i(0)\}'\{\hat f_t(0)-f_t(0)\} \\
		&=M_C\bar S_{C,i}+M_D\bar S_{D,t}+O_p(\delta_{nT}^{-2}) \\
		&\equiv C+D+O_p(\delta_{nT}^{-2}).
	\end{align*}
	Here $C:=M_C\bar S_{C,i}$ and $D:=M_D\bar S_{D,t}$. Hence
	\begin{align*}
		\hat\tau_{it}^*-\tau_{it}^*
		&=\{\hat C_{it}(1)-C_{it}(1)\}-\{\hat C_{it}(0)-C_{it}(0)\} \\
		&=A+B-C-D+O_p(\delta_{nT}^{-2}).
	\end{align*}
	By the joint central limit theorem and the cross-sectional uncorrelatedness
	and weak serial dependence assumptions, the Gaussian limits of $A$, $B$, $C$,
	and $D$ are mutually uncorrelated, hence asymptotically independent. Define
	the leading variance by
	\[
	V_{it}
	:=
	\operatorname{Var}(A)+\operatorname{Var}(B)
	+\operatorname{Var}(C)+\operatorname{Var}(D).
	\]
	Then
	\[
	\operatorname{Var}(\hat\tau_{it}^*-\tau_{it}^*)
	=
	V_{it}+o(\delta_{nT}^{-2}).
	\]
	Under the maintained nondegeneracy
	conditions, $V_{it}$ is of order
	$n_0^{-1}+T_1^{-1}+T_0^{-1}+n_1^{-1}$, while the remainder
	$O_p(\delta_{nT}^{-2})$ is negligible relative to $V_{it}^{1/2}$. Therefore
	\[
	V_{it}^{-1/2}(\hat\tau_{it}^*-\tau_{it}^*)\overset{d}{\to}N(0,1).
	\]
	It remains to describe the feasible variance estimator.
	Define
	\[
	\widehat Q_{\lambda,1}
	:=\frac1{n_0}\sum_{k=1}^{n_0}\hat\lambda_k(1)\hat\lambda_k(1)',
	\qquad
	\widehat Q_{f,1}
	:=\frac1{T_1}\sum_{s=T_0+1}^{T}\hat f_s(1)\hat f_s(1)',
	\]
	\[
	\widehat Q_{f,0}
	:=\frac1{T_0}\sum_{s=1}^{T_0}\hat f_s(0)\hat f_s(0)',
	\qquad
	\widehat Q_{\lambda,0}
	:=\frac1{n_1}\sum_{k=n_0+1}^{n}\hat\lambda_k(0)\hat\lambda_k(0)'.
	\]
	Let
	\[
	\hat M_A:=\hat\lambda_i(1)'\widehat Q_{\lambda,1}^{-1},
	\quad
	\hat M_B:=\hat f_t(1)'\widehat Q_{f,1}^{-1},
	\quad
	\hat M_C:=\hat f_t(0)'\widehat Q_{f,0}^{-1},
	\quad
	\hat M_D:=\hat\lambda_i(0)'\widehat Q_{\lambda,0}^{-1}.
	\]
	With residuals from the corresponding principal-components or regression
	steps, define the normalized middle matrices
	\[
	\widehat\Omega_{A,t}
	:=\frac1{n_0}\sum_{k=1}^{n_0}
	\hat\lambda_k(1)\hat\lambda_k(1)'\hat\varepsilon_{kt}(1)^2,
	\qquad
	\widehat\Omega_{B,i}
	:=\frac1{T_1}\sum_{s=T_0+1}^{T}
	\hat f_s(1)\hat f_s(1)'\hat\varepsilon_{is}(1)^2,
	\]
	\[
	\widehat\Omega_{C,i}
	:=\frac1{T_0}\sum_{s=1}^{T_0}
	\hat f_s(0)\hat f_s(0)'\hat\varepsilon_{is}(0)^2,
	\qquad
	\widehat\Omega_{D,t}
	:=\frac1{n_1}\sum_{k=n_0+1}^{n}
	\hat\lambda_k(0)\hat\lambda_k(0)'\hat\varepsilon_{kt}(0)^2.
	\]
	The four leading variance components are estimated by
	\[
	\begin{aligned}
		\widehat{\operatorname{Var}}(A)
		&=\frac1{n_0}\hat M_A\widehat\Omega_{A,t}\hat M_A',
		&
		\widehat{\operatorname{Var}}(B)
		&=\frac1{T_1}\hat M_B\widehat\Omega_{B,i}\hat M_B', \\
		\widehat{\operatorname{Var}}(C)
		&=\frac1{T_0}\hat M_C\widehat\Omega_{C,i}\hat M_C',
		&
		\widehat{\operatorname{Var}}(D)
		&=\frac1{n_1}\hat M_D\widehat\Omega_{D,t}\hat M_D'.
	\end{aligned}
	\]
	Thus
	\[
	\widehat V_{it}
	=
	\widehat{\operatorname{Var}}(A)+\widehat{\operatorname{Var}}(B)
	+\widehat{\operatorname{Var}}(C)+\widehat{\operatorname{Var}}(D)
	\]
	is the feasible first-order variance estimator. Consistency of the estimated
	factor and loading spaces, together with the law of large numbers for the
	normalized middle matrices, gives $\widehat V_{it}/V_{it}\overset{p}{\to}1$.
	Consequently,
	\[
	\widehat V_{it}^{-1/2}(\hat\tau_{it}^*-\tau_{it}^*)\overset{d}{\to}N(0,1).
	\]
	
	In practice (see Section \ref{sec:mcmc}), we implement the following
	finite-sample degree-of-freedom adjustment. In the normalized notation above,
	replace the four variance components by
	\[
	\begin{aligned}
		\widehat{\operatorname{Var}}(A)
		&=\frac{1}{n_0-2r}\hat M_A\widehat\Omega_{A,t}\hat M_A',
		&
		\widehat{\operatorname{Var}}(B)
		&=\frac{1}{T_1-2r}\hat M_B\widehat\Omega_{B,i}\hat M_B', \\
		\widehat{\operatorname{Var}}(C)
		&=\frac{1}{T_0-2r}\hat M_C\widehat\Omega_{C,i}\hat M_C',
		&
		\widehat{\operatorname{Var}}(D)
		&=\frac{1}{n_1-2r}\hat M_D\widehat\Omega_{D,t}\hat M_D'.
	\end{aligned}
	\]
	The sum of these four adjusted expressions is the finite-sample version of
	$\widehat V_{it}$. 	\textbf{ Q.E.D.}
	
\end{appendix}

\bibliographystyle{econsoc}
\bibliography{causalfactor_bib}

\end{document}